\documentclass[12pt,reqno]{amsart}
\usepackage[latin9]{inputenc}
\usepackage{geometry}
\usepackage{float}
\usepackage{amsbsy}
\usepackage{amstext}
\usepackage{amsthm}
\usepackage{amssymb}
\usepackage{graphicx}
\PassOptionsToPackage{normalem}{ulem}
\usepackage{ulem}
\usepackage[hyperfootnotes=false,unicode=true,
 bookmarks=false,
 breaklinks=false,pdfborder={0 0 1},backref=section,colorlinks=false]
 {hyperref}

\makeatletter
\numberwithin{equation}{section}
\numberwithin{figure}{section}
\theoremstyle{plain}
\newtheorem{thm}{\protect\theoremname}[section]
\theoremstyle{definition}
\newtheorem{defn}[thm]{\protect\definitionname}
\theoremstyle{remark}
\newtheorem{rem}[thm]{\protect\remarkname}
\theoremstyle{plain}
\newtheorem{cor}[thm]{\protect\corollaryname}
\theoremstyle{plain}
\newtheorem{conjecture}[thm]{\protect\conjecturename}
\theoremstyle{plain}
\newtheorem{prop}[thm]{\protect\propositionname}
\theoremstyle{plain}
\newtheorem{lem}[thm]{\protect\lemmaname}
\theoremstyle{remark}
\newtheorem*{rem*}{\protect\remarkname}

\@ifundefined{date}{}{\date{}}

\usepackage{amscd}
\usepackage{comment}
\usepackage{mathrsfs}
\usepackage{setspace}
\usepackage{color}

\definecolor{green}{RGB}{0, 140, 0}

\renewcommand\[{\begin{equation}}
\renewcommand\]{\end{equation}}

\providecommand{\conjecturename}{Conjecture}
\providecommand{\corollaryname}{Corollary}
\providecommand{\definitionname}{Definition}
\providecommand{\lemmaname}{Lemma}
\providecommand{\propositionname}{Proposition}
\providecommand{\remarkname}{Remark}
\providecommand{\theoremname}{Theorem}

\makeatother

\providecommand{\conjecturename}{Conjecture}
\providecommand{\corollaryname}{Corollary}
\providecommand{\definitionname}{Definition}
\providecommand{\lemmaname}{Lemma}
\providecommand{\propositionname}{Proposition}
\providecommand{\remarkname}{Remark}
\providecommand{\theoremname}{Theorem}

\begin{document}
\global\long\def\yv{\vec{y}}%
\global\long\def\zv{\vec{z}}%
\global\long\def\at{\left.\right|_{\Omega}}%

\global\long\def\Id{\mathbb{I}}%

\global\long\def\rme{\mathrm{e}}%

\global\long\def\rmi{\mathrm{i}}%

\global\long\def\rmd{\mathrm{d}}%

\global\long\def\A{\mathcal{A}}%

\global\long\def\J{\mathcal{J}}%

\global\long\def\F{\mathcal{F}}%

\global\long\def\G{\mathcal{G}}%

\global\long\def\Gsing{\mathcal{G}_{\mathrm{singletons}}}%

\global\long\def\M{\mathcal{M}}%

\global\long\def\L{\mathcal{L}}%

\global\long\def\D{\mathcal{D}}%

\global\long\def\U{\mathcal{U}}%

\global\long\def\I{\mathcal{I}}%

\global\long\def\P{\mathrm{P}}%

\global\long\def\Rv{\rho^{(v)}}%

\global\long\def\Nv{N^{(v)}}%

\global\long\def\Rvb{\boldsymbol{\rho^{(v)}}}%

\global\long\def\Nvb{\boldsymbol{N^{(v)}}}%

\global\long\def\gve{\boldsymbol{g_{e}^{(v)}}}%

\global\long\def\d{\partial}%

\global\long\def\dg{\partial\Gamma}%

\global\long\def\do{\partial\Omega}%

\global\long\def\ndv{\Omega_{n}^{(v)}}%

\global\long\def\rve{\mathcal{R}_{v,e}}%

\global\long\def\E{\mathcal{E}}%

\global\long\def\V{\mathcal{V}}%

\global\long\def\Vint{\mathcal{V}\setminus\partial\Gamma}%

\global\long\def\Vin{V_{\textrm{in}}}%

\global\long\def\la{\lambda}%

\global\long\def\hess{\mathrm{Hess}}%

\global\long\def\H{H}%

\global\long\def\Z{\mathbb{Z}}%

\global\long\def\R{\mathbb{R}}%

\global\long\def\C{\mathbb{C}}%

\global\long\def\N{\mathbb{N}}%

\global\long\def\Q{\mathbb{Q}}%

\global\long\def\P{\mathbb{P}}%

\global\long\def\lap{\Delta}%

\global\long\def\na{\nabla}%

\global\long\def\opcl#1{\left(#1\right]}%

\global\long\def\clop#1{\left[#1\right)}%

\global\long\def\bs#1{\boldsymbol{#1}}%

\global\long\def\deg#1{\mathrm{deg}(#1)}%

\global\long\def\torus{\mathbb{T}^{E}}%

\global\long\def\BGm{\mu_{\vec{l}}}%

\global\long\def\msing{\Sigma^{\mathrm{sing}}}%

\global\long\def\mreg{\Sigma^{\mathrm{reg}}}%

\global\long\def\mgen{\Sigma^{\mathrm{gen}}}%

\global\long\def\mloop{\Sigma^{\mathrm{loops}}}%

\global\long\def\mgensing{\Sigma_{\mathrm{singletons}}^{\mathrm{gen}}}%

\global\long\def\mgencont{\Sigma_{\mathrm{cont}}^{\mathrm{gen}}}%

\global\long\def\mgendisc{\Sigma_{\mathrm{disc}}^{\mathrm{gen}}}%

\global\long\def\ts{t_{S}}%

\global\long\def\lv{\vec{l}}%

\global\long\def\av{\vec{a}}%

\global\long\def\Lv{\vec{L}}%

\global\long\def\tv{\vec{\theta}_{v}}%

\global\long\def\kv{\vec{\kappa}}%

\global\long\def\xv{\vec{x}}%

\global\long\def\dL{d_{\vec{L}}}%

\global\long\def\gl{\Gamma_{\lv}}%

\global\long\def\gk{\Gamma_{\kv}}%

\global\long\def\sgn{\mathrm{sgn}}%

\global\long\def\modp#1{ \left[#1\right] }%

\global\long\def\norm#1{ \left\Vert #1\right\Vert _{1} }%

\global\long\def\imag{\mathrm{Image}}%

\global\long\def\flow{\varphi_{\lv}}%

\global\long\def\bs#1{\boldsymbol{#1}}%

\global\long\def\set#1#2{\left\{  #1\,\,:\,\,#2\right\}  }%

\global\long\def\undercom#1#2{\underset{_{#2}}{\underbrace{#1}}}%

\title{Neumann Domains on Quantum Graphs }
\author{Lior Alon and Ram Band}
\address{$^{1}${\small{}Lior Alon, School of Mathematics, Institute for Advanced
Study, Princeton, New Jersey 08540, USA.}}
\address{$^{2}${\small{}Ram Band, Department of Mathematics, Technion--Israel
Institute of Technology, Haifa 32000, Israel.}}
\begin{abstract}
The Neumann points of an eigenfunction $f$ on a quantum (metric)
graph are the interior zeros of $f'$. The Neumann domains of $f$
are the sub-graphs bounded by the Neumann points. Neumann points and
Neumann domains are the counterparts of the well-studied nodal points
and nodal domains.

We prove some foundational results in this field: bounds on the number
of Neumann points and properties of the probability distribution of
this number. Two basic properties of Neumann domains are presented:
the wavelength capacity and the spectral position. We state and prove
bounds on those as well as key features of their probability distributions.

To rigorously investigate those probabilities, we establish the notion
of random variables for quantum graphs. In particular, we provide
conditions for considering spectral functions of quantum graphs as
random variables with respect to the natural density on $\N$.
\end{abstract}

\maketitle

\section{Introduction \label{sec: Introduction}}

Nodal domains of Laplacian eigenfunctions form a central research
area within spectral geometry. Historically, the first rigorous results
in the field are by Sturm \cite{Stu_jmpa36}, Courant \cite{Cou_ngwg23}
and Pleijel \cite{Pleijel56}. Many works appeared since then, treating
nodal domains on manifolds, metric graphs and discrete graphs. The
nodal domain study on quantum (metric) graphs is a relatively modern
topic, starting\footnote{Noting that the the work of Sturm \cite{Stu_jmpa36} on the interval
may also be considered as a result on the simplest metric graph.} with \cite{GnuSmiWeb_wrm04} which provides an analogue of Courant's
bound for graphs and initial results on the statistics of the nodal
count. Further results came afterwards, including proofs of bounds
on the nodal count \cite{Schapotschnikow06,AlO_viniti92,PokPryObe_mz96,Ber_cmp08,BanBerWey_jmp15},
study of nodal statistics \cite{AloBanBer_cmp18,AloBanBer_conj},
solutions of nodal inverse problems \cite{Ban_ptrsa14,BanShaSmi_jpa06,JuuJoy_jphys18}
and variational characterizations of the nodal count \cite{BerWey_ptrsa14,BanBerRazSmi_cmp12}.

The current paper is devoted to a closely related notion, called Neumann
domains. On a metric graph, nodal domains are sub-graphs bounded by
the zeros of the eigenfunction. Similarly, Neumann domains are the
sub-graphs bounded by the zeros of the eigenfunction's derivative.
To the best of our knowledge, this is the first work on Neumann domains
on graphs\footnote{It is worthwhile to mention the interesting recent work on the related
topic of Neumann partitions on graphs \cite{KenKurLenMug_arxiv20,HofKenMugPlu20arxiv}. }. Even on manifolds, Neumann domains are a very recent topic of research
within spectral theory and is currently mentioned only in \cite{Zel_sdg13,McDFul_ptrs13,BanFaj_ahp16,BanEggTay_jga20,AloBanBerEgg_lms20,BanCoxEgg_arxiv20}.
Partial results of the current paper were already announced in \cite{AloBanBerEgg_lms20}
which reviews the Neumann domain research on manifolds and on graphs.

The paper is structured as follows. The rest of this section provides
the required preliminary definitions. Our main results are stated
in Section \ref{sec: main_results}. The proofs are then split between
a few sections: Section \ref{sec: proofs_of_bounds} provides the
proofs for the bounds and Sections \ref{sec: proofs_Neumann_count_and_spectral_position}
and \ref{sec: proofs_wavelength_capacity} contain the proofs of the
probabilistic statements. In between, there are two sections which
present and develop the tools needed for proving the probabilistic
statements. Section \ref{sec: tools_and_methods} presents existing
methods from the literature, whereas in Section \ref{sec: functions_on_secular_manifold}
we state and prove the additional required lemmas. In particular,
section \ref{sec: functions_on_secular_manifold} provides tools which
ought to be useful for anyone considering random variables in the
context of quantum graphs. The paper is concluded with a summary section
(Section \ref{sec:Discussion}). Appendix \ref{sec: Appendix-examples}
describes in detail calculations of Neumann and nodal counts of some
particular graph families. Appendices \ref{sec: appendix=00005C-proof-of-two-lemmas}
and \ref{sec:Appendix nodal domains} contain proofs to some lemmas.

\subsection{Basic graph definitions and notations\label{subsec: metric_graphs_introduction}}

Throughout this paper, the graphs we consider are connected and have
finite number of edges and vertices. We denote a graph by $\Gamma$
and denote by $\V$ and $\E$ its sets of vertices and edges, correspondingly.
We will always assume that these sets are non-empty and denote their
cardinalities by $V:=\left|\V\right|>0$ and $E:=\left|\E\right|>0$.
The graph is not necessarily simple. Namely, two vertices may be connected
by more than one edge and it is also possible for an edge to connect
a vertex to itself. An edge connecting a vertex to itself is called
a \emph{loop}.

Given a vertex $v\in\V$ we denote the multi-set of edges connected
to $v$ by $\E_{v}$. We note that every loop connected to $v$ will
appear twice in $\E_{v}$. The \emph{degree} of a vertex is denoted
by $\deg v:=\left|\E_{v}\right|$. The \emph{boundary of a graph}
is defined to be $\dg:=\left\{ v\in\V\thinspace|\,\deg v=1\right\} .$
The rest of the vertices, $\Vint$, are called \emph{interior vertices}.
We denote the first Betti number of a graph by
\begin{equation}
\beta:=E-V+1.\label{eq:Betti_number}
\end{equation}
Formally, $\beta$ is the rank of the graph's first homology. Intuitively,
$\beta$ is the number of ``independent'' simple closed paths in
the graph. A graph with $\beta=0$ is called a \emph{tree graph.}

We may further identify each edge $e_{j}\in\E$ with a real interval
of finite length $l_{j}>0$. Such a graph whose edges are supplied
with lengths is called a \emph{metric graph}. We commonly put all
of the graph edge lengths into a vector, $\lv=\left(l_{1},l_{2}...l_{E}\right)$
and denote the sum of all its entries by $\left|\Gamma\right|:=\sum_{j=1}^{E}l_{j}$.
This is also called the total length of the graph. A common assumption
in this paper is that the set of edge lengths form a linear independent
set over $\Q$. We abbreviate and call such a set \emph{rationally
independent.}

\subsection{Standard quantum graphs \label{subsec: Spectral-theory-of-Quantum-Graphs}}

It is convenient to describe a function $f$ on a metric graph $\Gamma$
in terms of its restrictions to edges, $f|_{e}:\left[0,l_{e}\right]\rightarrow\C$,
for $e\in\E$. The following function spaces are also defined in this
manner:
\begin{equation}
L^{2}\left(\Gamma\right):=\oplus_{e\in\E}L^{2}\left(\left[0,l_{e}\right]\right),\quad H^{2}\left(\Gamma\right):=\oplus_{e\in\E}H^{2}\left(\left[0,l_{e}\right]\right),
\end{equation}
where $H^{2}$ denotes a Sobolev space of order two. The Laplace operator
$\Delta:H^{2}\left(\Gamma\right)\rightarrow L^{2}\left(\Gamma\right)$
is defined edgewise by
\[
\Delta\ :\ f|_{e}\mapsto-\frac{d^{2}}{dx_{e}^{2}}f|_{e},
\]
where $x_{e}\in\left[0,l_{e}\right]$ is a coordinate chosen along
the edge $e$. In order for the Laplacian to be self-adjoint, its
domain is restricted to functions in $H^{2}\left(\Gamma\right)$ that
satisfy certain vertex conditions. A description of all vertex conditions
for which the Laplacian is self-adjoint can be found for example in
\cite{BerKuc_graphs}. Throughout this paper we only consider the\emph{
Neumann vertex conditions} (for which the Laplacian is indeed self-adjoint)
. A function $f\in H^{2}\left(\Gamma\right)$ is said to satisfy Neumann
vertex conditions (also known as Kirchhoff or standard conditions)
at a vertex $v\in\V$ if
\begin{enumerate}
\item The function $f$ is continuous at $v\in\V$, i.e.,
\begin{equation}
\forall e_{1},e_{2}\in\E_{v\,\,\,\,\,}f|_{e_{1}}\left(v\right)=f|_{e_{2}}\left(v\right).\label{eq:Neumann_continuity}
\end{equation}
\item The outgoing derivatives of $f$ at $v$, denoted by $\partial_{e}f\left(v\right)$
for every $e\in\E_{v}$, satisfy
\begin{equation}
\sum_{e\in\E_{v}}\partial_{e}f\left(v\right)=0.\label{eq:Neumann_deriv_conditions}
\end{equation}
\end{enumerate}
A degree two vertex with Neumann conditions may be eliminated without
changing the graph's spectral properties (see \cite[ex. 2.2]{Berkolaiko_qg-intro17},\cite[q. 2,(e)]{BanGnu_qg-exerices18}).
This allows to assume that the graph has no vertices of degree two,
which we indeed assume throughout this paper.
\begin{defn}
A \emph{standard quantum graph }$\Gamma$ is a connected metric graph,
with finitely many vertices and edges, and no vertices of degree two.
This metric graph is equipped with the Laplace operator and Neumann
vertex conditions at all vertices.
\end{defn}

\begin{rem}
This definition excludes the case of a single loop graph. In such
case all eigenvalues are degenerate and so non of the results in this
paper applies to the loop graph.
\end{rem}

If $\Gamma$ is a standard quantum graph then the corresponding Laplacian
is self-adjoint with real discrete spectrum, which we order increasingly,
as follows
\begin{equation}
0=\lambda_{0}<\lambda_{1}\le\lambda_{2}\nearrow\infty,
\end{equation}
 noting that each eigenvalue in this sequence appears as many times
as its multiplicity. There exists a choice of a real orthonormal $L^{2}\left(\Gamma\right)$
basis of eigenfunctions $\left\{ f_{n}\right\} _{n=0}^{\infty}$ corresponding
to the eigenvalues sequence \cite{BerKuc_graphs}. The choice of this
basis may not be unique (if there are non-simple eigenvalues) but
the results in this paper hold for any choice of basis. Note that
the first index is zero, so that $f_{0}$ is the constant eigenfunction
which corresponds to the eigenvalue $\lambda_{0}=0$ (which is a simple
eigenvalue, as we assume $\Gamma$ is connected). For convenience,
instead of the eigenvalues themselves, we consider their square roots,
$k_{n}:=\sqrt{\lambda_{n}}.$ Further information on the fundamental
theory of quantum graphs may be found in \cite{BerKuc_graphs,GnuSmi_ap06}.

\subsection{Loop-eigenfunctions and generic eigenfunctions\label{subsec: loop_eigenfunctions_and_generic_eigenfunctions}}

Let $\Gamma$ be a standard graph. An eigenfunction which is supported
on a single loop and vanishes elsewhere on the graph is called a\emph{
loop-eigenfunction.} Explicitly, a function $f$ is a loop-eigenfunction
supported on the loop $e$ if and only if
\begin{align}
k & \in\frac{2\pi}{l_{e}}\N\quad\quad\textrm{and }\quad f|_{e}\left(x\right)=A\sin\left(kx\right),\quad f|_{\Gamma\setminus e}\equiv0,
\end{align}
for some $A\in\C$ and arc-length parametrization $x\in\left[0,l_{e}\right]$.
In particular if a graph has loops then each of the loops has infinitely
many loop-eigenfunctions supported on it.
\begin{defn}
\label{def: generic_eigenfunction} Let $\Gamma$ be a standard graph.
Let $f$ be an eigenfunction of $\Gamma$. We call $f$ a \emph{generic
}eigenfunction\emph{ }if it satisfies all of the following.
\begin{enumerate}
\item \label{enu: def-generic-simple-evalue}It corresponds to a simple
eigenvalue.
\item \label{enu: def-generic-non-zero-value-at-vertices} It does not vanish
at vertices, $\forall v\in\V$~~ $f\left(v\right)\ne0$.
\item \label{enu:None-of-theenu: def-generic-non-zero-derivative-at-vertices}
None of the outgoing derivatives vanish at interior vertices,\\
 $\forall v\in\V\setminus\partial\Gamma\,,\,\forall e\in\E_{v}\,\,\,\,\partial_{e}f\left(v\right)\ne0$.
\end{enumerate}
\end{defn}

\begin{rem}
\label{rem: generic for trees}It is shown in \cite[Corollary 3.1.9]{BerKuc_graphs}
that if $\Gamma$ is a tree, then any eigenfunction $f$ satisfying
condition (\ref{enu: def-generic-non-zero-value-at-vertices}) must
correspond to a simple eigenvalue. Hence for trees its enough for
an eigenfunction to satisfy both conditions (\ref{enu: def-generic-non-zero-value-at-vertices})
and (\ref{enu:None-of-theenu: def-generic-non-zero-derivative-at-vertices})
to be generic.
\end{rem}

Once a certain basis of eigenfunctions, $\left\{ f_{n}\right\} _{n=0}^{\infty}$
is chosen we may define the following subsets of $\N$:
\begin{align}
\mathcal{G} & :=\left\{ n\in\N~:~\textrm{\ensuremath{f_{n}\,}is\,generic}\right\} \\
\L & :=\left\{ n\in\N~:~\textrm{\ensuremath{f_{n}\,}is\,a\,loop-eigenfunction}\right\} .
\end{align}

Observe that a loop-eigenfunction is not generic, so that $\G\cap\L=\emptyset$.
In order to quantify how many of the eigenfunctions belong to those
sets, we introduce.
\begin{defn}
\label{def: density}~
\begin{enumerate}
\item Let $A\subset\N$ and denote $A\left(N\right):=A\cap\left\{ 1,2,...N\right\} $
for some $N\in\N$. We say that $A$ has \emph{natural density} $d\left(A\right)$
if the following limit exists
\[
d\left(A\right):=\lim_{N\rightarrow\infty}\frac{\left|A\left(N\right)\right|}{N}.
\]
\item If $\G$ has a positive density, $d(\G)>0$ and $A\subset\G$ such
that $A$ has density, we define the relative density of $A$ in $\G$
by
\[
d_{\G}\left(A\right):=\lim_{N\rightarrow\infty}\frac{\left|A(N)\right|}{\left|\G\left(N\right)\right|}=\frac{d\left(A\right)}{d\left(\G\right)}.
\]
\end{enumerate}
\end{defn}

The densities of $\L$ and $\G$ are given in the following theorem\footnote{This Theorem generalizes \cite[Proposition A.1]{AloBanBer_cmp18}
and \cite[Theorem 3.6]{BerLiu_jmaa17}, as is proven and discussed
in \cite{Alon}.}.
\begin{thm}
\cite{Alon,AloBanBer_cmp18}\label{thm: density-of-generic-and-loop-eigenfunctions}
Let $\Gamma$ be a standard graph with rationally independent edge
lengths. Then both $\L$ and $\G$ have natural densities, given by
\[
d\left(\L\right)=\frac{L_{\mathrm{loops}}}{2\left|\Gamma\right|},\quad\quad d\left(\G\right)=1-d\left(\L\right),
\]
where $L_{\mathrm{loops}}$ is the total length of the loops of the
graph.

In particular $d(\G)\geq\frac{1}{2}$ and almost all non-loop-eigenfunctions
are generic.
\end{thm}

The results of this paper are stated for generic eigenfunctions. Accordingly,
the probabilistic statements in the paper are stated using $d_{\G}$,
rather than the natural density, $d$ (see Section \ref{sec: main_results}).
The last theorem shows that $d_{\G}$ and $d$ differ only if the
graph has loops.

\subsection{Neumann domains and Neumann count}
\begin{defn}
\label{def: Neumann_points_Nodal_points_etc} Let $f$ be a generic
eigenfunction of a standard graph $\Gamma$. An interior point $x\in\Gamma\setminus\V$
is called a \emph{nodal point} if $f\left(x\right)=0$ and is called
a \emph{Neumann point }if $f'\left(x\right)=0$.

Removing the nodal points of $f$ from $\Gamma$ disconnects the graph.
The connected components of this new graph are called the \emph{nodal
domains} of $f$. Similarly, the connected components of $\Gamma$
without $f$'s Neumann points, are called the \emph{Neumann domains}
of $\Gamma$ (see Figure \ref{fig:Neumann domain example}) . If a
Neumann domain is a single interval we call it a \emph{trivial Neumann
domain}. Hence, a Neumann domain is \emph{non-trivial} if it contains
some vertex of degree at least three.
\end{defn}

\begin{figure}[h]
\centering{}\includegraphics[width=1\textwidth]{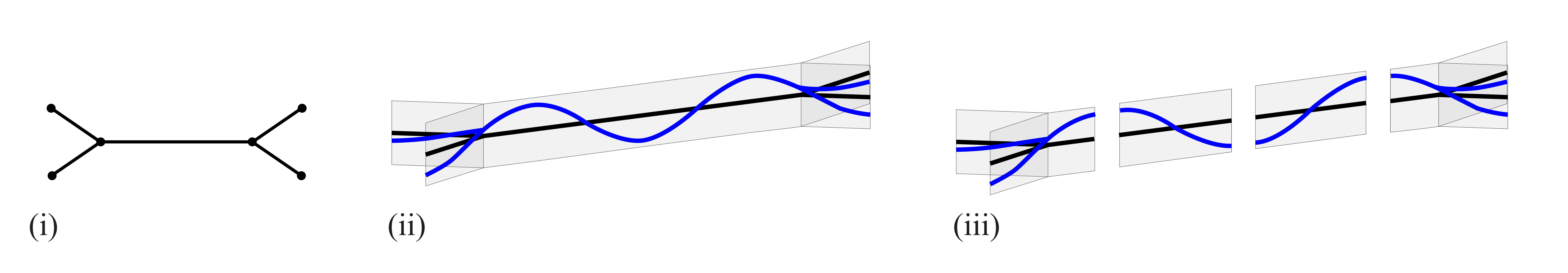}
\caption{(i) A graph $\Gamma$ (ii) An eigenfunction $f$ of $\Gamma$ (iii)
A decomposition of $\Gamma$ into four Neumann domains according to
$f$. The middle two Neumann domain are trivial Neumann domains. }
\label{fig:Neumann domain example}
\end{figure}

The names nodal domains and Neumann domains are adopted from similar
definitions for manifolds. For graphs these names might be deceiving,
as a nodal domain is actually a graph (which may be considered as
a sub-graph of $\Gamma$), and similarly for a Neumann domain.

If $f$ is a generic eigenfunction then it has a finite number of
nodal points and a finite number of Neumann points and we denote those
numbers by
\begin{align*}
\phi(f) & :=\left|\left\{ x\in\Gamma\setminus\V\,|\,f\left(x\right)=0\right\} \right|,\\
\xi(f) & :=\left|\left\{ x\in\Gamma\setminus\V\,|\,f'\left(x\right)=0\right\} \right|,
\end{align*}
and correspondingly name them the \emph{nodal count} and the \emph{Neumann
count }of $f$ ($\phi$ stands for sign \textbf{\emph{f}}lips of $f$,
and $\xi$ stands for e\textbf{\emph{x}}tremal points of $f$)\emph{.}
The nodal count is well-studied both for graphs and manifolds. A fundamental
result (see\footnote{To avoid confusion when comparing to those works, recall that we start
indexing the eigenvalues from $n=0$.} \cite{GnuSmiWeb_wrm04,Ber_cmp08,BanBerSmi_ahp12}) is the following
bounds for the nodal count of the $n^{th}$ eigenfunction, $f_{n}$,
\begin{equation}
0\le\phi\left(f_{n}\right)-n\le\beta,\label{eq: nodal surplus bound}
\end{equation}
where $\beta$ is the graph's first Betti number, (\ref{eq:Betti_number}).
It is common to denote \emph{$\sigma(n):=\phi(f_{n})-n$} and name
this by \emph{nodal surplus. }Analogously, denote \emph{
\[
\omega\left(n\right):=\xi\left(f_{n}\right)-n,
\]
}and call $\omega(n)$ the \emph{Neumann surplus}. Note that the name
surplus might be misleading in this case, as possibly $\xi(f_{n})<n$
(see Theorem \ref{thm:Neumann_surplus_main}).

We note that the functions $\sigma$ and $\omega$ are defined not
for every natural number, but only for every $n\in\G$.

\subsection{Spectral position and wavelength capacity\label{subsec:Spectral-position-and-wavelength-capacity}}

\subsubsection*{\uline{Spectral position}}

The restriction of an eigenfunction to one of its nodal domains is
an eigenfunction of that domain with Dirichlet boundary conditions.
Furthermore, a well known (and commonly used) observation is that
the restricted eigenfunction is the first Dirichlet eigenfunction
(ground state) of that domain. This statement holds for manifolds,
as well as for graphs.

It is natural to inquire whether an analogous result holds for Neumann
domains. Namely, whether the restriction of an eigenfunction to a
Neumann domain is a Neumann eigenfunction of that domain; and whether
it is the first non-trivial (i.e., non constant) Neumann eigenfunction.
The quick (and somewhat superficial) answer is that for manifolds,
generically, the restricted eigenfunction is indeed a Neumann eigenfunction,
but it is not necessarily the first non-trivial Neumann eigenfunction
\cite{BanFaj_ahp16,BanEggTay_jga20,BanCoxEgg_arxiv20,BanCoxEgg_prep21}.
In the current paper we treat this problem for quantum graphs. Given
a generic eigenfunction $f$ with eigenvalue $k^{2}$ and a Neumann
domain $\Omega$ of $f$, it is easy to show that $\left.f\right|_{\Omega}$
is an eigenfunction with eigenvalue $k^{2}$ of $\Omega$, considered
as a standard graph (just observe that the vertex conditions are satisfied).
The interesting question would be what is the position of $k^{2}$
in the spectrum of $\Omega$. To fix terminology, we introduce
\begin{defn}
\label{def: spectral position} ~

Let $f$ be a generic eigenfunction with eigenvalue $k^{2}$ and let
$\Omega$ be a Neumann domain of $f$. We define the \emph{spectral
position} of $\Omega$ as
\[
N(\Omega):=\left|\left\{ 0\leq\lambda<k^{2}\,|\,\lambda\textrm{ is an eigenvalue of }\Omega\right\} \right|,
\]
where $\Omega$ is considered as a standard graph. This notation does
not explicitly include $k$, which is assumed to be understood from
the context.
\end{defn}

\subsubsection*{\uline{Wavelength capacity}}
\begin{defn}
\label{def: rho} \,

Let $f$ be a generic eigenfunction with eigenvalue $k^{2}$ and let
$\Omega$ be a Neumann domain of $f$. We define the \emph{wavelength
capacity} of $\Omega$ as
\[
\rho(\Omega):=\frac{\left|\Omega\right|k}{\pi},
\]
where $\left|\Omega\right|$ is the sum of edge lengths of $\Omega$.
This notation of the wavelength capacity does not explicitly include
$k$. It is assumed that the value of $k$ is understood from the
context.
\end{defn}

The meaning of the definition is easily demonstrated when $\Omega$
is an interval. Then, $\rho(\Omega)$ counts the number of oscillations
of an eigenfunction $f$ of $\Omega$ which corresponds to the eigenvalue
$k^{2}$. In other words, $\rho(\Omega)$ counts the number of 'half-wavelengths'
of $f$ within $\Omega$; hence, its name.

We further note that the wavelength capacity is the one-dimensional
analogue of a similar parameter for Neumann domains on manifolds.
The normalized area to perimeter ratio of a Neumann domain $\Omega$
on a two-dimensional manifold is defined to be $\frac{\left|\Omega\right|\sqrt{\lambda}}{\left|\partial\Omega\right|}$,
where $\left|\Omega\right|$ is the area of the Neumann domain and
$\left|\partial\Omega\right|$ is its perimeter length, \cite{AloBanBerEgg_lms20,EGJS07}.

\section{Main results \label{sec: main_results}}

Our main results concern the properties mentioned in the previous
section: Neumann count, nodal count, spectral position and wavelength
capacity. The Neumann count and the nodal count are properties of
an eigenfunction on the whole graph, so we call those global observables.
The spectral position and wavelength capacity are attributes of individual
Neumann domains and they are called local observables. For each of
those observables we prove bounds and basic features of their probability
distributions. Denote $\G\left(N\right):=\G\cap\left\{ 0,1,...N\right\} $.
The probability distribution we assign to an observable $h:\G\rightarrow\R$
is characterized by
\[
d_{\G}\left(h^{-1}\left(B\right)\right):=\lim_{N\rightarrow\infty}\frac{\left|\set{n\in\G\left(N\right)}{h\left(n\right)\in B}\right|}{\left|\G\left(N\right)\right|}\,;\,B\subset\R,
\]
if such a limit exists. Indeed, it is natural to define the probability
distributions of those observables in terms of the density, $d_{\G}$.
However, $d_{\G}$ is not necessarily a probability measure on $\G$.
Hence, we are required to determine whether each of the observables
is a random variable with respect to $d_{\G}$ (the answer is not
always obvious, as can be seen in the rest of this section).
\begin{defn}
\label{def: random variable}A function $h:\G\rightarrow\R$ is called
a random variable with respect to $d_{\G}$ if for every $B\subset\R$
Borel set, $d_{\G}(h^{-1}(B))$ exists and $d_{\G}$ is a probability
measure on the $\sigma$-algebra generated by $h$.
\end{defn}

\subsection{Global observables}
\begin{thm}
\label{thm:Neumann_surplus_main} Let $\Gamma$ be a standard graph
whose first Betti number is $\beta$ and boundary size is $\left|\d\Gamma\right|$.
\begin{enumerate}
\item \label{enu:thm-Neumann_surplus_bounds} Let $f_{n}$ be the $n$'s
eigenfunction and assume it is generic. The Neumann surplus $\omega(n):=\mu\left(f\right)-n$
is bounded by
\begin{equation}
1-\beta-\left|\d\Gamma\right|\leq\omega(n)\leq2\beta-1.\label{eq: Neumann_Surplus_bounds}
\end{equation}
\item \label{enu:thm-density_of_Neumann_surplus} Further assume that $\Gamma$
has rationally independent edge lengths. Then the following holds
\begin{enumerate}
\item \label{enu:thm-density_of_Neumann_surplus-a} The Neumann surplus,
$\omega$, is a finite random variable with respect to $d_{\G}$.
In particular, the probability distribution of $\omega$ is
\[
\P(\omega=j):=d_{\G}\left(\omega^{-1}(j)\right).
\]

\item \label{enu:thm-density_of_Neumann_surplus-b} If $\omega^{-1}(j)\neq\emptyset$
then $\P(\omega=j)>0$.\\
In particular, any value which $\omega$ attains is obtained infinitely
often.
\item \label{enu:thm-density_of_Neumann_surplus-c} The probability distribution
of $\omega$ is symmetric around $\frac{1}{2}(\beta-\partial\Gamma)$.
Namely,
\[
\P(\omega=j)=\P(\omega=\beta-\left|\d\Gamma\right|-j).
\]
\end{enumerate}
\end{enumerate}
\end{thm}

An interesting (and simple) observation is that the symmetry of the
random variables $\omega$ and $\sigma$ implies their expected values.
Namely, from the last part of Theorem \ref{thm:Neumann_surplus_main}
together with \cite[Theorem 2.1]{AloBanBer_cmp18} we deduce
\begin{cor}
\label{cor:Expected_values_of_nodal_and_Neumann_supluses} Denoting
$\G(N):=\G\cap\{1,\ldots,N\}$, we have
\begin{align*}
\mathbb{E}\left(\sigma\right)=\lim_{N\rightarrow\infty}\frac{1}{\left|\G(N)\right|}\sum_{n\in\G(N)}\sigma(n) & =\frac{\beta}{2}\\
\mathbb{E}\left(\omega\right)=\lim_{N\rightarrow\infty}\frac{1}{\left|\G(N)\right|}\sum_{n\in\G(N)}\omega(n) & =\frac{\beta-\left|\partial\Gamma\right|}{2}.
\end{align*}
\end{cor}

Thus, the expected values of the nodal surplus and of the Neumann
surplus store two topological properties of the graph. We discuss
the importance of this result in the context of inverse problems in
Section \ref{sec:Discussion}.

The first part of Theorem \ref{thm:Neumann_surplus_main} provides
the bounds (\ref{eq: Neumann_Surplus_bounds}) on the Neumann surplus.
We conjecture that for large $\beta$ values, better bounds hold.
\begin{conjecture}
\label{conj:strict_bounds_Neumann_surplus} Let $f$ be a generic
eigenfunction whose spectral position is $n$. The Neumann surplus
$\omega(n):=\mu\left(f\right)-n$ is bounded by
\begin{equation}
-1-\left|\d\Gamma\right|\le\omega(n)\le\beta+1.\label{eq:strict_bounds_Neumann_surplus}
\end{equation}
\end{conjecture}

We note that for $\beta>2$ the bounds in (\ref{eq:strict_bounds_Neumann_surplus})
are stricter than those in (\ref{eq: Neumann_Surplus_bounds}). This
conjecture is supported by a numerical investigation as well as analytic
study of particular graph families; See Appendix \ref{sec: Appendix-examples}.

Theorem \ref{thm:Neumann_surplus_main} sums up our knowledge on the
probability distribution of the Neumann surplus for general graphs.
No explicit expression for this distribution is known in general.
Nevertheless, for particular graph families we do have a concrete
expression.
\begin{defn}
Let $d_{1},d_{2}\in\N$. A graph, each of whose vertices is either
of degree $d_{1}$ or of degree $d_{2}$ is called a $(d_{1},d_{2})$-regular
graph.
\end{defn}

\begin{thm}
\label{Thm: (3,1)-regular-trees} Let $\Gamma$ be a standard graph
which is a $(3,1)$-regular tree. If $\Gamma$ has rationally independent
edge lengths, then the probability distribution of the random variable
$-\omega-1$ is binomial, $\textrm{Bin}(\left|\partial\Gamma\right|-2,\frac{1}{2})$.
Explicitly, for every integer $-\left|\dg\right|+1\le j\le-1$,
\begin{equation}
\P(\omega=j)=\binom{\left|\partial\Gamma\right|-2}{-j-1}2^{2-\left|\partial\Gamma\right|}.\label{eq: Binomial_distribution_of_Neumann_surplus}
\end{equation}
\end{thm}

This last theorem may be perceived as the Neumann domain analogue
of \cite[Theorem 2.3]{AloBanBer_cmp18}, though each of those theorems
applies to completely different families of graphs. A further discussion
appears in Section \ref{sec:Discussion}.

\subsection{Local observables}
\begin{prop}
\label{prop:local_observables_bounds} Let $\Gamma$ be a standard
graph with minimal edge length $L_{min}$. Let $f_{n}$ be a generic
eigenfunction which corresponds to $k_{n}>\frac{\pi}{L_{min}}$ and
let $\Omega$ be a Neumann domain of $f_{n}$. The bounds on the spectral
position of $\Omega$ and its wavelength capacity are
\begin{align}
1\le & N(\Omega)\le\left|\partial\Omega\right|-1\label{eq: Spectral_Position_bounds}\\
1\leq\frac{1}{2}(N(\Omega)+1)\le & \rho(\Omega)\le\frac{1}{2}(N(\Omega)+\left|\partial\Omega\right|-1)\leq\left|\partial\Omega\right|-1\label{eq: Rho_bounds}
\end{align}
\end{prop}

\begin{rem}
\label{rem:why_conditioning_on_large_eigenvalue} The condition $k_{n}>\frac{\pi}{L_{\mathrm{min}}}$
is satisfied for almost all eigenvalues. Indeed, there are at most
$2\frac{\left|\Gamma\right|}{L_{\mathrm{min}}}$ eigenvalues which
do not satisfy this condition \cite[Theorem 1]{Fri_aif05}. This condition
is needed in the proposition above in order to guarantee that $\Omega$
is a star graph (for more details, see Lemma \ref{lem:spectral-position-equals-nodal-count}).
We also note that the lower bounds in (\ref{eq: Spectral_Position_bounds})
and (\ref{eq: Rho_bounds}) hold even without conditioning on the
value of $k$ (see also the proof of the proposition).
\end{rem}

\begin{rem}
In Appendix \ref{sec:Appendix nodal domains} we prove an analogous
proposition for nodal domains.
\end{rem}

Next, we discuss the probability distributions of the local observables
$N$ and $\rho$. As implied by Proposition \ref{prop:local_observables_bounds}
those observables have non-trivial values only if the corresponding
Neumann domain is non-trivial, i.e., has $\left|\partial\Omega\right|>2$,
or equivalently if the corresponding Neumann domain contains an interior
vertex $v\in\V\setminus\partial\Gamma$.
\begin{defn}
\label{def: vertex-Neumann-domains} Let $f_{n}$ be a generic eigenfunction
of a standard graph $\Gamma$ and let $v\in\V\setminus\partial\Gamma$
be an interior vertex. We denote by $\Omega_{n}^{(v)}$ the unique
Neumann domain of $f_{n}$ which contains $v$. We denote the spectral
position of this Neumann domain by $N^{(v)}(n):=N(\Omega_{n}^{(v)})$
and its wavelength capacity by $\rho^{(v)}(n):=\rho(\Omega_{n}^{(v)})$.
\end{defn}

\begin{prop}
\label{prop: statistics_of_spectral_position} Let $\Gamma$ be a
standard graph with rationally independent edge lengths. Let $v\in\V\setminus\partial\Gamma$.
Then the following hold.
\begin{enumerate}
\item \label{enu:prop-statistics_of_local_spectral_position-1} The spectral
position, $N^{(v)}$, is a finite random variable with respect to
$d_{\G}$. In particular, the probability distribution of $N^{(v)}$
is
\[
\P(N^{(v)}=j):=d_{\G}\left(\left(N^{(v)}\right)^{-1}(j)\right).
\]

\item \label{enu:prop-statistics_of_local_spectral_position-2} The probability
distribution of $N^{(v)}$ is symmetric around $\frac{1}{2}\deg v$.
Namely,
\[
\P(N^{(v)}=j)=\P(N^{(v)}=\deg v-j).
\]
\end{enumerate}
\end{prop}

\begin{prop}
\label{prop: statistics_of_wavelength_capacity} Let $\Gamma$ be
a standard graph with rationally independent edge lengths. Let $v\in\V\setminus\partial\Gamma$.
\begin{enumerate}
\item \label{enu:prop-statistics_of_local_wave_capacity-1} There exists
a discrete measure $p^{(v)}$ supported on a finite set $\left\{ x_{j}\right\} _{j=1}^{m}$,
and a density function $\pi^{(v)}$, such that for every interval
$\left(a,b\right)\subset\R$
\begin{equation}
d_{\G}\left(\left(\rho^{(v)}\right)^{-1}\left((a,b)\right)\right)=\int_{a}^{b}\pi^{(v)}\left(x\right)dx+\sum_{x_{j}\in\left(a,b\right)}p^{(v)}\left(x_{j}\right).\label{eq:denstiy-as-integral-of-distribution}
\end{equation}
\item \label{enu:prop-statistics_of_local_wave_capacity-2} Both $\pi^{(v)}$
and $p^{(v)}$ are symmetric around $\frac{1}{2}\deg v$.
\end{enumerate}
\end{prop}

By the proposition above the pre-image $\left(\rho^{(v)}\right)^{-1}\left((a,b)\right)$
has density for every interval $\left(a,b\right)\subset\R$. Yet,
if $\pi^{(v)}\not\equiv0$ then we prove that there exist Borel sets,
$B\subset\R$, whose pre-images $\left(\rho^{(v)}\right)^{-1}(B)$
do not have density\footnote{Note that $\pi^{(v)}\not\equiv0$ is equivalent to $\imag(\rho^{(v)})$
being infinite.}. In this case, $\rho^{(v)}$ is not a random variable with respect
to $d_{\G}$. See Lemma \ref{lem: random_variable_implies_countable_image}
and Remark \ref{rem: wavelength-capacity-not-random-variable}. Furthermore,
we conjecture that $\pi^{(v)}\not\equiv0$ and $p^{(v)}\equiv0$,
i.e., the distribution on the RHS of (\ref{eq:denstiy-as-integral-of-distribution})
is absolutely continuous with no atoms and so $\rho^{(v)}$ is never
a random variable with respect to $d_{\G}$. See further discussion
after the proof of this proposition, in Remark \ref{rem: no_atoms_in_rho_distribution}.

\subsection{Local-global connections}

We end the results section by presenting some connections between
the values of local observables and those of global ones.
\begin{prop}
\label{prop:local-global-connections} Let $\Gamma$ be a standard
graph with minimal edge length $L_{min}$. Let $f_{n}$ be a generic
eigenfunction which corresponds to $k_{n}>\frac{\pi}{L_{min}}$.
\begin{enumerate}
\item The sum of the spectral positions of all non-trivial Neumann domains
is
\begin{equation}
\sum_{v\in\V\setminus\d\Gamma}N(\Omega_{n}^{(v)})=\sigma(n)-\omega(n)+\left(E-\left|\partial\Gamma\right|\right).\label{eq:sum_of_spec_pos}
\end{equation}
\item The sum of the wavelength capacities of all non-trivial Neumann domains
is
\begin{equation}
\sum_{v\in\V\setminus\d\Gamma}\rho(\Omega_{n}^{(v)})=\frac{\left|\Gamma\right|k_{n}}{\pi}-\left(\omega(n)+n\right)+\left(E-\left|\partial\Gamma\right|\right).\label{eq:sum_of_wavelength_capacities}
\end{equation}
\item The difference of the two expressions above is
\begin{equation}
\sum_{v\in\V\setminus\d\Gamma}\left(N(\Omega_{n}^{(v)})-\rho(\Omega_{n}^{(v)})\right)=\left(n-\frac{\left|\Gamma\right|k_{n}}{\pi}\right)+\sigma(n).\label{eq: sum_of_oscillatory_terms}
\end{equation}
\end{enumerate}
\end{prop}

Note that in the second part of the proposition, the term $\frac{\left|\Gamma\right|k_{n}}{\pi}$
may be perceived as the wavelength capacity for the whole graph (and
so it is another global observable). This is also the value of the
well known Weyl term, which approximates the spectral position of
the eigenvalue $k_{n}$.

The last part of the proposition is an immediate implication of the
first two parts. We mention it explicitly thanks to its insightful
spectral meaning. The term $n-\frac{\left|\Gamma\right|k_{n}}{\pi}$
in the RHS is the difference between the actual spectral position
of an eigenvalue $k_{n}$ and the Weyl term. This difference appears
in the so called trace formula for quantum graphs, where it is expressed
as an infinite sum of oscillatory terms corresponding to periodic
orbits on the graph (\cite{KotSmi_ap99,KotSmi_prl97}). In the LHS
of (\ref{eq: sum_of_oscillatory_terms}) we have a sum over the local
analogues of a similar quantity, $N(\Omega_{n}^{(v)})-\rho(\Omega_{n}^{(v)})$.
It is interesting that the deviation between the global observable
$n-\frac{\left|\Gamma\right|k_{n}}{\pi}$ and the sum of the local
observables, $N(\Omega_{n}^{(v)})-\rho(\Omega_{n}^{(v)})$, is given
by the nodal surplus.

\section{Proofs of bounds and basic identities \label{sec: proofs_of_bounds}}
\begin{proof}
[Proof of Theorem \ref{thm:Neumann_surplus_main}, (\ref{enu:thm-Neumann_surplus_bounds})
]

The main step in the proof is to show that the difference between
the nodal count and the Neumann count of any generic eigenfunction
$f$ on $\Gamma$ satisfies the bounds
\begin{equation}
1-\beta\leq\phi\left(f\right)-\xi\left(f\right)\leq\beta-1+\left|\partial\Gamma\right|.\label{eq:bounds_on_nodal_Neumann_diff}
\end{equation}
Once we have this, the bounds (\ref{eq: Neumann_Surplus_bounds})
follow by
\begin{enumerate}
\item Observing that if $f$ is the $n$-th eigenfunction then $\phi(f)-\xi(f)=\sigma(n)-\omega(n)$.
\item Applying the bounds for $\sigma(n)$ when $f_{n}$ is generic (see
\cite[Thm 2.6]{Ber_cmp08},\cite[(1.16)]{BanBerRazSmi_cmp12}):
\begin{equation}
0\leq\sigma(n)\leq\beta\label{eq: nodal surplus bounds}
\end{equation}
\end{enumerate}
Now, to prove (\ref{eq:bounds_on_nodal_Neumann_diff}) we start with
the following decomposition
\begin{equation}
\phi\left(f\right)-\xi\left(f\right)=\sum_{e\in\E}\phi\left(f|_{e}\right)-\xi\left(f|_{e}\right),\label{eq:Diff-nodal_Neumann-decomposition_to_edges}
\end{equation}
where $\phi\left(f|_{e}\right),\xi\left(f|_{e}\right)$ are the nodal
and Neumann counts on the edge $e$.

Denoting the vertices of $e$ by $u,v$ and the outgoing derivatives
of $f$ at those vertices by $\partial_{e}f\left(v\right),\partial_{e}f\left(u\right)$,
we show next that

\begin{equation}
\phi\left(f|_{e}\right)-\xi\left(f|_{e}\right)=-\frac{\sgn\left(f\left(v\right)\partial_{e}f\left(v\right)\right)+\sgn\left(f\left(u\right)\partial_{e}f\left(u\right)\right)}{2},\label{eq:Diff-nodal-Neumann-on-edge}
\end{equation}
where
\[
\sgn\left(x\right):=\begin{cases}
1 & x>0\\
-1 & x\le0
\end{cases}.
\]

Clearly, for each $e\in\E$ we have $f|_{e}\left(x\right)=A_{e}\cos\left(\varphi_{e}+kx\right)$,
for some $A_{e},\varphi_{e}\in\R$ and arc-length parametrization
$x\in\left[0,l_{e}\right]$. We continue by assuming that all edge
lengths satisfy $l_{e}>\frac{2\pi}{k}$. This assumption is justified
by two observations: (a) extending the interval $\left[0,l_{e}\right]$
by $\frac{2\pi}{k}$, while keeping $f|_{e}\left(x\right)=A_{e}\cos\left(\varphi_{e}+kx\right)$
does not change the values and derivatives of $f$ at the endpoints
of the interval (i.e., $f(u),~f(v),~\partial_{e}f(u),~\partial_{e}f(v)$
are not changed by such an extension); (b) this extension does not
change the value of $\phi\left(f|_{e}\right)-\xi\left(f|_{e}\right)$.

The assumption $l_{e}>\frac{2\pi}{k}$ guarantees that there are at
least two Neumann points and two nodal points on each edge. Now, examine
the sets of nodal and Neumann points along the edge $e$. As the locations
of these points interlace, the difference in their count can be either
$0$ or $\pm1$. To determine the value of this difference we only
need to know whether a nodal point is the nearest to the vertex $v$
or is it a Neumann point (and similarly for the vertex $u$).

If $v\in\V\setminus\d\Gamma$ then $f\left(v\right)\partial_{e}f\left(v\right)\ne0$
since $f$ is assumed to be generic. If $f\left(v\right)\partial_{e}f\left(v\right)>0$
then a Neumann point is closer to $v$ than any nodal point and if
$f\left(v\right)\partial_{e}f\left(v\right)<0$ it is the other way
around. If $v\in\d\Gamma$ then the nearest point to $v$ is always
a nodal point and by the vertex conditions which $f$ satisfies we
have $f\left(v\right)\partial_{e}f\left(v\right)=0$. The arguments
above yield (\ref{eq:Diff-nodal-Neumann-on-edge}).

Substituting (\ref{eq:Diff-nodal-Neumann-on-edge}) into (\ref{eq:Diff-nodal_Neumann-decomposition_to_edges})
and changing summation to be over vertices gives
\begin{align}
\phi\left(f\right)-\xi\left(f\right) & =-\frac{1}{2}\sum_{v\in\V}\sum_{e\in\E_{v}}\sgn\left(f\left(v\right)\partial_{e}f\left(v\right)\right)\nonumber \\
 & =\frac{\left|\partial\Gamma\right|}{2}-\frac{1}{2}\sum_{v\in\V\setminus\partial\Gamma}\sum_{e\in\E_{v}}\sgn\left(f\left(v\right)\partial_{e}f\left(v\right)\right),\label{eq:Diff-nodal-Neumann_by_vertices}
\end{align}
where moving to the last line, we used that $\sgn\left(f\left(v\right)\partial_{e}f\left(v\right)\right)=-1$
for all $v\in\d\Gamma$. Recalling that $f$ is assumed to be generic
we get that if $v\in\V\setminus\partial\Gamma$ then $f\left(v\right)\partial_{e}f\left(v\right)\ne0$
for every $e\in\E_{v}$. But since by Neumann conditions we have $\sum_{e\in\E_{v}}f\left(v\right)\partial_{e}f\left(v\right)=0$,
we conclude that this sum must include at least one positive term
and at least one negative term, so that
\begin{equation}
\forall v\in\V\setminus\partial\Gamma,\quad\left|\sum_{e\in\E_{v}}\sgn\left(f\left(v\right)\partial_{e}f\left(v\right)\right)\right|\le\deg v-2.\label{eq: sum_of_signs_absolute_value}
\end{equation}
Substituting (\ref{eq: sum_of_signs_absolute_value}) in (\ref{eq:Diff-nodal-Neumann_by_vertices})
and using the identity
\begin{equation}
2E=\sum_{v\in\V}\deg v=\left|\d\Gamma\right|+\sum_{v\in\V\setminus\d\Gamma}\deg v\label{eq: graph_identity_number_of_edges}
\end{equation}
 gives
\begin{align}
\left|\phi\left(f\right)-\xi\left(f\right)-\frac{\left|\partial\Gamma\right|}{2}\right| & \leq\frac{1}{2}\sum_{v\in\V\backslash\partial\Gamma}\left(\deg v-2\right)\nonumber \\
= & \frac{1}{2}\left(2E-\left|\d\Gamma\right|\right)-\left(V-\left|\partial\Gamma\right|\right)\\
= & E-V+\frac{\left|\partial\Gamma\right|}{2}=\beta-1+\frac{\left|\partial\Gamma\right|}{2},\label{eq:Diff-nodal-Neumann-abs-value}
\end{align}
which are exactly the bounds (\ref{eq:bounds_on_nodal_Neumann_diff}).
\end{proof}
Before proceeding to prove Propositions \ref{prop:local_observables_bounds}
and \ref{prop:local-global-connections}, we bring a lemma which is
used in the proofs of both propositions (and in other proofs as well).
\begin{lem}
\label{lem:spectral-position-equals-nodal-count} Let $\Gamma$ be
a standard graph with minimal edge length $L_{min}$. Let $f$ be
a generic eigenfunction with eigenvalue $k>\frac{\pi}{L_{min}}$,
and let $\Omega$ be a Neumann domain of $f$. Then
\begin{enumerate}
\item $\Omega$ is either a star graph or an interval.
\item $\left.f\right|_{\Omega}$ is a generic eigenfunction of $\Omega$,
considered as a standard graph.
\item The spectral position of $\Omega$ equals the nodal count of $\left.f\right|_{\Omega}$,
i.e.,
\end{enumerate}
\begin{equation}
N(\Omega)=\phi(\left.f\right|_{\Omega}).\label{eq:spectral-position-equals-nodal-count}
\end{equation}
\end{lem}

\begin{proof}
[Proof of Proposition \ref{lem:spectral-position-equals-nodal-count}]~The
first statement of the lemma follows quite straightforwardly from
the assumption $k>\frac{\pi}{L_{min}}$. Indeed, for each edge $e\in\E$
we have $l_{e}>L_{min}>\frac{\pi}{k}$ which implies that the eigenfunction
$f$ has at least one Neumann point at each edge. Hence, no Neumann
domain contains an entire edge of the original graph. It follows that
a Neumann domain is either an interval (included in one of the graph
edges) or a star graph (which contains one of the graph's interior
vertices).

Next, we show that $f|_{\Omega}$ is a generic eigenfunction of $\Omega$,
where $\Omega$ is considered as a standard graph. It is clear that
$f|_{\Omega}$ is an eigenfunction of $\Omega$. In addition, $f$
does not vanish nor have vanishing derivatives at interior vertices
of $\Gamma$. Hence, the same holds for  $f|_{\Omega}$, which means
that $f|_{\Omega}$ satisfies conditions (\ref{enu: def-generic-non-zero-value-at-vertices})
and (\ref{enu:None-of-theenu: def-generic-non-zero-derivative-at-vertices})
in the genericity definition (Definition \ref{def: generic_eigenfunction}).
As mentioned in Remark \ref{rem: generic for trees}, it is enough
to conclude that $f|_{\Omega}$ is generic since $\Omega$ is a tree
graph.

Finally, we prove (\ref{eq:spectral-position-equals-nodal-count}).
The genericity of $\left.f\right|_{\Omega}$ together with $\Omega$
being a tree graph ($\beta=0$) allows to apply (\ref{eq: nodal surplus bounds})
and conclude that the spectral position of $\left.f\right|_{\Omega}$
(which is by definition the spectral position of $\Omega$) equals
the nodal count of $\left.f\right|_{\Omega}$, i.e., $N(\Omega)=\phi(\left.f\right|_{\Omega})$.

\vspace{6mm}
\end{proof}
\begin{proof}
[Proof of Proposition \ref{prop:local-global-connections}]~

In the current proof we omit everywhere for brevity the subscripts
'$n$', using $k$, $f$ and $\Omega^{(v)}$ instead of $k_{n}$,
$f_{n}$ and $\Omega_{n}^{(v)}$ as in the statement of the Proposition.\uline{}\\
\uline{~}\\
\uline{Proof of (\mbox{\ref{eq:sum_of_spec_pos}})} Let $u\in\V\backslash\partial\Gamma$.
We start by applying equation (\ref{eq:Diff-nodal-Neumann_by_vertices})
from the proof of Theorem \ref{thm:Neumann_surplus_main},(\ref{enu:thm-Neumann_surplus_bounds}).
We apply (\ref{eq:Diff-nodal-Neumann_by_vertices}) for the graph
$\Omega^{(u)}$ and its eigenfunction $\left.f\right|_{\Omega^{(u)}}$
to get
\begin{equation}
\phi\left(\left.f\right|_{\Omega^{(u)}}\right)-\xi\left(\left.f\right|_{\Omega^{(u)}}\right)=\frac{\left|\partial\Omega^{(u)}\right|}{2}-\frac{1}{2}\sum_{e\in\E_{u}}\sgn\left(f\left(u\right)\partial_{e}f\left(u\right)\right).\label{eq: nodal-Neumann-diff-on-star-domain}
\end{equation}
By Lemma \ref{lem:spectral-position-equals-nodal-count} we get that
$\Omega^{(u)}$ is a star graph, so that $\left|\partial\Omega^{(u)}\right|=\deg u$.
Lemma \ref{lem:spectral-position-equals-nodal-count} also gives $\phi(\left.f\right|_{\Omega^{(u)}})=N(\Omega^{(u)})$.
Furthermore, $\xi(\left.f\right|_{\Omega^{(u)}})=0$, since $\Omega^{(u)}$
is a Neumann domain and does not contain interior Neumann points.
Substituting all that in (\ref{eq: nodal-Neumann-diff-on-star-domain})
gives
\begin{equation}
N\left(\Omega^{(u)}\right)=\frac{\deg u}{2}-\frac{1}{2}\sum_{e\in\E_{u}}\sgn\left(f\left(u\right)\partial_{e}f\left(u\right)\right).\label{eq: spectral_position_single_ND}
\end{equation}
Summing (\ref{eq: spectral_position_single_ND}) over all $u\in\V\backslash\partial\Gamma$
and using (\ref{eq:Diff-nodal-Neumann_by_vertices}) again yields
\[
\sum_{u\in\V\setminus\partial\Gamma}N\left(\Omega^{(u)}\right)=\sum_{u\in\V\setminus\partial\Gamma}\frac{\deg u}{2}+\phi\left(f\right)-\xi\left(f\right)-\frac{\left|\partial\Gamma\right|}{2}.
\]

Applying identity (\ref{eq: graph_identity_number_of_edges}) gives
\[
\sum_{u\in\V\setminus\partial\Gamma}N\left(\Omega^{(u)}\right)=\phi\left(f\right)-\xi\left(f\right)+E-\left|\partial\Gamma\right|,
\]
which proves (\ref{eq:sum_of_spec_pos}) since $\sigma(n)-\omega(n)=\phi\left(f\right)-\xi\left(f\right)$.\\
~\\
~\\
\uline{Proof of (\mbox{\ref{eq:sum_of_wavelength_capacities}})}

We denote by $\mathcal{W}$ the set of all trivial Neumann domains
of $f$. Those are the Neumann domains which are intervals and do
not contain any interior vertex of the graph. We have that
\begin{align}
\frac{\left|\Gamma\right|k}{\pi} & =\sum_{v\in\V\setminus\partial\Gamma}\rho\left(\Omega^{(u)}\right)+\sum_{\Omega\in\mathcal{W}}\rho\left(\Omega\right)\nonumber \\
 & =\sum_{v\in\V\setminus\partial\Gamma}\rho\left(\Omega^{(u)}\right)+\left|\mathcal{W}\right|,\label{eq: sum_of_wave_capacities_on_all_graph}
\end{align}
where the first equality follows since the Neumann domains form a
partition of the graph $\Gamma$ and the second equality follows since
for all $\Omega\in\mathcal{W}$, $\rho\left(\Omega\right)=\frac{\left|\Omega\right|}{\pi}k=1$.
To complete the proof we express $\left|\mathcal{W}\right|$ by using
the following counting argument. Each Neumann domain in $\mathcal{W}$
has two boundary points. Each non-trivial Neumann domain is a star
graph (Lemma \ref{lem:spectral-position-equals-nodal-count}) with
$\deg v$ boundary points ($v$ being the central vertex of the star).
So, counting the boundary points of all Neumann domains of $\Gamma$
(either trivial or non-trivial) gives $2\left|\mathcal{W}\right|+\sum_{v\in\V\backslash\partial\Gamma}\deg v$.
On the other hand, in this sum each Neumann point appears twice and
each boundary point of the graph appears once,
\begin{equation}
2\left|\mathcal{W}\right|+\sum_{v\in\V\backslash\partial\Gamma}\deg v=2\xi(f)+\left|\partial\Gamma\right|.\label{eq: counting_trivial_Neumann_domains}
\end{equation}
Substituting (\ref{eq: counting_trivial_Neumann_domains}), (\ref{eq: graph_identity_number_of_edges})
and $\omega(n)=\xi(f)-n$ in (\ref{eq: sum_of_wave_capacities_on_all_graph})
gives the required (\ref{eq:sum_of_wavelength_capacities}).

\vspace{6mm}
\end{proof}
\begin{proof}
[Proof of Proposition \ref{prop:local_observables_bounds}]

Employing Lemma \ref{lem:spectral-position-equals-nodal-count} the
bounds we need to prove in (\ref{eq: Spectral_Position_bounds}) are
equivalent to
\[
1\leq\phi(\left.f\right|_{\Omega})\leq\left|\partial\Omega\right|-1.
\]
These bounds follow immediately from (\ref{eq:bounds_on_nodal_Neumann_diff}),
when taking $\Gamma=\Omega$ and noting that $\xi(\left.f\right|_{\Omega})=0$
(no Neumann points within a Neumann domain) and that $\Omega$ is
a star graph, so its first Betti number is $\beta=0$. Let us only
remark that the lower bound in (\ref{eq: Spectral_Position_bounds})
is trivial since $N(\Omega)\geq1$ by definition. Hence, the lower
bound holds even without the condition $k>\frac{\pi}{L_{\textrm{min}}}$
(as is also mentioned in Remark \ref{rem:why_conditioning_on_large_eigenvalue}).

Next we prove the bounds in (\ref{eq: Rho_bounds}). Clearly, the
external bounds follow immediately from (\ref{eq: Spectral_Position_bounds})
and we only need to prove the internal bounds,
\begin{equation}
\frac{1}{2}(N(\Omega)+1)\leq\rho(\Omega)\leq\frac{1}{2}(N(\Omega)+\left|\partial\Omega\right|-1).\label{eq:bounds-for-rho-within-proof}
\end{equation}
 The lower bound follows by applying \cite[Theorem 1]{Fri_aif05}.
With our notations, the statement of \cite[Theorem 1]{Fri_aif05}
is $k\geq\frac{\pi}{2\left|\Omega\right|}(N+1)$. From here, the required
lower bound in (\ref{eq:bounds-for-rho-within-proof}) follows, as
$\rho=\frac{\left|\Omega\right|k}{\pi}$. We note just as above that
this lower bound holds without assuming $k>\frac{\pi}{L_{\textrm{min}}}$
(as is also mentioned in Remark \ref{rem:why_conditioning_on_large_eigenvalue}).

We proceed to prove the upper bound in (\ref{eq:bounds-for-rho-within-proof}).
It follows from the next lemma, whose proof is given after the proof
of the proposition.
\begin{lem}
\label{lem:dual-star}Let $\Omega$ be a standard star graph and $f$
be a generic eigenfunction of $\Omega$ with eigenvalue $k$. Assume
that $f$ has no Neumann points. Then there exists a dual standard
star graph, $\widetilde{\Omega}$, which satisfies
\begin{enumerate}
\item \label{enu:lem-dual-star-1} Both star graphs have the same number
of edges, i.e., $\left|\partial\Omega\right|=\left|\partial\widetilde{\Omega}\right|$.
\item \label{enu:lem-dual-star-2} There exists a generic eigenfunction
$\tilde{f}$ of $\widetilde{\Omega}$ with eigenvalue $k$.
\item \label{enu:lem-dual-star-3} The eigenfunction $\tilde{f}$ has no
Neumann points. \\
Hence, it has a single Neumann domain, which is the whole of $\widetilde{\Omega}$.
\item \label{enu:lem-dual-star-4} The spectral positions and the wavelength
capacities of both graphs obey
\begin{equation}
N(\Omega)+N(\widetilde{\Omega})=\left|\partial\Omega\right|\textrm{~~~and~~~ }\rho(\Omega)+\rho(\widetilde{\Omega})=\left|\partial\Omega\right|\label{eq: sum_of_star_and_dual_star}
\end{equation}
\end{enumerate}
\end{lem}

Employing Lemma \ref{lem:spectral-position-equals-nodal-count}, we
have that $\Omega$ is star graph (or an interval, which is a particular
case of a star graph) and that $\left.f\right|_{\Omega}$ is a generic
eigenfunction of $\Omega$ with eigenvalue $k$. Furthermore, as $\Omega$
is a Neumann domain, $\left.f\right|_{\Omega}$ has no Neumann points.
Hence, we may apply Lemma \ref{lem:dual-star} and write
\begin{align*}
\rho(\Omega) & =\left|\partial\Omega\right|-\rho(\widetilde{\Omega})\\
 & \leq\left|\partial\Omega\right|-\frac{1}{2}(N(\widetilde{\Omega})+1)\\
 & =\left|\partial\Omega\right|-\frac{1}{2}(\left|\partial\Omega\right|-N(\Omega)+1)\\
 & =\frac{1}{2}(N(\Omega)+\left|\partial\Omega\right|-1),
\end{align*}
where in the first and the third lines we have used (\ref{eq: sum_of_star_and_dual_star})
from Lemma \ref{lem:dual-star}; and in the second line we have applied
the lower bound in (\ref{eq:bounds-for-rho-within-proof}) for $\widetilde{\Omega}$.
As a result we get the upper bound in (\ref{eq:bounds-for-rho-within-proof}).

This concludes the proof of the proposition and it is left to provide
a proof for Lemma \ref{lem:dual-star}, which we do next.
\end{proof}
\begin{proof}
[Proof of Lemma \ref{lem:dual-star}]

The lemma is proved by providing an explicit construction of the mentioned
dual star, $\widetilde{\Omega}$. We describe this construction below
and it is also demonstrated in Figure \ref{fig: Dual graph example}.

\begin{figure}[h]
\begin{centering}
\includegraphics[width=0.9\columnwidth]{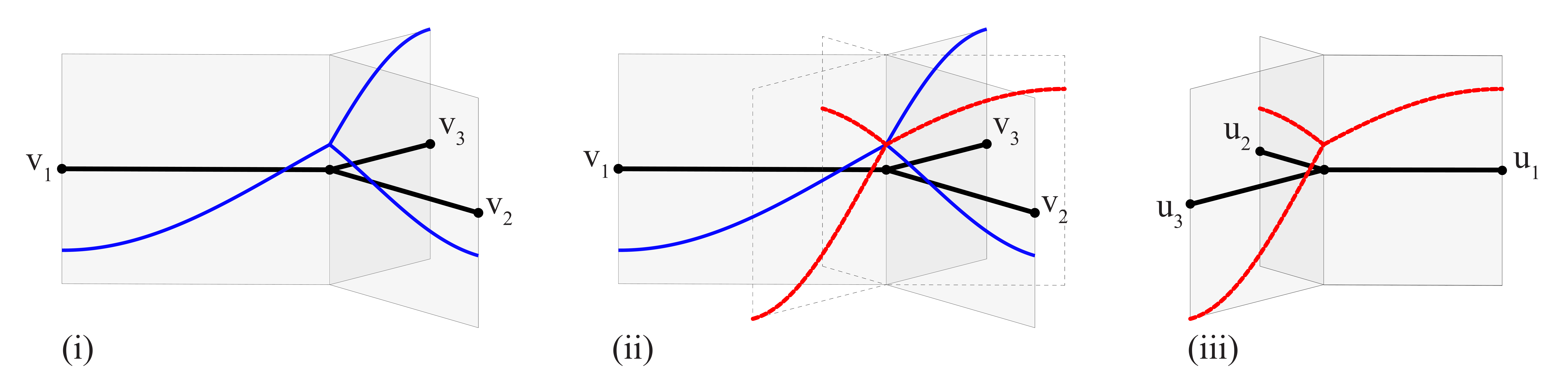}
\par\end{centering}
\centering{}\caption{(i) A star graph $\Omega$ with three edges (in black) and an eigenfunction
$f$ (in blue) whose derivative does not vanish in the interior of
$\Omega$, so that $\Omega$ is a single Neumann domain of $f$. (ii)
Adding the continuation of each of the restrictions $f|_{e_{j}}$
(in dashed red), which corresponds to $\tilde{f}|_{\tilde{e}_{j}}$
(in the proof of Lemma \ref{lem:dual-star}). (iii) The star graph
$\widetilde{\Omega}$ (in black) which is the dual of $\Omega$ and
the corresponding eigenfunction $\tilde{f}$ (in dashed red) of $\widetilde{\Omega}$.}
\label{fig: Dual graph example}
\end{figure}

Denote by $\left\{ l_{j}\right\} _{j=1}^{\left|\partial\Omega\right|}$
the edge lengths of the star graph $\Omega$. We take $\widetilde{\Omega}$
to be a star graph with the same number of edges, $\left|\partial\Omega\right|$,
and such that the edge lengths are taken to be $\tilde{l}_{j}:=\frac{\pi}{k}-l_{j}$.
So, every edge of $\Omega$ has a dual edge in $\widetilde{\Omega}$.
As usual, we consider $\widetilde{\Omega}$ to be a standard graph.
First, to justify that the construction above defines a valid graph
we need to show that $\tilde{l}_{j}>0$ for all $j$. This is equivalent
to showing that $l_{j}<\frac{\pi}{k}$ for all $j$. But this is clear,
since otherwise the derivative of $f$ would vanish somewhere within
the edge $e_{j}$, which violates the assumption of the lemma.

Note that by construction, statement (\ref{enu:lem-dual-star-1})
of the Lemma is satisfied and we need to verify that all other statements
(\ref{enu:lem-dual-star-2}),(\ref{enu:lem-dual-star-3}),(\ref{enu:lem-dual-star-4})
hold as well. To do so, note that $f$ may be written as
\[
\forall1\leq j\leq\left|\partial\Omega\right|,\quad\quad f|_{e_{j}}\left(x\right)=A_{j}\cos\left(k(l_{j}-x)\right),
\]
where $x=0$ at the central vertex, and the coefficients $A_{j}$
need to satisfy certain relations to ensure that $f$ satisfies Neumann
conditions at the central vertex. Define a function $\tilde{f}$ on
$\widetilde{\Omega}$ by
\begin{equation}
\forall1\leq j\leq\left|\partial\Omega\right|,\quad\quad\tilde{f}|_{\tilde{e}_{j}}\left(x\right)=-A_{j}\cos\left(k(\tilde{l}_{j}-x)\right),\label{eq:expressing-f-on-an-edge-of-Omega-tilde}
\end{equation}
where here as well $x=0$ at the central vertex. It is easy to see
that $\tilde{f}$ satisfies Neumann condition at all the boundary
vertices. At the central vertex we have
\begin{align}
\tilde{f}|_{\tilde{e}_{j}}\left(0\right) & =-A_{j}\cos\left(k\tilde{l}_{j}\right)=A_{j}\cos\left(kl_{j}\right)=f|_{e_{j}}\left(0\right)\nonumber \\
\tilde{f}'|_{\tilde{e}_{j}}\left(0\right) & =-kA_{j}\sin\left(k\tilde{l}_{j}\right)=-kA_{j}\sin\left(kl_{j}\right)=-f'|_{e_{j}}\left(0\right).\label{eq:values-of-f-related-to-values-of-dual-at-central-vertex}
\end{align}
Since $f$ satisfies Neumann boundary conditions at the central vertex
of $\Omega$ and we have the two relations above, we conclude that
$\tilde{f}$ satisfies Neumann vertex conditions at the central vertex
of $\widetilde{\Omega}$. Hence, $\tilde{f}$ satisfies Neumann vertex
conditions at all vertices of $\widetilde{\Omega}$ and it is therefore
an eigenfunction of the standard graph $\widetilde{\Omega}$ with
an eigenvalue $k$. Moreover, the relations (\ref{eq:values-of-f-related-to-values-of-dual-at-central-vertex})
together with the genericity of $f$ implies that $\tilde{f}$ satisfies
conditions (\ref{enu: def-generic-non-zero-value-at-vertices}) and
(\ref{enu:None-of-theenu: def-generic-non-zero-derivative-at-vertices})
of Definition \ref{def: generic_eigenfunction} which is enough to
conclude that $\tilde{f}$ is generic since $\widetilde{\Omega}$
is a tree (see Remark \ref{rem: generic for trees}). This proves
that statement (\ref{enu:lem-dual-star-2}) of the lemma holds.

By construction, the edge lengths of $\widetilde{\Omega}$ are bounded,
$\tilde{l}_{j}<\frac{\pi}{k}$. This together with (\ref{eq:expressing-f-on-an-edge-of-Omega-tilde})
shows that $\tilde{f'}$ vanishes only at the graph boundary vertices,
so that statement (\ref{enu:lem-dual-star-3}) of the lemma holds
as well.

Finally, to show statement (\ref{enu:lem-dual-star-4}) of the lemma,
we compute
\begin{align*}
\rho(\widetilde{\Omega}) & =\sum_{j=1}^{\left|\partial\Omega\right|}\frac{k}{\pi}\tilde{l_{j}}=\sum_{j=1}^{\left|\partial\Omega\right|}\frac{k}{\pi}\left(\frac{\pi}{k}-l_{j}\right)=\left|\partial\Omega\right|-\rho(\Omega).
\end{align*}
To compute the spectral position $N(\widetilde{\Omega})$ we observe
that since $\tilde{f}$ was shown to be generic eigenfunction and
since $\widetilde{\Omega}$ is a star graph with first Betti number
$\beta=0$, it follows (e.g., from (\ref{eq: nodal surplus bounds}))
that the spectral position $N(\widetilde{\Omega})$ equals the nodal
count of $\tilde{f}$. Now, since $\tilde{l}_{j}<\frac{\pi}{k}$ and
from (\ref{eq:expressing-f-on-an-edge-of-Omega-tilde}) the number
of nodal points on each edge $\tilde{e}_{j}$ is either zero (if $k\tilde{l}_{j}<\frac{\pi}{2}$)
or one (if $k\tilde{l}_{j}>\frac{\pi}{2}$). The genericity of $\tilde{f}$
implies $k\tilde{l}_{j}\neq\frac{\pi}{2}$ (as otherwise $f$ would
vanish at the central vertex). With those observations, we get
\begin{align*}
N(\widetilde{\Omega})= & \left|\left\{ 1\leq j\le\left|\partial\Omega\right|\,:\,k\tilde{l}_{j}>\frac{\pi}{2}\right\} \right|\\
= & \left|\left\{ 1\leq j\le\left|\partial\Omega\right|\,:\,kl_{j}<\frac{\pi}{2}\right\} \right|\\
= & \left|\partial\Omega\right|-N(\Omega).
\end{align*}
\end{proof}

\section{Review of existing tools for proving the probabilistic statements
\label{sec: tools_and_methods}}

\subsection{The characteristic torus and the secular manifold}

Let $\Gamma$ be a standard graph with $E=\left|\E\right|$ edges.
We will use the notation $\Gamma_{\lv}$ to emphasize the dependence
of $\Gamma$ on its edge lengths, $\lv\in(0,\infty)^{E}$. In particular,
we consider the standard graphs $\left\{ \Gamma_{\kv}\right\} _{\kv\in(0,2\pi]^{E}}$.
Each with the discrete graph structure of $\Gamma$ but with edge
lengths $\kv\in\opcl{0,2\pi}^{E}$.
\begin{defn}
Let $\Gamma$ be a graph with $E$ edges.
\begin{enumerate}
\item The flat torus $\torus:=\left(\R/2\pi\Z\right)^{E}$ is called the
characteristic torus of $\Gamma$. We consider the coordinates of
$\torus$ as taking values in $(0,2\pi]^{E}$. We also denote by $\modp{~~}:\R^{E}\rightarrow(0,2\pi]^{E}$
the projection to the torus by taking modulus $2\pi$ (with a slight
abuse of the usual modulus operator, as $2\pi$ is in its image here,
rather than $0$).
\item The following subsets of $\torus$,
\begin{align}
\Sigma & :=\set{\kv\in\torus}{1~\textrm{is an eigenvalue of }\Gamma_{\kv}}~~~\textrm{and}\nonumber \\
\mreg & :=\set{\kv\in\torus}{1~\textrm{is a simple eigenvalue of }\Gamma_{\kv}},\label{eq: def-of-sigma-and-sigma-reg}
\end{align}
 are called the secular manifold and its regular part.
\end{enumerate}
\end{defn}

We note that despite its suggestive name, $\Sigma$ is not necessarily
a smooth manifold, and it may have singularities. Nevertheless, the
set of its regular points is exactly $\mreg$, which is a real analytic
manifold of dimension $E-1$, see \cite{CdV_ahp15,AloBanBer_cmp18,Alon_PhDThesis}.

The following Lemma summarizes some results from \cite{CdV_ahp15,AloBanBer_cmp18,Alon_PhDThesis}.
The proofs of all sections of this Lemma appear in \cite{CdV_ahp15,AloBanBer_cmp18,Alon_PhDThesis}
either explicitly or between the lines. Nevertheless, for completeness
and didactic purpose we provide in Appendix \ref{sec: appendix=00005C-proof-of-two-lemmas}
a concise proof of the Lemma.
\begin{lem}
\label{lem: Canonical_eigenfunctions_and_Secular_mfld}\cite{CdV_ahp15,AloBanBer_cmp18,Alon_PhDThesis}
Let $\Gamma$ be a standard graph.
\begin{enumerate}
\item \label{enu: lem-Canonical_eigenfunctions_and_Secular_mfld-1} $k>0$
is a simple eigenvalue of $\Gamma_{\lv}$ if and only if $\modp{k\lv}\in\mreg$.
\item \label{enu: lem-Canonical_eigenfunctions_and_Secular_mfld-2} There
exists a family of functions $\left\{ f_{\kv}\right\} _{\kv\in\mreg}$
such that
\begin{enumerate}
\item \label{enu: lem-Canonical_eigenfunctions_and_Secular_mfld-2a}For
every $\kv\in\mreg$, $f_{\kv}$ is an eigenfunction of $\Gamma_{\kv}$
corresponding to the eigenvalue $1$.
\item \label{enu: lem-Canonical_eigenfunctions_and_Secular_mfld-2b}For
every $v,u\in\V$ and $e\in\E_{\V}$, there exist two real trigonometric
polynomials $p_{u,v}$ and $q_{u,v,e}$ such that for every $\kv\in\mreg$:
\begin{equation}
p_{u,v}(\kv)=f_{\kv}\left(u\right)\overline{f_{\kv}\left(v\right)},\text{ and}\label{eq: p_u_v_trig_poly}
\end{equation}
\begin{equation}
q_{u,v,e}(\kv)=f_{\kv}\left(u\right)\overline{\left(\d_{e}f_{\kv}\left(v\right)\right)}.\label{eq: q_u_v_e_trig_poly}
\end{equation}
\end{enumerate}
\item \label{enu: lem-Canonical_eigenfunctions_and_Secular_mfld-3} Let
$k>0$ and $\lv$ such that $\kv:=\modp{k\lv}\in\mreg$. Denote by
$f$ the real eigenfunction of $\Gamma_{\lv}$ which corresponds to
the eigenvalue $k$. There exists $c\in\C\setminus\left\{ 0\right\} $
such that for every $v\in\V$ and $e\in\E_{v}$,
\begin{align}
f\left(v\right) & =cf_{\kv}\left(v\right)\label{eq:value_of_eigenfunction_equals_to_canonical}\\
\frac{1}{k}\d_{e}f\left(v\right) & =c\d_{e}f_{\kv}\left(v\right).\label{eq:derivative_of_eigenfunction_equals_to_canonical}
\end{align}
Note that $c$ may depend on $\kv$.
\end{enumerate}
\end{lem}

We follow the terminology of \cite{AloBanBer_cmp18} and call $\left\{ f_{\kv}\right\} _{\kv\in\mreg}$
the canonical eigenfunctions of the graph $\Gamma$.

The Lemma above emphasizes the importance of the subset $\mreg$ when
dealing with simple eigenvalues. To specialize our discussion for
generic eigenfunctions we need to define the \emph{generic manifold},
\begin{align}
\mgen:= & \left\{ \kv\in\mreg\,|\,f_{\kv}\,\textrm{is\,generic}\right\} ,\label{eq: definition generic}
\end{align}
the properties of which are described in the following lemma.
\begin{lem}
\label{lem: generic secman} \cite[Thm 3.9]{AloBanBer_cmp18}, \cite{Alon}
\cite[Section 5]{Alon_PhDThesis}
\begin{enumerate}
\item \label{enu: lem-generic-secman-1} $\mgen$ is a real analytic sub-manifold
of $\mreg$ of dimension $\left|\E\right|-1$ and has finitely many
connected components.
\item \label{enu: lem-generic-secman-2} Let $f$ be an eigenfunction of
$\Gamma_{\lv}$ with eigenvalue $k$. Then\\
$f$ is generic $\quad\Leftrightarrow\quad$ $\kv:=\modp{k\lv}\in\Sigma_{\G}$
$\quad\Leftrightarrow\quad$$f_{\kv}$ is generic.
\end{enumerate}
\end{lem}

The proof of the first part of the Lemma can be found in \cite[Thm 3.9]{AloBanBer_cmp18},
\cite[Section 5]{Alon_PhDThesis} (note that our set $\mgen$ is different
than the one in \cite[Thm 3.9]{AloBanBer_cmp18}; yet the proof there
carries over to our case). The second part of the Lemma is a straightforward
implication of (\ref{eq:value_of_eigenfunction_equals_to_canonical}),
(\ref{eq:derivative_of_eigenfunction_equals_to_canonical}) and (\ref{eq: definition generic}).

\subsection{The Barra-Gaspard measure}

We start by introducing the following general framework.
\begin{defn}
\label{def: equidistributed_sequence}\cite[Definition 4.19]{einsiedler2013ergodic}
Let $X$ be a compact metric space and let $\mu$ be a Borel measure
on $X$. A sequence $\left\{ x_{n}\right\} _{n\in\N}$ of points in
$X$ is equidistributed according to $\mu$ if for any continuous
function $f\in C\left(X\right)$,
\begin{equation}
\lim_{N\rightarrow\infty}\frac{1}{N}\sum_{n=1}^{N}f(x_{n})=\int_{X}f\rmd\mu.\label{eq: equidistribution_definition}
\end{equation}
\end{defn}

\begin{lem}
\label{lem: Equidistribution-for_Riemann-integrable-and-Jordan-sets}\cite[Exercise 4.4.2]{einsiedler2013ergodic}
In Definition \ref{def: equidistributed_sequence}, the continuous
function $f$ may be replaced by a Riemann integrable function (i.e.,
a function whose discontinuity set is of $\mu$-measure zero).
\end{lem}

Next, we specialize the discussion to graphs and their secular manifolds.
\begin{defn}
Let $\Gamma$ be a standard graph with edge lengths $\lv$. Let $\left\{ k_{n}\right\} _{n=1}^{\infty}$
be the multi-set of the (square root of) eigenvalues of $\Gamma$,
where multiple eigenvalues appear more than once in this multi-set.
The map $\flow:\G\rightarrow\torus$ of the graph is defined as
\[
\flow(n):=\modp{k_{n}\lv}.
\]
\end{defn}

By Lemma \ref{lem: generic secman} we have that $\flow:\G\rightarrow\mgen$.
\begin{defn}
\label{def: BG-measure} Define a measure $\BGm$ on the closure of
the generic manifold, $\overline{\mgen}$ by
\begin{equation}
\rmd\BGm(\kv)=\begin{cases}
C\left|\lv\cdot\hat{n}\right|\rmd s & \kv\in\mgen\\
0 & \kv\in\partial\mgen,
\end{cases}\label{eq: BG-measure-definition}
\end{equation}
where $\rmd s$ is the volume element on $\mgen$ (which is an $E-1$
Riemannian manifold), $\hat{n}$ is the normal to $\mgen$, and $C=$$\left(\int_{\mgen}\left|\lv\cdot\hat{n}\right|\rmd s\right)^{-1}$
is a normalization constant\footnote{The normalization constant is computed explicitly in \cite{Alon,Alon_PhDThesis,AloBanBer_cmp18}
as part of the proof of Theorem \ref{thm: density-of-generic-and-loop-eigenfunctions}.
It is given by $\frac{1}{C}=\frac{\pi}{L}\frac{1}{\left(2\pi\right)^{E}}\left(1-\frac{L_{loops}}{2L}\right)$
and if $\lv$ is rationally independent then $\frac{1}{C}=\frac{\pi}{L}\frac{1}{\left(2\pi\right)^{E}}d\left(\G\right)$.}.
\end{defn}

Following \cite{CdV_ahp15} we call $\BGm$, the Barra-Gaspard measure.
\begin{thm}
\label{thm: equidistribution_by_BG_measure}\cite{BarGas_jsp00} \cite[Prop. 4.4]{BerWin_tams10}\cite[Lem. 3.1]{CdV_ahp15}Let
$\Gamma$ be a standard graph with rationally independent edge lengths
$\vec{l}$. Then
\end{thm}

\begin{enumerate}
\item \label{enu: thm-equidistribution_by_BG_measure-1} The sequence $\left\{ \flow(n)\right\} _{n\in\G}$
is equidistributed on $\overline{\mgen}$ with respect to the measure
$\BGm$.
\item \label{enu: thm-equidistribution_by_BG_measure-2} The measure $\BGm$
is an $\lv$ dependent smooth strictly positive regular Borel probability
measure on $\mgen$.
\end{enumerate}
\begin{rem}
In the references, \cite{BarGas_jsp00} \cite[Prop. 4.4]{BerWin_tams10}\cite[Lem. 3.1]{CdV_ahp15}
similar statements to the above appear for the manifolds $\Sigma$
and $\mreg$, with the measure extended to the larger manifold (and
normalized appropriately). Here, we restrict to $\overline{\mgen}$
as all statements of the current paper concern generic eigenfunctions
and those form a large enough venue for our explorations (see Theorem
\ref{thm: density-of-generic-and-loop-eigenfunctions}). We further
note that as the sequence $\left\{ \flow(n)\right\} _{n\in\G}$ is
contained in $\mgen$ and the measure $\BGm$ is supported on $\mgen$,
we employ the above theorem directly for $\mgen$ and not for its
closure $\overline{\mgen}$.
\end{rem}

\begin{defn}
We say that a set $\A\subset\overline{\mgen}$ has \emph{measure zero}
if it has zero Barra-Gaspard measure, $\mu_{\lv}\left(\A\right)=0$.
We say that $\A$ is \emph{Jordan }if its boundary $\partial\mathcal{A}\subset\overline{\mgen}$
is of measure zero.
\end{defn}

We do not specify for which $\lv$ as the above definitions are $\lv$
independent. To see that observe that for any $\lv$, $\mu_{\lv}$
has a strictly positive density on $\mgen$ (see (\ref{eq: BG-measure-definition})).
Hence, for any measurable $\A\subset\overline{\mgen}$,
\[
\mu_{\lv}\left(\A\right)=0\iff\int_{\mathcal{A}}\rmd s=0.
\]
 The observation that the indicator function $\chi_{\mathcal{A}}$
is Riemann integrable if and only if $\A$ is Jordan gives:
\begin{cor}
\label{cor: density_equals_BG_for_Jordan_set} Let $\Gamma$ be a
standard graph with rationally independent edge lengths $\vec{l}$.
Let $\A\subset\mgen$ be a Jordan set, then
\begin{equation}
d_{\G}\left(\set{n\in\G}{\flow(n)\in\A}\right)=\BGm\left(\A\right).\label{eq: density_equals_BG_for_Jordan_set}
\end{equation}
\end{cor}

It is easy to demonstrate that the corollary above does not hold for
sets whose boundary is not of measure zero. Take for example $\A=\cup_{n\in\G}\flow(n)$
which violates (\ref{eq: density_equals_BG_for_Jordan_set}), since
$d_{\G}\left(\set{n\in\G}{\flow(n)\in\A}\right)=1$, but $\mu_{\lv}\left(\A\right)=0$.

\subsection{Inversion map on the secular manifold\label{subsec: inversion}}

The following lemma is useful for proving that some probability distributions
are symmetric (Theorem \ref{thm:Neumann_surplus_main} and Propositions
\ref{prop: statistics_of_spectral_position} and \ref{prop: statistics_of_wavelength_capacity}).
A proof of a similar lemma is found in \cite{AloBanBer_cmp18}. Nevertheless,
we provide a concise proof of the lemma in Appendix \ref{sec: appendix=00005C-proof-of-two-lemmas}.
\begin{lem}
\label{lem: Inversion-properties} Let $\Gamma$ be a standard graph
with edge lengths $\lv$. Let $\mathcal{I}:\torus\rightarrow\torus$
be the inversion map of the torus, defined by $\mathcal{I}\left(\kv\right)=\modp{-\kv}$.
\begin{enumerate}
\item \label{enu: lem-Inversion-properties-1} Each of the manifolds, $\Sigma$,
$\mreg$ and $\mgen$ is invariant under the inversion map.
\item \label{enu: lem-Inversion-properties-3} The restriction of the inversion
to the generic part of the secular manifold, $\I|_{\mgen}$, preserves
the Barra-Gaspard measure, $\BGm$.
\item \label{enu: lem-Inversion-properties-2} For any $v,u\in\V$ and $e\in\E$,
there are $p_{u,v}$ and $q_{u,v,e}$ which satisfy the requirements
of Lemma \ref{lem: Canonical_eigenfunctions_and_Secular_mfld} (\ref{enu: lem-Canonical_eigenfunctions_and_Secular_mfld-2})
and has the following symmetry\textbackslash anti-symmetry relations:
\begin{align}
p_{u,v}\circ\I= & p_{u,v}\label{eq: inversion_of_ff}\\
q_{u,v,e}\circ\I= & -q_{u,v,e}.\label{eq: inversion_of_ff'}
\end{align}
\end{enumerate}
\end{lem}

\section{Developing further tools: functions on $\protect\mgen$ and random
variables on $\protect\G$ \label{sec: functions_on_secular_manifold}}

This section provides the needed tools for the proofs related to the
probability distributions of the observables discussed in the paper
(Neumann count, spectral position, wavelength capacity). In order
to do so we first relate those observables to functions on the secular
manifold (Lemma \ref{lem: functions_on_secular_manifold_existence_and_symmetry}).
Then, we provide lemmas which aid in determining whether those observables
may be considered as random variables (Lemma \ref{lem: random-variables})
or not (Lemma \ref{lem: random_variable_implies_countable_image}).

\subsection{Functions on the secular manifold}

The next theorem is an essential ingredient in the proofs of the main
results of \cite{AloBanBer_cmp18,Ban_ptrsa14}.
\begin{thm}
\cite{AloBanBer_cmp18}\label{thm: surplus_function_on_secular_mnfld}
Let $\Gamma$ be a graph with first Betti number $\beta$. There exists
a function $\boldsymbol{\sigma}$ on $\mgen$ with the following properties:
\begin{enumerate}
\item For any choice of edge lengths $\lv\in\left(0,\infty\right)^{E}$,
let $\Gamma_{\lv}$ be the corresponding standard graph with generic
index set $\G$. Then for any $n\in\G$,
\begin{equation}
\forall n\in\G,\quad\sigma(n)=\bs{\sigma}(\flow(n)).\label{eq: Nodal_surplus_on_secular_manifold}
\end{equation}
\item $\boldsymbol{\sigma}$ is constant on each connected component of
$\mgen$
\item $\boldsymbol{\sigma}$ is anti-symmetric with respect to the inversion
$\mathcal{I}$ in the following sense
\begin{equation}
\boldsymbol{\sigma}\left(\I(\kv)\right)=\beta-\boldsymbol{\sigma}\left(\kv\right).\label{eq: nodal surplus inversion}
\end{equation}
\end{enumerate}
\end{thm}

The next lemma is similar in spirit to Theorem \ref{thm: surplus_function_on_secular_mnfld}.
\begin{lem}
\label{lem: functions_on_secular_manifold_existence_and_symmetry}
Let $\Gamma$ be a graph with first Betti number $\beta$ and let
$v\in\V\backslash\partial\Gamma$ be an interior vertex of degree
$\deg v$. There exist functions $\bs{\omega}$, $\bs{N^{\left(v\right)}}$
and $\bs{\rho^{\left(v\right)}}$ on $\mgen$ with the following properties:
\begin{enumerate}
\item \label{enu: lem: functions_on_secular_manifold_existence_and_symmetry-1}
For any choice of edge lengths $\lv\in\left(0,\infty\right)^{E}$,
let $\Gamma_{\lv}$ be the corresponding standard graph with generic
index set $\G$. Then for any $n\in\G$,
\begin{equation}
\omega(n)=\bs{\omega}(\flow(n)).\label{eq: Neumann_surplus_on_secular_manifold}
\end{equation}
Moreover, denote the minimal edge length by $L_{min}$ and the sum
of edge lengths by $\left|\Gamma_{\lv}\right|$, and let $\{\ndv\}_{n\in\G}$
be the sequence of Neumann domains containing $v$ (see Definition
\ref{def: vertex-Neumann-domains}). Then for any $n\in\G$ such that
$n\ensuremath{\geq}2\frac{\left|\Gamma\right|}{L_{\textrm{min}}}$,
\begin{equation}
N(\ndv)=\bs{N^{(v)}}(\flow(n)),\,\text{ and}\label{eq: Spectral_position_on_secular_manifold}
\end{equation}
\begin{equation}
\rho(\ndv)=\bs{\rho^{(v)}}(\flow(n)).\label{eq: Wave_capacity_on_secular_manifold}
\end{equation}
\item \label{enu: lem: functions_on_secular_manifold_existence_and_symmetry-2}
The functions $\bs{\omega}$ and $\bs{N^{\left(v\right)}}$ are constant
on each connected component of $\mgen$.
\item \label{enu: lem: functions_on_secular_manifold_existence_and_symmetry-3}
$\bs{\rho^{(v)}}$ is real analytic on $\mgen$.
\item \label{enu: lem: functions_on_secular_manifold_existence_and_symmetry-4}
$\bs{\omega}$, $\bs{N^{\left(v\right)}}$ and $\bs{\rho^{\left(v\right)}}$
are anti-symmetric with respect to the inversion $\mathcal{I}$ in
the following sense
\begin{equation}
\bs{\omega}(\I(\kv))=\beta-\left|\partial\Gamma\right|-\bs{\omega}(\kv),\label{eq: anti-symmetry_of_Neumann_surplus_on_secular_mnfld}
\end{equation}
\begin{equation}
\bs{N^{(v)}}(\I(\kv))=\deg v-\bs{N^{(v)}}(\kv),\,\text{ and}\label{eq: anti-symmetry_of_spectral_position_on_secular_mnfld}
\end{equation}
\begin{equation}
\bs{\rho^{(v)}}(\I(\kv))=\deg v-\bs{\rho^{(v)}}(\kv).\label{eq: anti-symmetry_of_wavelength_capacity_on_secular_mnfld}
\end{equation}
\end{enumerate}
\end{lem}

\begin{proof}
The functions $p_{v,v}$ and $q_{v,v,e}$ (defined in Lemmas \ref{lem: Canonical_eigenfunctions_and_Secular_mfld}
and \ref{lem: Inversion-properties}) are used for expressing the
functions $\bs{\omega}$, $\bs{N^{\left(v\right)}}$ and $\bs{\rho^{\left(v\right)}}$
stated in the Lemma. Let us therefore recall some of their properties
that are stated in Lemmas \ref{lem: Canonical_eigenfunctions_and_Secular_mfld}
and \ref{lem: Inversion-properties}. For each $v\in\V\backslash\partial\Gamma$
and $e\in\E_{v}$, the functions $p_{v,v}$ and $q_{v,v,e}$ are real
trigonometric polynomials in $\kv$ that do not vanish on $\mgen$.
Their signs, $\sgn(p_{v,v})$ and $\sgn(q_{v,v,e})$, are therefore
constant on each connected component of $\mgen$. Moreover, $p_{v,v}$
and $q_{v,v,e}$ are symmetric\textbackslash anti-symmetric with
respect to $\I$. That is,
\begin{align}
p_{v,v}\circ\I & =p_{v,v},\,\text{and}\label{eq: p symmetry}\\
q_{v,v,e}\circ\I & =-q_{v,v,e}.\label{eq: q antisymmetry}
\end{align}
In addition, for any $\kv\in\mgen$, the canonical eigenfunction $f_{\kv}$
(i.e., the eigenfunction of $\Gamma_{\kv}$ with eigenvalue $k=1$)
is generic, and satisfies:
\begin{align}
p_{v,v}\left(\kv\right) & =\left|f_{\kv}\left(v\right)\right|^{2}\ne0,\,\text{and}\\
q_{v,v,e}\left(\kv\right) & =f_{\kv}\left(v\right)\overline{\d_{e}f_{\kv}\left(v\right)}\ne0.
\end{align}
Another useful observation is that since both $p_{v,v}$ and $q_{v,v,e}$
are real, continuous and non vanishing on $\mgen$.

The rest of the proof is divided into three parts, each treating one
of the functions $\bs{\omega}$, $\bs{N^{\left(v\right)}}$ and $\bs{\rho^{\left(v\right)}}$.

~\\
\uline{Proof for the Neumann surplus, \mbox{$\bs{\omega}$}}

We refer to the proof of the Theorem \ref{thm:Neumann_surplus_main},
(\ref{enu:thm-Neumann_surplus_bounds})\emph{ }and employ equation
(\ref{eq:Diff-nodal-Neumann_by_vertices}) to write

\begin{align}
\omega(n) & =\sigma(n)-\left(\phi(f_{n})-\mu(f_{n})\right)\nonumber \\
 & =\sigma(n)-\frac{\left|\d\Gamma\right|}{2}+\frac{1}{2}\sum_{v\in\V\setminus\d\Gamma}\sum_{e\in\E_{v}}\sgn\left(f_{n}\left(v\right)\overline{\d_{e}f_{n}\left(v\right)}\right).\label{eq:Neumann-surplus-as-sum-over-terms-1}
\end{align}

Now, we express all terms in the RHS of (\ref{eq:Neumann-surplus-as-sum-over-terms-1})
as functions on $\mgen$. By Theorem \ref{thm: surplus_function_on_secular_mnfld}
there exists a function $\bs{\sigma}$ on $\mgen$ such that $\sigma(n)=\bs{\sigma}(\flow(n))$
for all $n\in\G$ and the function $\bs{\sigma}$ is constant on each
connected component of $\mgen$. In addition, by (\ref{eq: q_u_v_e_trig_poly}),(\ref{eq:value_of_eigenfunction_equals_to_canonical})
and (\ref{eq:derivative_of_eigenfunction_equals_to_canonical}) in
Lemma \ref{lem: Canonical_eigenfunctions_and_Secular_mfld}
\begin{align}
\forall n\in\G,~\forall v\in\V,\,\forall e\in\E_{v},\quad\quad\,\,\frac{1}{k_{n}}f_{n}\left(v\right)\overline{\d_{e}f_{n}\left(v\right)} & =\left|c\right|^{2}f_{\flow(n)}(v)\overline{\partial_{e}f_{\flow(n)}(v)}\label{eq:function_at_vertex_equals_canonical_with_conjugation-1}\\
 & =\left|c\right|^{2}q_{v,v,e}\left(\flow(n)\right),
\end{align}
for some $c\in\C\setminus\left\{ 0\right\} $. Hence defining
\begin{equation}
\bs{\omega}(\kv):=\bs{\sigma}(\kv)-\frac{\left|\partial\Gamma\right|}{2}+\frac{1}{2}\sum_{v\in\V\setminus\d\Gamma}\sum_{e\in\E_{v}}\sgn\left(q_{v,v,e}\left(\kv\right)\right),\label{eq:Neumann-surplus-as-function-on-torus-1}
\end{equation}
we get that for all $n\in\G$, $\omega(n)=\bs{\omega}(\flow(n))$,
as required. Parts (\ref{enu: lem: functions_on_secular_manifold_existence_and_symmetry-2})
and (\ref{enu: lem: functions_on_secular_manifold_existence_and_symmetry-4})
of the Lemma (for $\bs{\omega}$) follow easily. Indeed, $\bs{\omega}$
is constant on each connected component of $\mgen$, as both $\bs{\sigma}$
and $\sgn\left(q_{v,v,e}\right)$ have this property. That $\bs{\omega}$
attains only finitely many values follows from (\ref{eq:Neumann-surplus-as-function-on-torus-1})
together with the statement that $\bs{\sigma}$ attains finitely many
values by Theorem \ref{thm: surplus_function_on_secular_mnfld}. In
addition, the anti-symmetric property of $\bs{\omega}$, (\ref{eq: anti-symmetry_of_Neumann_surplus_on_secular_mnfld}),
follows immediately from the anti-symmetry of $q_{v,v,e}$, (\ref{eq: inversion_of_ff'})
and of $\boldsymbol{\sigma}$ (\ref{eq: nodal surplus inversion}).

~\\
\uline{Proof for the spectral position, \mbox{$\bs{N^{(v)}}$}}

The $n^{\textrm{th}}$ eigenvalue is bounded, $k_{n}\geq\frac{\pi}{2\left|\Gamma\right|}(n+1)$,
\cite[Theorem 1]{Fri_aif05}. Hence, assuming $n\ensuremath{\geq}2\frac{\left|\Gamma\right|}{L_{\textrm{min}}}$
as in (\ref{eq: Spectral_position_on_secular_manifold}), we get that
$k_{n}>\frac{\pi}{L_{\textrm{min}}}$. This, together with $n\in\G$
are the required assumptions in Lemma \ref{lem:spectral-position-equals-nodal-count}.
That lemma guarantees that $\ndv$ is a star graph, that $\left.f_{n}\right|_{\ndv}$
is a generic eigenfunction and that $N(\ndv)=\phi(\left.f_{n}\right|_{\ndv})$.

To do so, we employ equation (\ref{eq:Diff-nodal-Neumann_by_vertices}),
taking the graph to be $\ndv$, and getting
\begin{align}
\phi(\left.f_{n}\right|_{\ndv}) & =\frac{1}{2}\deg v-\frac{1}{2}\sum_{e\in\E_{v}}\sgn\left(f_{n}\left(v\right)\overline{\d_{e}f_{n}\left(v\right)}\right)\label{eq: N of Omega as sum of signs-1}\\
 & =\frac{1}{2}\deg v-\frac{1}{2}\sum_{e\in\E_{v}}\sgn\left(f_{\flow(n)}(v)\overline{\partial_{e}f_{\flow(n)}(v)}\right)\\
 & =\frac{1}{2}\deg v-\frac{1}{2}\sum_{e\in\E_{v}}\sgn\left(q_{v,v,e}\left(\flow(n)\right)\right),
\end{align}
where moving to the second line we used (\ref{eq:value_of_eigenfunction_equals_to_canonical})
and (\ref{eq:derivative_of_eigenfunction_equals_to_canonical}) in
Lemma \ref{lem: Canonical_eigenfunctions_and_Secular_mfld} and the
last line follows from (\ref{eq: q_u_v_e_trig_poly}). This motivates
us to define $\bs{N^{\left(v\right)}}:\mgen\rightarrow\N$ by
\begin{equation}
\bs{N^{\left(v\right)}}\left(\kv\right):=\frac{1}{2}\deg v-\frac{1}{2}\sum_{e\in\E_{v}}\sgn\left(q_{v,v,e}\left(\kv\right)\right).\label{eq:definition of Nv-1}
\end{equation}
From all the arguments above, we obtain that $N(\ndv)=\phi(\left.f_{n}\right|_{\ndv})=\bs{N^{\left(v\right)}}(\flow(n))$,
for any $n\in\G$ such that $n\ensuremath{\geq}2\frac{\left|\Gamma\right|}{L_{\textrm{min}}}$.
Parts (\ref{enu: lem: functions_on_secular_manifold_existence_and_symmetry-2})
and (\ref{enu: lem: functions_on_secular_manifold_existence_and_symmetry-4})
of the Lemma (for $\bs{N^{\left(v\right)}}$) now follow easily. Indeed,
$\bs{N^{\left(v\right)}}$ is constant on connected components of
$\mgen$, as $\sgn\left(q_{v,v,e}\right)$ have this property. That
$\bs{N^{\left(v\right)}}$ attains only finitely many values follows
immediately from (\ref{eq:definition of Nv-1}). In addition, (\ref{eq: anti-symmetry_of_spectral_position_on_secular_mnfld}),
the anti-symmetric property of $\bs{N^{(v)}}$ follows from the anti-symmetry
of $q_{v,v,e}$, (\ref{eq: inversion_of_ff'}).

~\\

\uline{Proof for the wavelength capacity, \mbox{$\bs{\rho^{(v)}}$}}

Similarly to the previous part of the proof, we get that the assumption
$n\in\G$ and $n\ensuremath{\geq}2\frac{\left|\Gamma\right|}{L_{\textrm{min}}}$
in (\ref{eq: Wave_capacity_on_secular_manifold}) implies that $\ndv$
is a star graph and that $\left.f_{n}\right|_{\ndv}$ is a generic
eigenfunction. We denote the edges of $\Omega_{n}^{(v)}$ by $\{\tilde{e}_{j}\}_{j=1}^{\deg v}$
and the corresponding edge lengths by $\{\tilde{l}_{j}\}_{j=1}^{\deg v}$.
With this notation we may write the eigenfunction $\left.f_{n}\right|_{\ndv}$
as
\begin{equation}
\left.f_{n}\right|_{\tilde{e}_{j}}\left(x\right)=C_{j}\cos\left(k_{n}(\tilde{l}_{j}-x)\right),\label{eq:expressing-f-on-an-edge-of-Omega-2-1}
\end{equation}
for some real coefficients $C_{j}$ and using the arc-length parametrization
$x\in[0,\tilde{l}_{j}]$ with $x=0$ at the central vertex. The genericity
of $f_{n}|_{\ndv}$ implies that neither its value nor its derivative
vanish at the central vertex. Since $\Omega_{n}^{(v)}$ is a Neumann
domain, the derivative of $\left.f_{n}\right|_{\ndv}$ does not vanish
at the interior of the edges, $\tilde{e}_{j}$. It vanishes only at
the boundary vertices. These restrictions on values and derivatives
of $\left.f_{n}\right|_{\tilde{e}_{j}}$ together with (\ref{eq:expressing-f-on-an-edge-of-Omega-2-1})
imply that for all $j$, $k_{n}\tilde{l}_{j}\in\left(0,\frac{\pi}{2}\right)\cup\left(\frac{\pi}{2},\pi\right)$.
Hence, using (\ref{eq:expressing-f-on-an-edge-of-Omega-2-1}) we write
\begin{align}
\tan\left(k_{n}\tilde{l}_{j}\right) & =\left.\frac{\left.f_{n}\right|_{\tilde{e}_{j}}\cdot\frac{1}{k_{n}}\left.\partial_{\tilde{e}_{j}}f_{n}\right|_{\tilde{e}_{j}}}{\left.f_{n}^{2}\right|_{\tilde{e}_{j}}}\right|_{x=0}\nonumber \\
 & =\frac{f_{\flow(n)}(v)\cdot\overline{\partial_{e_{j}}f_{\flow(n)}(v)}}{f_{\flow(n)}(v)\overline{f_{\flow(n)}(v)}}\nonumber \\
 & =\frac{q_{v,v,e_{j}}\left(\flow(n)\right)}{p_{v,v}\left(\flow(n)\right)},\label{eq: tan_kl}
\end{align}
where moving to the second line we used (\ref{eq:value_of_eigenfunction_equals_to_canonical})
and (\ref{eq:derivative_of_eigenfunction_equals_to_canonical}) in
Lemma \ref{lem: Canonical_eigenfunctions_and_Secular_mfld}, and also
identified the edge $\tilde{e}_{j}$ of $\ndv$ as a subset of an
edge $e_{j}$ on the whole graph, so that $\partial_{\tilde{e}_{j}}=\partial_{e_{j}}$.
The the last line in (\ref{eq: tan_kl}) follows from (\ref{eq: p_u_v_trig_poly}),(\ref{eq: q_u_v_e_trig_poly}).
Next, to express $k_{n}\tilde{l}_{j}$, we use the following branch
of the inverse tangent, $\tan^{-1}:\R\backslash\left\{ 0\right\} \rightarrow\left(0,\frac{\pi}{2}\right)\cup\left(\frac{\pi}{2},\pi\right)$,
which is valid since we have shown above that $k_{n}\tilde{l}_{j}\in\left(0,\frac{\pi}{2}\right)\cup\left(\frac{\pi}{2},\pi\right)$.
Summing over all edges of the star graph $\ndv$ gives
\begin{align*}
\rho(\ndv) & =\frac{1}{\pi}\sum_{j=1}^{\deg v}k_{n}\tilde{l}_{j}\\
 & =\frac{1}{\pi}\sum_{j=1}^{\deg v}\tan^{-1}\left(\frac{q_{v,v,e_{j}}\left(\flow(n)\right)}{p_{v,v}\left(\flow(n)\right)}\right).
\end{align*}
Hence, to obtain the required relation (\ref{eq: Wave_capacity_on_secular_manifold}),
we define the wavelength capacity function $\bs{\rho^{(v)}}:\mgen\rightarrow\R$
as
\begin{equation}
\bs{\rho^{(v)}}(\kv):=\frac{1}{\pi}\sum_{j=1}^{\deg v}\tan^{-1}\left(\frac{q_{v,v,e_{j}}\left(\kv\right)}{p_{v,v}\left(\kv\right)}\right).\label{eq: rho_as_function_of_kappa}
\end{equation}
Parts (\ref{enu: lem: functions_on_secular_manifold_existence_and_symmetry-3})
and (\ref{enu: lem: functions_on_secular_manifold_existence_and_symmetry-4})
of the Lemma (for $\bs{\rho^{\left(v\right)}}$) now follow. To show
that $\Rvb$ is real analytic, we first notice that $q_{v,v,e_{j}}$
and $p_{v,v}$ are real trigonometric polynomials (on $\torus$) and
as such are real analytic on $\torus$. With the aid of the real analytic
version of the implicit function theorem \cite{KraPar_real_analytic_funcs_book}
we get that the restrictions of $q_{v,v,e_{j}}$ and $p_{v,v}$ to
the real analytic manifold $\mgen$ are real analytic, and so does
$\frac{q_{v,v,e_{j}}}{p_{v,v}}$, as $p_{v,v}$ does not vanish on
$\mgen$. Since $\tan^{-1}$ is also real analytic, we get that $\Rvb$
is real analytic on $\mreg$.

In addition, the anti-symmetric property of $\bs{\rho^{(v)}},$(\ref{eq: anti-symmetry_of_wavelength_capacity_on_secular_mnfld}),
follows straightforwardly by combining the anti-symmetry of $q_{v,v,e}$,
(\ref{eq: q antisymmetry}), the symmetry of $p_{v,v}$, (\ref{eq: p symmetry}),
and the anti-symmetry of the branch of $\tan^{-1}$ we have chosen,
$\tan^{-1}(-x)=\pi-\tan^{-1}(x)$.
\end{proof}

\subsection{Observables as random variables}

This paper concerns various functions (observables) such as the nodal
surplus, $\sigma$, the Neumann surplus, $\omega$, the spectral position,
$\Nv$ and the wavelength capacity, $\Rv$. All those functions are
defined on the generic index set, $\G$. In this section we provide
the tools needed to show whether such a function may be considered
as a random variable with respect to $d_{\G}$ (see Definition \ref{def: random variable}).
To this end, we need to achieve two tasks; first present a well-defined
probability space on $\G$ with the probability measure $d_{\G}$;
and second, to show that the considered function is measurable. Even
the first task is non-trivial, as the density $d_{\G}$ is not countably
additive on the power set of $\G$ and we should carefully choose
our $\sigma$-algebra on $\G$ to ensure this.

\begin{lem}
\label{lem: random-variables} Let $\Gamma$ be a standard graph with
rationally independent edge lengths $\vec{l}$. Let $\bs{\alpha}:\mgen\rightarrow\R$
and $\alpha:\mathcal{G}\rightarrow\R$ such that
\begin{equation}
\forall n\in\G,\quad\alpha(n):=\bs{\alpha}\left(\flow(n)\right).\label{eq: definition_of_alpha_on_G-1}
\end{equation}
If each element of the $\sigma$-algebra generated by $\bs{\alpha}$
is a Jordan set (i.e., having boundary of measure zero) then $\alpha$
is a random variable with respect to $d_{\G}$.\\
Moreover the probability distribution of $\alpha$ is given by
\begin{equation}
\forall j\in\mathrm{Image}\left(\alpha\right)\,\,\,d_{\mathcal{G}}\left(\alpha^{-1}\left(j\right)\right)=\mu_{\lv}\left(\boldsymbol{\alpha}^{-1}\left(j\right)\right),\label{eq: density_equals_BG_measure_for_level_set-1}
\end{equation}
where $\text{\ensuremath{\BGm}}$ is the Barra-Gaspard measure (\ref{eq: BG-measure-definition}).
\end{lem}

\begin{proof}
Let $\F_{\alpha}$ be the $\sigma$-algebra on $\G$ generated by
$\alpha$, and let $\boldsymbol{\F}_{\boldsymbol{\boldsymbol{\alpha}}}$
be the $\sigma$-algebra on $\mgen$ generated by $\boldsymbol{\alpha}$.
These can be described by $\F_{\alpha}=\set{\alpha^{-1}\left(B\right)}{B\subset\R\,\text{ is Borel}}$
and similarly for $\boldsymbol{\F_{\boldsymbol{\alpha}}}$. Let $B\subset\R$
be a Borel set and consider the two sets, $\alpha^{-1}\left(B\right)\subset\G$
and $\boldsymbol{\alpha}^{-1}\left(B\right)\subset\mgen$, so that
$\alpha^{-1}\left(B\right)=\set{n\in\G}{\flow(n)\in\boldsymbol{\alpha}^{-1}\left(B\right)}$.
By the assumption, $\boldsymbol{\alpha}^{-1}\left(B\right)$ is Jordan
since it is an element in $\boldsymbol{\F_{\boldsymbol{\alpha}}}$.
Hence, by Corollary \ref{cor: density_equals_BG_for_Jordan_set}:

\begin{equation}
\mu_{\lv}\left(\boldsymbol{\alpha}^{-1}\left(B\right)\right)=d_{\G}\left(\alpha^{-1}\left(B\right)\right).\label{eq: alpha preimage B}
\end{equation}
This proves (\ref{eq: density_equals_BG_measure_for_level_set-1})
and that every element in $\F_{\alpha}$ has density. In order to
conclude that $d_{\G}$ is a probability measure on $\left(\G,\F_{\alpha}\right)$
we are left with showing that $d_{\G}$ is $\sigma$-additive on $\F_{\alpha}$
and $d_{\G}\left(\G\right)=1$. Since the image of $\alpha$ is countable,
then $d_{\G}$ is $\sigma$-additive on $\F_{\alpha}$ if for any
subset $J\subset\mathrm{Image}\left(\alpha\right)$,
\[
d_{\G}\left(\cup_{j\in J}\alpha^{-1}\left(j\right)\right)=\sum_{j\in J}d_{\G}\left(\alpha^{-1}\left(j\right)\right).
\]
This follows from (\ref{eq: alpha preimage B}), using $\alpha^{-1}\left(J\right)=\cup_{j\in J}\alpha^{-1}\left(j\right)$
and the fact that $\mu_{\lv}$ is a measure (hence $\sigma$-additive):
\[
d_{\G}\left(\alpha^{-1}\left(J\right)\right)=\mu_{\lv}\left(\boldsymbol{\alpha}^{-1}\left(J\right)\right)=\sum_{j\in J}\mu_{\lv}\left(\boldsymbol{\alpha}^{-1}\left(j\right)\right)=\sum_{j\in J}d_{\G}\left(\alpha^{-1}\left(j\right)\right).
\]
Applying (\ref{eq: alpha preimage B}) to $B=\R$ gives:
\[
d_{\G}\left(\G\right)=d_{\G}\left(\alpha^{-1}\left(\R\right)\right)=\mu_{\lv}\left(\boldsymbol{\alpha}^{-1}\left(\R\right)\right)=\mu_{\lv}\left(\mgen\right)=1.
\]
Therefore, $d_{\G}$ is a probability measure on $\left(\G,\F_{\alpha}\right)$,
and so $\alpha$ is a random variable with respect to $d_{\G}$ (see
Definition \ref{def: random variable}).
\end{proof}
The Lemma above is used in order to show that all the observables
discussed in the paper, with the exception of $\Rv$, are random variables
(see proofs for Theorem \ref{thm:Neumann_surplus_main} and Proposition
\ref{prop: statistics_of_spectral_position}). In the other direction,
the next lemma aids in showing that there are observables (we believe
that $\Rv$ is such) which cannot be considered as random variables.
\begin{lem}
\label{lem: random_variable_implies_countable_image}Let $\Gamma$
be a standard graph with rationally independent edge lengths $\vec{l}$.
Let $\bs{\alpha}:\mgen\rightarrow\R$ be a Riemann integrable function
(i.e., having discontinuity set of measure zero), and define $\alpha:\mathcal{G}\rightarrow\R$
by
\[
\forall n\in\G,\quad\alpha(n):=\bs{\alpha}\left(\flow(n)\right).
\]
If $\alpha$ is a random variable with respect to $d_{\mathcal{G}}$
then there exists $X\subset\Sigma^{gen}$ of full measure, $\mu_{\lv}\left(X\right)=1$,
such that $\mathrm{Image}\left(\boldsymbol{\alpha}|_{X}\right)$ is
countable.
\end{lem}

\begin{rem}
\label{rem: wavelength-capacity-not-random-variable} As an immediate
corollary from the Lemma above we deduce that if $\bs{\alpha}:\mgen\rightarrow\R$
is continuous and non-constant on some connected open set, then $\alpha$
is not a random variable. Due to this we cannot generally regard the
wavelength capacity, $\Rv$, as a random variable. For the wavelength
capacity to be considered as a random variable, we need that the corresponding
function on the secular manifold, $\bs{\Rv}$ be constant on each
connected component of $\mgen$. This is unlikely by (\ref{eq: rho_as_function_of_kappa})
(see also numerical findings in Section \ref{sec:Discussion}) and
we even make a stronger conjecture that there is no open set of $\mgen$
at which $\bs{\Rv}$ is constant (see Remark \ref{rem: no_atoms_in_rho_distribution}).
\end{rem}

\begin{proof}
[Proof of Lemma \ref{lem: random_variable_implies_countable_image}]

Observe that the image of $\alpha$ is countable and that for any
value $j\in\imag\left(\alpha\right)$, the set $\bs{\alpha}^{-1}(j)\subset\mreg$
is measurable since $\bs{\alpha}$ is a measurable function (it is
even Riemann integrable). Define $X:=\sqcup_{j\in\imag(\alpha)}\bs{\alpha}^{-1}(j)$.
We will show that for all $j\in\imag(\alpha)$, $\BGm(\bs{\alpha}^{-1}(j))\geq d_{\G}(\alpha^{-1}(j))$
and conclude
\begin{align}
\BGm(X)=\sum_{j\in\imag(\alpha)}\BGm(\bs{\alpha}^{-1}(j)) & \geq\sum_{j\in\imag(\alpha)}d_{\G}(\alpha^{-1}(j))\nonumber \\
 & =d_{\G}(\sqcup_{j\in\imag(\alpha)}\alpha^{-1}(j))=d_{\G}(\G)=1,\label{eq: mu_tilde_X-decomposed-1}
\end{align}
where the first equality follows since $\mu_{\lv}$ is a measure and
$X=\sqcup_{j\in\imag(\alpha)}\bs{\alpha}^{-1}(j)$ is a disjoint countable
union; and all equalities on the second lines are by assumption of
the lemma that $d_{\G}$ is a probability measure. The statement of
the lemma now follows from (\ref{eq: mu_tilde_X-decomposed-1}) and
$\BGm(X)\le\BGm(\mgen)=1$. It is left to finish the proof by showing
that for all $j\in\imag(\alpha)$, $\BGm(\bs{\alpha}^{-1}(j))\geq d_{\G}(\alpha^{-1}(j))$,
which we do next.

Let $j\in\imag(\alpha)$ and denote the closure of $\bs{\alpha}^{-1}(j)$
by $\boldsymbol{\A}=\overline{\bs{\alpha}^{-1}(j)}$. Notice that
the points in $\boldsymbol{\A}\setminus\bs{\alpha}^{-1}(j)$ are discontinuity
points of $\boldsymbol{\alpha}$. Since $\bs{\alpha}$ was assumed
to be Riemann integrable then $\boldsymbol{\A}\setminus\bs{\alpha}^{-1}(j)$
is of measure zero and so
\begin{equation}
\BGm(\bs{\alpha}^{-1}(j))=\BGm(\boldsymbol{\A}).\label{eq: measure of preimage and its closure}
\end{equation}
Using that $\BGm$ is a regular measure (Theorem \ref{thm: equidistribution_by_BG_measure})
we have that for every $\varepsilon>0$, there exist an open set $U_{\varepsilon}$
such that $\bs{\A}\subset U_{\varepsilon}$ and $\BGm\left(U_{\varepsilon}\backslash\bs{\A}\right)<\varepsilon$.
By Urysohn's Lemma there exist a continuous function $f_{\varepsilon}:\mgen\rightarrow\left[0,1\right]$
supported inside $U_{\varepsilon}$ and such that $f_{\varepsilon}|_{\bs{\A}}\equiv1$.
By this construction,
\begin{equation}
\forall n\in\alpha^{-1}(j),\qquad f_{\varepsilon}(\flow(n))=1.\label{eq: f epsilon}
\end{equation}
Denoting $\G(N):=\set{n\in\G}{n\leq N}$, we get
\begin{align}
d_{\G}(\alpha^{-1}(j)) & =\lim_{N\rightarrow\infty}\frac{\left|\set{n\in\G(N)}{n\in\alpha^{-1}(j)}\right|}{\left|\G(N)\right|}\nonumber \\
 & \leq\lim_{N\rightarrow\infty}\frac{1}{\left|\G(N)\right|}\sum_{n\in\G(N)}f_{\varepsilon}(\flow(n))\nonumber \\
 & =\int_{\mgen}f_{\varepsilon}\rmd\BGm\nonumber \\
 & =\int_{\bs{\A}}f_{\varepsilon}\rmd\BGm+\int_{U_{\varepsilon}\backslash\bs{\A}}f_{\varepsilon}\rmd\BGm\nonumber \\
 & <\BGm\left(\bs{\A}\right)+\varepsilon=\BGm\left(\bs{\alpha}^{-1}(j)\right)+\varepsilon,\label{eq: density_smaller_than_BGm}
\end{align}
where moving to the third line we use that $\left\{ \flow(n)\right\} _{n\in\G}$
is an equidistributed sequence and $f_{\varepsilon}$ is continuous,
and the last equality is due to (\ref{eq: measure of preimage and its closure}).
As (\ref{eq: density_smaller_than_BGm}) holds for every $\varepsilon>0$,
we get the required inequality $\BGm(\bs{\alpha}^{-1}(j))\geq d_{\G}(\alpha^{-1}(j))$,
which finishes the proof.
\end{proof}
\begin{rem*}
We note that the results of the current subsection (Lemmata \ref{lem: random-variables}
and \ref{lem: random_variable_implies_countable_image}) may be similarly
proved for a general compact metric space, a measure on it and a corresponding
equidistributed sequence. Also, the random variables discussed above
are scalars, but the same statements hold for random vector variables
whose range of definition is $\R^{n}$ (for any $n\in\N$).
\end{rem*}

\section{Probability distributions of the Neumann count and the spectral position\protect \\
(proofs of Theorem \ref{thm:Neumann_surplus_main},(\ref{enu:thm-density_of_Neumann_surplus}),
Theorem \ref{Thm: (3,1)-regular-trees} and Proposition \ref{prop: statistics_of_spectral_position}).
\label{sec: proofs_Neumann_count_and_spectral_position}}

\subsection{Proving the existence and symmetry of the probability distributions
of $\omega$ and $N^{\left(v\right)}$\protect \\
(Theorem \ref{thm:Neumann_surplus_main},(\ref{enu:thm-density_of_Neumann_surplus})
and Proposition \ref{prop: statistics_of_spectral_position}).}

\vspace{2mm}

\begin{proof}
[Proof of Theorem \emph{ }\ref{thm:Neumann_surplus_main},(\ref{enu:thm-density_of_Neumann_surplus})]

We wish to apply Lemma \ref{lem: random-variables} in order to show
that the Neumann surplus, $\omega:\G\rightarrow\Z$, is a random variable
with respect to $d_{\G}$. For this purpose, we use the function $\bs{\omega}:\mgen\rightarrow\Z$,
whose existence and properties were established in Lemma \ref{lem: functions_on_secular_manifold_existence_and_symmetry}
and in particular, $\omega(n)=\bs{\omega}(\flow(n))$ for all $n\in\G$.
By Lemma \ref{lem: functions_on_secular_manifold_existence_and_symmetry},
the function $\bs{\omega}$ is constant on connected components of
$\mgen$. Let $\F_{\boldsymbol{\omega}}$ be the $\sigma$-algebra
generated by $\bs{\omega}$. Then each element in $\F_{\boldsymbol{\omega}}$
is a union of connected components of $\mgen$. $\mgen$ has a finite
number of connected components, by Lemma \ref{lem: generic secman},
and each of them is both open and close since $\mgen$ is locally
connected. Hence, each element in $\F_{\boldsymbol{\omega}}$ is also
open and closed and as such has no boundary and is Jordan. This is
exactly the condition in Lemma \ref{lem: random-variables}, by which
we get that $\omega$ is a random variable with respect to $d_{\G}$.
Furthermore, $\omega$ is a finite random variable. Indeed, by the
above $\imag(\bs{\omega})$ is finite and so $\imag(\omega)$ is finite
as well\footnote{This was also proven in Theorem \ref{thm:Neumann_surplus_main},(\ref{enu:thm-Neumann_surplus_bounds}).}.
This completes the proof of part (\ref{enu:thm-density_of_Neumann_surplus-a})
of the theorem.

To prove the next two parts of the theorem note that another implication
of Lemma \ref{lem: random-variables} is
\begin{equation}
\forall j\in\imag(\bs{\omega}),\quad d_{\G}(\omega^{-1}(j))=\BGm(\bs{\omega}^{-1}(j)).\label{eq: density-equals-BGm-Neumann-surplus}
\end{equation}

Let $j$ be such that $\omega^{-1}(j)\neq\emptyset$, so $\boldsymbol{\omega}^{-1}\left(j\right)\ne\emptyset$.
As argued above, $\bs{\omega}^{-1}(j)$ is an open set. This together
with $\BGm$ being strictly positive (Theorem \ref{thm: equidistribution_by_BG_measure},(\ref{enu: thm-equidistribution_by_BG_measure-2}))
yields $\BGm(\bs{\omega}^{-1}(j))>0$. By (\ref{eq: density-equals-BGm-Neumann-surplus}),
this proves part (\ref{enu:thm-density_of_Neumann_surplus-b}) of
the theorem.

Finally, to prove part (\ref{enu:thm-density_of_Neumann_surplus-c})
of the theorem, we use the inversion map, $\I:\kv\mapsto\modp{-\kv}$
which acts on $\mgen$. Since by Lemma \ref{lem: Inversion-properties},(\ref{enu: lem-Inversion-properties-3}),
$\I$ preserves the measure $\BGm$ we have

\begin{align}
\forall j\quad\BGm(\bs{\omega}^{-1}(j)) & =\BGm(\I(\bs{\omega}^{-1}(j)))\nonumber \\
 & =\BGm(\left(\bs{\omega}\circ\I\right)^{-1}(j))\nonumber \\
 & =\BGm(\bs{\omega}^{-1}(\beta-\left|\partial\Gamma\right|-j)),\label{eq: symmetry_of_BG_measure_on_Neumann_surplus}
\end{align}
where in the second line we used that $\I$ is an involution and the
third line is obtained from (\ref{eq: anti-symmetry_of_Neumann_surplus_on_secular_mnfld}),
which reads $\left(\bs{\omega}\circ\I\right)(\kv)=\beta-\left|\partial\Gamma\right|-\bs{\omega}(\kv)$.
Combining (\ref{eq: symmetry_of_BG_measure_on_Neumann_surplus}) with
(\ref{eq: density-equals-BGm-Neumann-surplus}) gives the desired
symmetry of the probability distribution.
\end{proof}
The next proof is quite similar to the previous one. The only essential
difference\footnote{This difference is due to the special role which is played by a finite
number of eigenvalues appearing in the beginning of the spectrum -
see details within the proof.} is in introducing an auxiliary random variable with respect to $d_{\G}$.
\begin{proof}
[Proof of Proposition \ref{prop: statistics_of_spectral_position}]

We start by recalling the function $\bs{N^{(v)}}:\mgen\rightarrow\N$,
whose existence and properties were established in Lemma \ref{lem: functions_on_secular_manifold_existence_and_symmetry}
and introduce $\widetilde{N}^{(v)}:\G\rightarrow\N$ by
\begin{equation}
\forall n\in\G,\quad\widetilde{N}^{(v)}(n)=\bs{N^{(v)}}(\flow(n)).\label{eq: spectral_position_auxiliary_func}
\end{equation}
Following exactly the same arguments as in the beginning of the preceding
proof (proof of Theorem \ref{thm:Neumann_surplus_main},(\ref{enu:thm-density_of_Neumann_surplus})),
we get that $\widetilde{N}^{(v)}$ is a finite random variable with
respect to $d_{\G}$. We need to obtain a similar statement for $N^{(v)}$.
To do so, we note that by (\ref{eq: Spectral_position_on_secular_manifold})
and (\ref{eq: spectral_position_auxiliary_func})
\begin{equation}
\forall n\in\G~\textrm{s.t. }n\geq2\frac{\left|\Gamma\right|}{L_{\textrm{min}}},\quad N^{(v)}(n)=\widetilde{N}^{(v)}(n),\label{eq: spec_pos_equals_spec_pos_auxiliary}
\end{equation}
where $L_{\mathrm{min}}$ is the minimal edge length of $\Gamma$
and $\left|\Gamma\right|$ is the sum of all edge lengths. Thanks
to (\ref{eq: spec_pos_equals_spec_pos_auxiliary}), for any Borel
set $B\subset\R$ the pre-images $\left(N^{(v)}\right)^{-1}(B)$ and
$\left(\widetilde{N}^{(v)}\right)^{-1}(B)$ differ by a finite number
of elements from $\Gsing$. Since the density of $\left(\widetilde{N}^{(v)}\right)^{-1}(B)$
exists (because $\widetilde{N}^{(v)}$ is a random variable with respect
to $d_{\G}$), so does the density of $\left(N^{(v)}\right)^{-1}(B)$,
and both are equal. A similar argument shows that $d_{\G}$ is a probability
measure on the $\sigma$-algebra generated by $N^{(v)}$, and so $N^{(v)}$
is a random variable with respect to $d_{\G}$. It is finite since
$\widetilde{N}^{(v)}$ is finite and their images may differ by at
most a finite number of values. This proves the first part of the
proposition.

To prove the second part, we use the inversion map, $\I:\kv\mapsto\modp{-\kv}$
which acts on $\mgen$. By Lemma \ref{enu: lem-Inversion-properties-3}
$\I$ preserves the measure $\BGm$, and so

\begin{align*}
\forall j\quad\BGm\left(\left(\bs{N^{(v)}}\right)^{-1}(j)\right) & =\BGm\left(\I\left(\left(\bs{N^{(v)}}\right)^{-1}(j)\right)\right)\\
 & =\BGm\left(\left(\bs{N^{(v)}}\circ\I\right)^{-1}(j)\right)\\
 & =\BGm\left(\left(\bs{N^{(v)}}\right)^{-1}(\deg v-j)\right),
\end{align*}
where in the second line we used that $\I$ is an involution and the
third line is obtained from (\ref{eq: anti-symmetry_of_spectral_position_on_secular_mnfld}),
which reads $\left(\bs{N^{(v)}}\circ\I\right)(\kv)=\deg v-\bs{N^{(v)}}(\kv)$.
To finish the proof observe that
\begin{equation}
\forall j\quad d_{\G}\left(\left(\Nv\right)^{-1}(j)\right)=d_{\G}\left(\left(\widetilde{N}^{(v)}\right)^{-1}(j)\right)=\BGm\left(\left(\Nvb\right)^{-1}(j)\right),\label{eq: density_equals_measure_spectral_pos}
\end{equation}
where the first equality was shown above, and the second follows from
(\ref{eq: density_equals_BG_measure_for_level_set-1}) in Lemma \ref{lem: random-variables}.
\end{proof}

\subsection{The Neumann count distribution of $(3,1)$-regular trees\protect \\
(proof of Theorem\emph{ }\ref{Thm: (3,1)-regular-trees}).\protect \\
~\protect \\
}

Theorem \ref{Thm: (3,1)-regular-trees} is proved using the following.
\begin{defn}
A bridge (or cut-edge) of a graph $\Gamma$ is an edge of $\Gamma$,
whose removal increases the number of connected components of $\Gamma$.
\end{defn}

\begin{prop}
\label{prop: independent_symmetry_of_spectral_positions} Let $\Gamma$
be a standard graph with rationally independent edge lengths. Let
$\widetilde{\Gamma}$ be a sub-graph of $\Gamma$ such that all edges
between $\widetilde{\Gamma}$ and $\Gamma\backslash\widetilde{\Gamma}$
are bridges of $\Gamma$. Let $\tilde{v}$ be a vertex of $\widetilde{\Gamma}$
with degree $\deg{\tilde{v}}>1$ and denote its spectral position
random variable by $N^{(\tilde{v})}$. Further denote by $\overrightarrow{N}^{(\Gamma\backslash\widetilde{\Gamma})}$
the vector of all the spectral position random variables of interior
vertices (of $\Gamma$) in $\Gamma\backslash\widetilde{\Gamma}$.
Then
\begin{equation}
\forall\thinspace1\leq j\leq d_{\tilde{v}}-1,\quad\P\left(N^{(\tilde{v})}=j\thinspace|\thinspace\overrightarrow{N}^{(\Gamma\backslash\widetilde{\Gamma})}\right)=\P\left(N^{(\tilde{v})}=d_{\tilde{v}}-j\thinspace|\thinspace\overrightarrow{N}^{(\Gamma\backslash\widetilde{\Gamma})}\right).\label{eq: independent_symmetry_of_local_spectral_positions}
\end{equation}
Namely, $N^{(\tilde{v})}$ is symmetrically distributed, independently
of all spectral position random variables of $\Gamma\backslash\widetilde{\Gamma}$.
\end{prop}

Before proving the proposition above we first point out and prove
a few of its implications, the last of which is exactly the statement
of Theorem\emph{ }\ref{Thm: (3,1)-regular-trees}.
\begin{cor}
\label{cor: independent_symmetry_of_local_spectra_positions_implications}
Let $\Gamma$ be a standard graph with rationally independent edge
lengths.
\end{cor}

\begin{enumerate}
\item \label{enu: cor-independent_symmetry-1} If $\Gamma$ is a tree graph
and $v,u\in\V\backslash\partial\Gamma$ are two different vertices
then their spectral positions are uncorrelated random variables.
\item \label{enu: cor-independent_symmetry-2} If $\Gamma$ is a $(3,1)$-regular
tree graph then its spectral position random variables, $\left\{ N^{(v)}\right\} _{v\in\V\backslash\partial\Gamma}$
are mutually independent.
\item \label{enu: cor-independent_symmetry-3} If $\Gamma$ is a $(3,1)$-regular
tree graph then the probability distribution of the random variable
$-\omega-1$ is binomial, $\textrm{Bin}(\left|\partial\Gamma\right|-2,\frac{1}{2})$.
\end{enumerate}
\begin{proof}
\uline{Proof of part (\mbox{\ref{enu: cor-independent_symmetry-1}})}.
To apply Proposition \ref{prop: independent_symmetry_of_spectral_positions}
choose for $\widetilde{\Gamma}$ any sub-graph of $\Gamma$ which
contains $v$, but does not contain $u$ (we may even choose $\widetilde{\Gamma}$
to be just $v$). Since $\Gamma$ is a tree, each of its edges is
a bridge and hence indeed any such choice of $\widetilde{\Gamma}$
satisfies the conditions of Proposition \ref{prop: independent_symmetry_of_spectral_positions}.

Now, as a particular case of (\ref{eq: independent_symmetry_of_local_spectral_positions})
we get that
\begin{equation}
\forall\thinspace1\leq j\leq d_{v}-1,\quad\P\left(N^{(v)}=j\thinspace|\thinspace N^{(u)}\right)=\P\left(N^{(v)}=d_{v}-j\thinspace|\thinspace N^{(u)}\right)\label{eq: symmetry_independent_of_two_vertices}
\end{equation}
 and so

\begin{align}
\mathbb{E}\left[\left(N^{(v)}-\frac{d_{v}}{2}\right)\left(N^{(u)}-\frac{d_{u}}{2}\right)\right] & =\sum_{j=1}^{d_{v}-1}\sum_{i=1}^{d_{u}-1}\left(j-\frac{d_{v}}{2}\right)\left(i-\frac{d_{u}}{2}\right)\thinspace\P\left(N^{(v)}=j\thinspace\wedge\thinspace N^{(u)}=i\right)\nonumber \\
=\sum_{j=1}^{d_{v}-1}\sum_{i=1}^{d_{u}-1} & \left(j-\frac{d_{v}}{2}\right)\left(i-\frac{d_{u}}{2}\right)\thinspace\P\left(N^{(v)}=j\thinspace|\thinspace N^{(u)}=i\right)\P\left(N^{(u)}=i\right)\nonumber \\
=\sum_{i=1}^{d_{u}-1}\left(i-\frac{d_{u}}{2}\right) & \P\left(N^{(u)}=i\right)\sum_{j=1}^{d_{v}-1}\left(j-\frac{d_{v}}{2}\right)\P\left(N^{(v)}=j\thinspace|\thinspace N^{(u)}=i\right)=0,\label{eq: expected_value_of_product_of_spctral_positions}
\end{align}
where the sum over $j$ in the last line contains terms which cancel
each other by (\ref{eq: symmetry_independent_of_two_vertices}) and
so the whole sum vanishes. By Proposition \ref{prop: statistics_of_spectral_position},(\ref{enu:prop-statistics_of_local_spectral_position-2})
we have $\mbox{\ensuremath{\mathbb{E}\left[N^{(v)}-\frac{d_{v}}{2}\right]=\mathbb{E}\left[N^{(u)}-\frac{d_{u}}{2}\right]=0}}$,
and so $\mathbb{E}\left[\left(N^{(v)}-\frac{d_{v}}{2}\right)\left(N^{(u)}-\frac{d_{u}}{2}\right)\right]=\mathbb{E}\left[N^{(v)}-\frac{d_{v}}{2}\right]\mathbb{E}\left[N^{(u)}-\frac{d_{u}}{2}\right]$
and $N^{(v)}-\frac{d_{v}}{2}$, $N^{(u)}-\frac{d_{u}}{2}$ are uncorrelated
random variables. We conclude that $N^{(v)}$, $N^{(u)}$ are uncorrelated
random variables.

We remark that the calculation in (\ref{eq: expected_value_of_product_of_spctral_positions})
may be extended to include not just two individual vertices, but any
two sets of vertices, and thus prove a more general statement.

\uline{Proof of part (\mbox{\ref{enu: cor-independent_symmetry-2}})}

Let $\tilde{v}$ be a vertex of $\Gamma\backslash\partial\Gamma$.
Namely, $\tilde{v}$ is not a boundary vertex, and so as $\Gamma$
is $(3,1)$-regular, $d_{\tilde{v}}=3.$ Choose $\widetilde{\Gamma}$
to be $\tilde{v}$ and apply (\ref{eq: independent_symmetry_of_local_spectral_positions})
in Proposition \ref{prop: independent_symmetry_of_spectral_positions},
with $d_{\tilde{v}}=3$ and $j=1$ to get
\[
\P\left(N^{(\tilde{v})}=1\thinspace|\thinspace\overrightarrow{N}^{(\V\backslash\tilde{v})}\right)=\P\left(N^{(\tilde{v})}=2\thinspace|\thinspace\overrightarrow{N}^{(\V\backslash\tilde{v})}\right),
\]
where $\overrightarrow{N}^{(\V\backslash\tilde{v})}$ indicates the
vector of random variables of spectral positions of all non-boundary
vertices of $\Gamma$ except $\tilde{v}$. Now, as $N^{(\tilde{v})}$
gets only the two values $1$ and $2$, we obtain that it is independent
from all other random variables in $\overrightarrow{N}^{(\V\backslash\tilde{v})}$,
as stated in part (\ref{enu: cor-independent_symmetry-2}) of the
Proposition.

\uline{Proof of part (\mbox{\ref{enu: cor-independent_symmetry-3}})}

Applying (\ref{eq:sum_of_spec_pos}) in Proposition \ref{prop:local-global-connections}
gives
\[
\sum_{v\in\V\backslash\partial\Gamma}N^{(v)}=\sigma-\omega+(\left|\E\right|-\left|\partial\Gamma\right|).
\]

As $\Gamma$ is a tree graph, we have $\left|\E\right|=\left|\V\right|-1$
and $\sigma\equiv0$ (see (\ref{eq: nodal surplus bounds})) and so
\[
\omega=-\sum_{v\in\V\backslash\partial\Gamma}N^{(v)}+\left|\V\backslash\partial\Gamma\right|-1,
\]
which we rewrite as
\begin{equation}
-\omega-1=\sum_{v\in\V\backslash\partial\Gamma}\left(N^{(v)}-1\right).\label{eq: Neumann_surplus_as_sum_of_Bernoullis}
\end{equation}
The right hand side of (\ref{eq: Neumann_surplus_as_sum_of_Bernoullis})
is a sum of independent random variables, $N^{(v)}-1$, as is proven
in the previous part. In addition, by Proposition \ref{prop: statistics_of_spectral_position},(\ref{enu:prop-statistics_of_local_spectral_position-2}),
each $N^{(v)}-1$ is a symmetric Bernoulli random variable (it obtains
$0$ and $1$, each with probability $\frac{1}{2}$), and so their
sum is distributed as $\mathrm{Bin}(\left|\V\backslash\partial\Gamma\right|,\frac{1}{2})$.
To prove the required statement we only need to show that $\left|\V\backslash\partial\Gamma\right|=\left|\partial\Gamma\right|-2$,
which follows from
\[
\left|\E\right|=\left|\V\right|-1\quad\textrm{and}\quad2\left|\E\right|=3\left|\V\backslash\partial\Gamma\right|+\left|\partial\Gamma\right|.
\]
The first equality above comes from $\Gamma$ being a tree and the
second from summing all vertex degrees and using that $\Gamma$ is
$(3,1)$-regular.
\end{proof}
We proceed to prove Proposition \ref{prop: independent_symmetry_of_spectral_positions}.
The next lemma forms an essential tool towards this proof. This lemma
is of similar spirit to Lemma \ref{lem: Inversion-properties}; it
provides an involution of the secular manifold, which is applicable
for a graph that contains a bridge.
\begin{lem}
\label{lem: bridge-inv-properties}\emph{\cite[Lem. 4.15]{AloBanBer_cmp18}}

Let $\Gamma$ be a standard connected graph. Let $e$ be a bridge
of $\Gamma$ and denote by $\Gamma_{1},\Gamma_{2}$ the sub-graphs
of $\Gamma$ which are the two connected components of $\Gamma\backslash\{e\}$.
Let $v$ be the vertex of $\Gamma_{1}$ which is connected to $e$.
There exists a map $\rve:\torus\rightarrow\torus$ (which we call
the bridge-inversion map) such that
\begin{enumerate}
\item \label{enu: lem: bridge-inv-properties-1} The bridge-inversion is
an involution, $(\rve)^{2}=\mathrm{Id}$.
\item \label{enu: lem: bridge-inv-properties-2} Each of the manifolds,
$\Sigma$, $\mreg$ and $\mgen$ is invariant under the bridge-inversion.
\item \label{enu: lem: bridge-inv-properties-3} The restriction of the
bridge-inversion to the generic part of the secular manifold, $\rve|_{\mgen}$,
is a map which preserves the Barra-Gaspard measure, $\BGm$.
\item \label{enu: lem: bridge-inv-properties-4} For every $u\in\V$ and
$e\in\E_{u}$ the functions $p_{u,u}$ and $q_{u,u,e}$ (defined in
Lemmas \ref{lem: Canonical_eigenfunctions_and_Secular_mfld} and \ref{lem: Inversion-properties})
satisfy the following symmetry\textbackslash anti-symmetry with respect
to $\rve$:
\[
\forall\kv\in\mreg,\quad p_{u,u}\left(\rve\left(\kv\right)\right)=p_{u,u}\left(\kv\right),
\]
and
\[
\forall\kv\in\mreg,\quad\quad q_{u,u,e}\left(\rve\left(\kv\right)\right)=\begin{cases}
q_{u,u,e}\left(\kv\right) & u\in\V_{1}\\
-q_{u,u,e}\left(\kv\right) & u\in\V_{2}
\end{cases}.
\]
\end{enumerate}
\end{lem}

The proof of this lemma may be found within the proofs of \cite[Lem. 4.15]{AloBanBer_cmp18}
and \cite[Lem. 4.42]{Alon_PhDThesis}.

\begin{proof}
[Proof of Proposition \ref{prop: independent_symmetry_of_spectral_positions}]

This proof is somewhat similar in spirit to the one of Proposition
\ref{prop: statistics_of_spectral_position},(\ref{enu:prop-statistics_of_local_spectral_position-2}),
where it was proven that the spectral position of any vertex is a
symmetric random variable. Yet, to prove that it is not only symmetric,
but independently symmetric (when conditioning on spectral position
of other vertices), requires some more work.

Let $\Gamma$, $\widetilde{\Gamma}$ and $\tilde{v}$ as in the statement
of the proposition. Let $j\in\N$ and let $\vec{n}$ be a vector of
natural numbers of length $\left|\Gamma\backslash\left(\partial\Gamma\cup\widetilde{\Gamma}\right)\right|$.
So, to prove the proposition we need to show that
\begin{equation}
\P\left(N^{(\tilde{v})}=j\thinspace\wedge\thinspace\overrightarrow{N}^{(\Gamma\backslash\widetilde{\Gamma})}=\vec{n}\right)=\P\left(N^{(\tilde{v})}=d_{\tilde{v}}-j\thinspace\wedge\thinspace\overrightarrow{N}^{(\Gamma\backslash\widetilde{\Gamma})}=\vec{n}\right).\label{eq: symmetric_propbabilites_spectral_position}
\end{equation}

Translating (\ref{eq: symmetric_propbabilites_spectral_position})
into densities gives
\begin{equation}
d_{\G}\left(\left(N^{(\tilde{v})}\right)^{-1}(j)\thinspace\cap\thinspace\left(\overrightarrow{N}^{(\Gamma\backslash\widetilde{\Gamma})}\right)^{-1}(\vec{n})\right)=d_{\G}\left(\left(N^{(\tilde{v})}\right)^{-1}(d_{\tilde{v}}-j)\thinspace\cap\thinspace\left(\overrightarrow{N}^{(\Gamma\backslash\widetilde{\Gamma})}\right)^{-1}(\vec{n})\right).\label{eq: symmetric_densities_spectral_position}
\end{equation}
This may be further translated into the corresponding Barra-Gaspard
measure on the secular manifold,
\begin{equation}
\BGm\left(\left(\bs{N^{(\tilde{v})}}\right)^{-1}(j)\thinspace\cap\thinspace\left(\bs{\overrightarrow{N}^{(\Gamma\backslash\widetilde{\Gamma})}}\right)^{-1}(\vec{n})\right)=\BGm\left(\left(\bs{N^{(\tilde{v})}}\right)^{-1}(d_{\tilde{v}}-j)\thinspace\cap\thinspace\left(\bs{\overrightarrow{N}^{(\Gamma\backslash\widetilde{\Gamma})}}\right)^{-1}(\vec{n})\right),\label{eq: symmetric_measure_spectral_position}
\end{equation}
where $\bs{N^{(\tilde{v})}}$ and $\bs{\bs{\overrightarrow{N}^{(\Gamma\backslash\widetilde{\Gamma})}}}$
are the spectral position functions on the secular manifold which
were introduced in Lemma \ref{lem: functions_on_secular_manifold_existence_and_symmetry}.
The argument for translating the natural densities, $d_{\G}$ to the
Barra-Gaspard measure, $\BGm$ , is as in the proof of Proposition
\ref{prop: statistics_of_spectral_position} (see (\ref{eq: density_equals_measure_spectral_pos})
there).

We proceed to prove (\ref{eq: symmetric_measure_spectral_position}).
Let $\left\{ e_{i}\right\} _{i=1}^{m}$ be all the edges which connect
$\widetilde{\Gamma}$ to $\Gamma\backslash\widetilde{\Gamma}$. By
assumption, all those edges are bridges. For each $e_{i}$ denote
by $v_{i}$ the vertex of $\widetilde{\Gamma}$ which is connected
to $e_{i}$. Further denote by $\Gamma_{i}$ the connected component
of $\Gamma\backslash e_{i}$ which does not contain $\widetilde{\Gamma}$.
With those notations, we may write the following decomposition of
the graph
\[
\Gamma=\widetilde{\Gamma}\cup\bigcup_{i=1}^{m}e_{i}\cup\bigcup_{i=1}^{m}\Gamma_{i},
\]
which is disjoint up to the vertices at the endpoints of the bridges
$\left\{ e_{i}\right\} _{i=1}^{m}$. Consider the following composition
of involutions, $\I\cdot\mathcal{R}_{v_{1},e_{1}}\cdot\ldots\cdot\mathcal{R}_{v_{m},e_{m}}$
and note that all those involutions are $\BGm$ measure preserving
(Lemma \ref{lem: Inversion-properties},(\ref{enu: lem-Inversion-properties-3})
and Lemma \ref{lem: bridge-inv-properties},(\ref{enu: lem: bridge-inv-properties-3})).
Therefore, in order to prove (\ref{eq: symmetric_measure_spectral_position})
it is enough to show that
\begin{equation}
\I\circ\mathcal{R}_{v_{1},e_{1}}\circ\ldots\circ\mathcal{R}_{v_{m},e_{m}}\left(\left(\bs{N^{(\tilde{v})}}\right)^{-1}(j)\right)=\left(\bs{N^{(\tilde{v})}}\right)^{-1}(d_{\tilde{v}}-j)\label{eq: action-of-involutions-1}
\end{equation}
 and
\begin{equation}
\I\circ\mathcal{R}_{v_{1},e_{1}}\circ\ldots\circ\mathcal{R}_{v_{m},e_{m}}\left(\left(\bs{\overrightarrow{N}^{(\Gamma\backslash\widetilde{\Gamma})}}\right)^{-1}(\vec{n})\right)=\left(\bs{\overrightarrow{N}^{(\Gamma\backslash\widetilde{\Gamma})}}\right)^{-1}(\vec{n}).\label{eq: action-of-involutions-2}
\end{equation}

To verify both, we recall that (see (\ref{eq:definition of Nv-1}))
\begin{equation}
\forall v\in\V,\quad\bs{N^{(v)}}(\kv)=\frac{1}{2}d_{v}-\frac{1}{2}\sum_{e\in\E_{v}}\sgn\left(q_{v,v,e}\left(\kv\right)\right)\label{eq: spectral_position_by_values_and_derivatives}
\end{equation}
and that (by Lemmas \ref{lem: Inversion-properties} and \ref{lem: bridge-inv-properties})
\begin{equation}
\forall v\in\V~\forall e\in\E_{v},\quad\quad q_{v,v,e}\left(\I\left(\kv\right)\right)=-q_{v,v,e}\left(\kv\right),\label{eq: action-of-involutions-3}
\end{equation}
\begin{equation}
\forall v\in\V~\forall e\in\E_{v},\quad\quad q_{v,v,e}\left(\mathcal{R}_{v_{i},e_{i}}(\kv)\right)=\begin{cases}
q_{v,v,e}\left(\kv\right) & v\in\Gamma\backslash\Gamma_{i}\\
-q_{v,v,e}\left(\kv\right) & v\in\Gamma_{i}
\end{cases}.\label{eq: action-of-involutions-4}
\end{equation}

Each vertex $v\in\Gamma\backslash\widetilde{\Gamma}$ belongs to exactly
one of the $\Gamma_{i}$ graphs, so by (\ref{eq: action-of-involutions-3})
and (\ref{eq: action-of-involutions-4}) we have that (\ref{eq: spectral_position_by_values_and_derivatives})
is invariant under $\I\circ\mathcal{R}_{v_{1},e_{1}}\circ\ldots\circ\mathcal{R}_{v_{m},e_{m}}$
(as each $\sgn\left(q_{v,v,e}\right)$ is inverted twice by its action
and hence unchanged). This implies
\[
\bs{N^{(v)}}(\I\circ\mathcal{R}_{v_{1},e_{1}}\circ\ldots\circ\mathcal{R}_{v_{m},e_{m}}(\kv))=\bs{N^{(v)}}(\kv),
\]
which proves (\ref{eq: action-of-involutions-2}). Similarly, any
vertex $\tilde{v}\in\widetilde{\Gamma}$ is not contained in any of
the $\Gamma_{i}$ graphs and so by (\ref{eq: action-of-involutions-3})
and (\ref{eq: action-of-involutions-4}) we have that each $\sgn\left(q_{\tilde{v},\tilde{v},e}\right)$
in (\ref{eq: spectral_position_by_values_and_derivatives}) is inverted
once by the action of $\I\circ\mathcal{R}_{v_{1},e_{1}}\circ\ldots\circ\mathcal{R}_{v_{m},e_{m}}$.
This implies
\[
\bs{N^{(\tilde{v})}}(\I\circ\mathcal{R}_{v_{1},e_{1}}\circ\ldots\circ\mathcal{R}_{v_{m},e_{m}}(\kv))=\deg v-\bs{N^{(\tilde{v})}}(\kv)
\]
 and proves (\ref{eq: action-of-involutions-1}).
\end{proof}

\section{Probability distribution of the wavelength capacity \protect \\
(proof of Proposition \ref{prop: statistics_of_wavelength_capacity})
\label{sec: proofs_wavelength_capacity}}

The proof of Proposition \ref{prop: statistics_of_wavelength_capacity}
is based on the existence and properties of the function $\Rvb$ defined
on $\mgen$ (Lemma \ref{lem: functions_on_secular_manifold_existence_and_symmetry}).
In the proof we need to argue that some level sets of $\Rvb$ are
of measure zero. To show it we use that $\mgen$ (and actually even
$\mreg$) is a real analytic manifold \cite{CdV_ahp15,AloBanBer_cmp18,Alon_PhDThesis}
and the following.
\begin{lem}
\label{lem: Zero-set-of-real-analyic} Let $\M\subset\torus$ be a
connected real analytic manifold of dimension $E-1$ and let $d\mu$
be its volume element. Let $g$ be a real analytic function on $\mathcal{M}$.
\begin{enumerate}
\item \label{enu: lem-zero-set-of-real-analytic-1} Either $g|_{\M}\equiv0$
or the zero set of $g$ is of co-dimension at least one in $\M$.
\item \label{enu: lem-zero-set-of-real-analytic-2} If $g$ is not constant
on $\mathcal{M}$ then the pre-image by $g$ of any set of (Lebesgue)
measure zero is of $\mu$-measure zero.
\end{enumerate}
\end{lem}

\begin{proof}
Denote the zero set of $g$ by $Z_{g}:=\left\{ \kv\in\M\,|\,g\left(\kv\right)=0\right\} .$
Choose some atlas $\left\{ \left(U_{n},~\varphi_{n}\right)\right\} $
for $\M$, such that for all $n$, $\varphi_{n}:U_{n}\rightarrow O_{n}\subset\R^{E-1}$.
Since $\mathcal{M}$ is a real analytic manifold and $g$ is real
analytic, we get that for all $n$, $g_{n}:=g\circ\varphi_{n}^{-1}$
is real analytic on $O_{n}\subset\R^{E-1}$. By \cite[Prop. 3]{Mit15}
we get that either $g_{n}|_{O_{n}}\equiv0$ or its zero set, $\varphi_{n}\left(Z_{g}\cap U_{n}\right)$
is of co-dimension at least one in $\R^{E-1}$. Hence, either $g|_{U_{n}}\equiv0$
or $Z_{g}\cap U_{n}$ has a positive co-dimension in $U_{n}$. We
use this dichotomy to define
\begin{align*}
N_{1} & :=\set n{U_{n}\subset Z_{g}}\\
N_{2} & :=\set n{Z_{g}\cap U_{n}\textrm{ has a positive co-dimension in }U_{n}}
\end{align*}

and

\begin{align*}
A=\cup_{n\in N_{1}}U_{n} & ,\quad B=\cup_{n\in N_{2}}U_{n}.
\end{align*}
Clearly, $A\cap B=\emptyset$ and so we get the partition of $\M=A\sqcup B$
as a union of disjoint open sets. Since $\mathcal{M}$ is connected
we get either $A=\M$ which implies $g|_{\M}\equiv0$ or $B=\M$ which
implies that $Z_{g}$ is of positive co-dimension in $\M$ and proves
the first part of the lemma.

To prove the second part of the Lemma, consider the gradient of $g_{n}$
which we denote by $\nabla g_{n}:O_{n}\rightarrow\R^{E-1}$ and its
zero set which we denote by $Z_{\nabla g_{n}}$. Each of the components
of $\nabla g_{n}$ is real analytic, since $g_{n}$ itself is real
analytic. Next, apply the argument from the first part of the proof
for each component $\partial_{i}g_{n}$ of the gradient. We get that
either the zero set of $\partial_{i}g_{n}$ is $O_{n}$ or it is of
positive co-dimension. If for all components $\partial_{i}g_{n}$
the zero set is $O_{n}$ this implies that $g_{n}$ is constant on
$U_{n}$. We denote
\begin{align*}
N_{3} & :=\set n{g_{n}\textrm{ is constant}}\\
N_{4} & :=\set n{g_{n}\textrm{ is not constant}}
\end{align*}
 and

\begin{align*}
C=\cup_{n\in N_{3}}U_{n} & ,\quad D=\cup_{n\in N_{4}}U_{n}.
\end{align*}
Exactly as above we get that since $\mathcal{M}$ is connected either
$C=\mathcal{M}$ or $D=\mathcal{M}$. In the former case $g$ is constant
on $\mathcal{M}$. But, $g$ is not constant by the assumption of
our lemma and so $D=\mathcal{M}$ and $g_{n}$ is not constant (for
any $n$). By the argument above, this means that for each $n$ there
exists $\partial_{i}g_{n}$, whose zero set is of positive co-dimension,
which implies that the zero set of the gradient, $\nabla g_{n}$ is
of positive co-dimension and hence is of measure zero (with respect
to Lebesgue measure on $\R^{E-1}$). This is the condition stated
in \cite[Thm. 1]{Ponomarev_smj87}, from which we deduce that if $B\subset\R$
is of measure zero then $g_{n}^{-1}(B)\subset\R^{E-1}$ is of measure
zero. Therefore, $\varphi_{n}^{-1}(g_{n}^{-1}(B))=g^{-1}(B)\cap U_{n}$
is of $\mu$-measure zero and this holds for all charts (since $D=\mathcal{M}$).
Hence, we get that $g^{-1}(B)$ is of $\mu$-measure zero.
\end{proof}
~\\

\begin{proof}
[Proof of Proposition \ref{prop: statistics_of_wavelength_capacity}]

The first step in the proof is to express the LHS of (\ref{eq:denstiy-as-integral-of-distribution})
as
\begin{equation}
\forall a<b,\quad d_{\G}\left(\left(\rho^{(v)}\right)^{-1}\left((a,b)\right)\right)=\BGm\left(\left(\bs{\rho^{(v)}}\right)^{-1}\left((a,b)\right)\right),\label{eq: pre-image_of_rho_density_and_measure}
\end{equation}
where $\bs{\rho^{\left(v\right)}}:\mgen\rightarrow\R$ is the real
analytic function whose existence and properties specified in Lemma
\ref{lem: functions_on_secular_manifold_existence_and_symmetry}.
In order to show (\ref{eq: pre-image_of_rho_density_and_measure})
we need to argue that $\left(\bs{\rho^{(v)}}\right)^{-1}\left((a,b)\right)$
is a Jordan set and then apply Corollary \ref{cor: density_equals_BG_for_Jordan_set}.

To prove that $\left(\bs{\rho^{(v)}}\right)^{-1}\left((a,b)\right)$
is a Jordan set, we start by noting that $\bs{\rho^{\left(v\right)}}:\mgen\rightarrow\R$
is a continuous function (being even real analytic). Hence, for every
open $I\subset\R$
\begin{align}
\partial\left\{ \left(\bs{\rho^{\left(v\right)}}\right)^{-1}\left(I\right)\right\} = & \overline{\left(\bs{\rho^{\left(v\right)}}\right)^{-1}\left(I\right)}\backslash\left(\bs{\rho^{\left(v\right)}}\right)^{-1}\left(I\right)\nonumber \\
\subset & \left(\bs{\rho^{\left(v\right)}}\right)^{-1}\left(\overline{I}\right)\backslash\left(\bs{\rho^{\left(v\right)}}\right)^{-1}\left(I\right)=\left(\bs{\rho^{\left(v\right)}}\right)^{-1}\left(\partial I\right)\label{eq: topology_of_inv_rho_1}
\end{align}
where $\overline{\phantom{I}}$ and $\partial\phantom{}$ denote,
correspondingly, the closure and boundary operators (in $\R$ or in
$\mgen$ according to the context). Now using that $I$ is open we
get $I\cap\partial I=\emptyset$ and in particular, $\left(\bs{\rho^{\left(v\right)}}\right)^{-1}\left(I\right)\cap\left(\bs{\rho^{\left(v\right)}}\right)^{-1}\left(\partial I\right)=\emptyset$.
As the first set is open and the second is closed we have
\begin{equation}
\partial\left\{ \left(\bs{\rho^{\left(v\right)}}\right)^{-1}\left(I\right)\right\} \cap\left\{ \left(\bs{\rho^{\left(v\right)}}\right)^{-1}\left(\partial I\right)\right\} ^{\circ}=\emptyset,\label{eq: topology_of_inv_rho_2}
\end{equation}
where $\phantom{I}^{\circ}$ denotes the interior operator. Combining
(\ref{eq: topology_of_inv_rho_2}) with (\ref{eq: topology_of_inv_rho_1})
yields the stronger inclusion
\[
\partial\left\{ \left(\bs{\rho^{\left(v\right)}}\right)^{-1}\left(I\right)\right\} \subset\partial\left\{ \left(\bs{\rho^{\left(v\right)}}\right)^{-1}\left(\partial I\right)\right\} .
\]
In particular, choosing $I=(a,b)$ gives
\begin{equation}
\partial\left\{ \left(\bs{\rho^{\left(v\right)}}\right)^{-1}\left((a,b)\right)\right\} \subset\partial\left\{ \left(\bs{\rho^{\left(v\right)}}\right)^{-1}\left(a\right)\right\} \cup\partial\left\{ \left(\bs{\rho^{\left(v\right)}}\right)^{-1}\left(b\right)\right\} .\label{eq:boundary of rho sets}
\end{equation}

Hence, since $\mgen$ has finitely many connected components (Lemma
\ref{lem: generic secman}) and by Lemma \ref{lem: Zero-set-of-real-analyic},
we get that each of $\left\{ \left(\bs{\rho^{\left(v\right)}}\right)^{-1}\left(a\right)\right\} $
and $\left\{ \left(\bs{\rho^{\left(v\right)}}\right)^{-1}\left(b\right)\right\} $
is a finite union of connected components of $\mgen$ and sets of
positive co-dimension in $\mgen$. The boundary of a connected component
of $\mgen$ is an empty set\footnote{We have shown in the proof of Theorem \ref{thm:Neumann_surplus_main},
(\ref{enu:thm-density_of_Neumann_surplus}) that a connected component
of $\mgen$ is open and closed and hence it has an empty boundary. }. The boundary of a set of positive co-dimension has itself positive
co-dimension and is therefore of measure zero. Hence, the boundaries
$\partial\left\{ \left(\bs{\rho^{\left(v\right)}}\right)^{-1}\left(a\right)\right\} $
and $\partial\left\{ \left(\bs{\rho^{\left(v\right)}}\right)^{-1}\left(b\right)\right\} $
are of measure zero and by (\ref{eq:boundary of rho sets}) we conclude
that $\partial\left\{ \left(\bs{\rho^{\left(v\right)}}\right)^{-1}\left((a,b)\right)\right\} $
is also of measure zero. We conclude that $\left(\bs{\rho^{\left(v\right)}}\right)^{-1}\left((a,b)\right)$
is indeed a Jordan set and may now apply Corollary \ref{cor: density_equals_BG_for_Jordan_set}
and get
\[
d_{\G}\left(\flow{}^{-1}\left(\left(\bs{\rho^{\left(v\right)}}\right)^{-1}\left((a,b)\right)\right)\right)=\BGm\left(\left(\bs{\rho^{(v)}}\right)^{-1}\left((a,b)\right)\right).
\]
From here, we obtain (\ref{eq: pre-image_of_rho_density_and_measure})
from the beginning of the proof by observing that the sets $\flow{}^{-1}\left(\left(\bs{\rho^{\left(v\right)}}\right)^{-1}\left((a,b)\right)\right)$
and $\left(\rho^{\left(v\right)}\right)^{-1}\left((a,b)\right)$ differ
by at most finitely many elements. Indeed, this was shown in (\ref{eq: Wave_capacity_on_secular_manifold})
of Lemma \ref{lem: functions_on_secular_manifold_existence_and_symmetry}.

Having shown (\ref{eq: pre-image_of_rho_density_and_measure}), the
first part of the Proposition is therefore equivalent to
\begin{equation}
\BGm\left(\left(\bs{\rho^{\left(v\right)}}\right)^{-1}\left((a,b)\right)\right)=\int_{a}^{b}\pi^{(v)}\left(x\right)dx+\sum_{x_{j}\in\left(a,b\right)}p^{(v)}\left(x_{j}\right),\label{eq: denstiy-as-integral-of-distribution-in-proof}
\end{equation}
for any interval $\left(a,b\right)$, where $\pi^{(v)}$ is a density
function and $p^{(v)}$ is a discrete measure supported on the finite
set $\left\{ x_{j}\right\} _{j=1}^{m}$. To define $p^{(v)}$ recall
that $\mgen$ has finitely many connected components (Lemma \ref{lem: generic secman}),
denote the connected components on which $\bs{\rho^{\left(v\right)}}$
is constant by $\left\{ M_{i}\right\} _{i=1}^{m}$ and denote their
union by
\[
\mgendisc=\cup_{j=1}^{m}M_{j}.
\]
Let $\left\{ x_{j}\right\} _{j=1}^{m}$ be the values of $\bs{\rho^{\left(v\right)}}$
on $\mgendisc$ such that $\bs{\rho^{\left(v\right)}}|_{M_{i}}\equiv x_{i}$.
Note that not all $x_{i}$ are necessarily different. Define a discrete
measure on $\R$ by
\[
p^{\left(v\right)}:=\sum_{j=1}^{m}\mu_{\lv}\left(M_{j}\right)\delta_{x_{j}},
\]
where $\delta_{x_{j}}$is the Dirac measure at $x_{j}$. For any Borel
set $B\subset\R$,
\begin{equation}
p^{\left(v\right)}\left(B\right)=\mu_{\lv}\left(\mgendisc\cap\left(\bs{\rho^{\left(v\right)}}\right)^{-1}\left(B\right)\right).\label{eq: discrete_probability_density}
\end{equation}
Denote the complement $\mgencont=\mgen\setminus\mgendisc$ and define
$\zeta^{\left(v\right)}:=\bs{\rho^{\left(v\right)}}_{*}\left(\mu_{\lv}|_{\mgencont}\right)$
as the push-forward by $\bs{\rho^{\left(v\right)}}$ of the restricted
Barra-Gaspard measure $\mu_{\lv}|_{\mgencont}$. Namely, for any Borel
$B\subset\R$,
\begin{equation}
\zeta^{\left(v\right)}\left(B\right)=\mu_{\lv}\left(\mgencont\cap\left(\bs{\rho^{\left(v\right)}}\right)^{-1}\left(B\right)\right).\label{eq: continuous_probability_density}
\end{equation}
Summing (\ref{eq: discrete_probability_density}) and (\ref{eq: continuous_probability_density})
gives that for any Borel $B\subset\R$,
\begin{align*}
\zeta^{\left(v\right)}\left(B\right)+p^{\left(v\right)}\left(B\right)= & \mu_{\lv}\left(\left(\bs{\rho^{\left(v\right)}}\right)^{-1}\left(B\right)\right).
\end{align*}
We will now show that $\zeta^{\left(v\right)}(B)+p^{\left(v\right)}(B)$
is a decomposition of a measure to an absolutely continuous part and
a discrete part\footnote{Actually, we conjecture that there is no connected component on which
$\bs{\rho^{\left(v\right)}}$ is constant, and so $p^{\left(v\right)}\equiv0$.
See Remark \ref{rem: no_atoms_in_rho_distribution} after the proof.} . As $p^{\left(v\right)}$ was defined as a discrete measure, we
only need to show that that $\zeta^{\left(v\right)}$ is absolutely
continuous. Once proving that $\zeta^{\left(v\right)}$ is absolutely
continuous, we may prove (\ref{eq: denstiy-as-integral-of-distribution-in-proof})
by defining $\pi^{\left(v\right)}\left(x\right)$ as the Radon-Nykodim
derivative of $\zeta^{\left(v\right)}$ with respect to the Lebesgue
measure, so that we get $d\zeta^{\left(v\right)}\left(x\right)=\pi^{\left(v\right)}\left(x\right)dx$.

The measure $\zeta^{\left(v\right)}$ is absolutely continuous if
for every $B\subset\R$ of Lebesgue measure zero, $\zeta^{\left(v\right)}\left(B\right)=0$.
Let $B\subset\R$ of Lebesgue measure zero and let $M$ be a connected
component of $\mgencont$, with the volume element $d\mu$. As $M$
is a connected real analytic manifold and the restriction $\bs{\rho^{\left(v\right)}}|_{M}$
is a non-constant real analytic function, we may use Lemma \ref{lem: Zero-set-of-real-analyic},(\ref{enu: lem-zero-set-of-real-analytic-2})
to deduce that $\mu((\bs{\rho^{(v)}})^{-1}(B)\cap M)=0$. Since $\mu$
and the restriction of the Barra-Gaspard measure $\mu_{\lv}|_{M}$
are equivalent, we get that also $\mu_{\lv}((\bs{\rho^{\left(v\right)}})^{-1}(B)\cap M)=0$.
As this holds for every connected component of $\mgencont$, taking
the union gives
\[
\zeta^{\left(v\right)}\left(B\right)=\mu_{\lv}\left(\bs{\rho^{\left(v\right)}}^{-1}\left(B\right)\cap\mgencont\right)=0,
\]
which finishes the first part of the proof.

~\\

Next, we prove part (\ref{enu:prop-statistics_of_local_wave_capacity-2})
of the proposition,
\begin{equation}
\pi^{\left(v\right)}(x)=\pi^{\left(v\right)}(\deg v-x)\quad\quad\textrm{and}\quad\quad p^{\left(v\right)}(x)=p^{\left(v\right)}(\deg v-x).\label{eq: symmetry_for_densities}
\end{equation}
By (\ref{eq: anti-symmetry_of_wavelength_capacity_on_secular_mnfld})
in Lemma \ref{lem: functions_on_secular_manifold_existence_and_symmetry},
we have that $\bs{\rho^{\left(v\right)}}\left(\I\left(\kv\right)\right)=\deg v-\bs{\rho^{\left(v\right)}}\left(\kv\right)$,
where $\I:\kv\mapsto\modp{-\kv}$ is the inversion (see Section \ref{subsec: inversion}).
If $M_{i}$ is a connected component of $\mgendisc$, on which $\bs{\rho^{\left(v\right)}}|_{M_{i}}\equiv x_{i}$
then $\I\left(M_{i}\right)$ is a connected open subset of $\mgen$
on which $\bs{\rho^{\left(v\right)}}|_{\I\left(M_{i}\right)}\equiv\deg v-x_{i}$.
It follows $\I\left(M_{i}\right)$ is contained in some $M_{j}$,
connected component of $\mgendisc$. Since $\I^{2}\left(M_{i}\right)=M_{i}$
then in fact $\I\left(M_{i}\right)=M_{j}$, with $x_{j}=\deg v-x_{i}$.
This argument shows that $\mgendisc$ is $\I$ invariant, and since
we know that $\mgen$ is also $\I$ invariant (Lemma \ref{lem: Inversion-properties},(\ref{enu: lem-Inversion-properties-1}))
then so does $\mgencont$. Given the relation $\bs{\rho^{\left(v\right)}}\left(\I\left(\kv\right)\right)=\deg v-\bs{\rho^{\left(v\right)}}\left(\kv\right)$,
then for any Borel set $B\subset\R$ and its reciprocal set:
\[
\deg v-B:=\set{\deg v-x}{x\in B},
\]
we have
\begin{align*}
\I\left(\mgendisc\cap\left(\bs{\rho^{\left(v\right)}}\right)^{-1}\left(B\right)\right) & =\mgendisc\cap\left(\bs{\rho^{\left(v\right)}}\right)^{-1}\left(\deg v-B\right),\,\,\text{and}\\
\I\left(\mgencont\cap\left(\bs{\rho^{\left(v\right)}}\right)^{-1}\left(B\right)\right) & =\mgencont\cap\left(\bs{\rho^{\left(v\right)}}\right)^{-1}\left(\deg v-B\right).
\end{align*}
The symmetry of both $p^{\left(v\right)}$ and $\zeta^{\left(v\right)}$
is now a consequence of $\I$ being $\BGm$ preserving (Lemma \ref{lem: Inversion-properties}):
\begin{align*}
p^{\left(v\right)}\left(B\right)=\mu_{\lv}\left(\mgendisc\cap\left(\bs{\rho^{\left(v\right)}}\right)^{-1}\left(B\right)\right) & =\mu_{\lv}\left(\mgendisc\cap\left(\bs{\rho^{\left(v\right)}}\right)^{-1}\left(\deg v-B\right)\right)=p^{\left(v\right)}\left(\deg v-B\right)\\
\zeta^{\left(v\right)}\left(B\right)=\mu_{\lv}\left(\mgencont\cap\left(\bs{\rho^{\left(v\right)}}\right)^{-1}\left(B\right)\right) & =\mu_{\lv}\left(\mgencont\cap\left(\bs{\rho^{\left(v\right)}}\right)^{-1}\left(\deg v-B\right)\right)=\zeta^{\left(v\right)}\left(\deg v-B\right).
\end{align*}
\end{proof}
\begin{rem}
\label{rem: no_atoms_in_rho_distribution}

We conjecture that the probability distribution of the wavelength
capacity does not contain a discrete part, namely that $p^{(v)}\equiv0$.
In view of the proof above, this is equivalent to the function $\bs{\Rv}$
not being constant on connected components of $\mgen$. Apart from
an intuition (as demanding $\bs{\Rv}$ to be constant is highly restrictive)
our conjecture is supported by numerical investigations, which show
no atoms for the probability distribution of $\Rv$ on various graphs
(see discussion in Section \ref{sec:Discussion} and Figure \ref{fig: graph rho example}).
\end{rem}

\section{Discussion and open problems \label{sec:Discussion}}

This paper presents Neumann domains on metric (quantum) graphs. Neumann
domains may be perceived as the counterpart of nodal domains. Nodal
domains are sub-graphs which are bounded by the zeros of the eigenfunction,
whereas Neumann domains are sub-graphs bounded by zeros of the eigenfunction's
derivative. This close similarity between nodal domains and Neumann
domains calls for a comparison. On one hand there are analogous results
for both: bounds and probability distributions of their counts (compare
Theorem \ref{thm: surplus_function_on_secular_mnfld} for Neumann
count with \cite[Theorem 2.1]{AloBanBer_cmp18} for nodal count).
Yet, there are also similar statements which have different incarnations.
The nodal count of graphs with disjoint cycles is similar to the Neumann
count of $(3,1)$-regular tree graphs - both counts have binomial
distributions (compare \cite[Theorem 2.3]{AloBanBer_cmp18} with Theorem
\ref{Thm: (3,1)-regular-trees} here).

A useful viewpoint for comparison between the nodal count and the
Neumann count is the perspective of inverse problems. Namely, which
information on the underlying graph is stored in the nodal\textbackslash Neumann
count sequence. A partial (but quite satisfying) answer lies in Corollary
\ref{cor:Expected_values_of_nodal_and_Neumann_supluses}. By this
corollary, the expected value of the nodal surplus distribution equals
$\frac{1}{2}\beta$, where $\beta$ is the first Betti number of the
graph; and the expected value of the Neumann nodal surplus distribution
equals $\frac{1}{2}(\beta-\left|\partial\Gamma\right|)$, where $\left|\partial\Gamma\right|$
is the size of the graph's boundary. An easy demonstration of the
difference between these two 'pieces' of information may be given
in terms of tree graphs. By Corollary \ref{cor:Expected_values_of_nodal_and_Neumann_supluses}
it is evident that the nodal surplus distribution reveals whether
the underlying graph is a tree ($\beta=0$) or not. Furthermore, all
tree graphs have exactly the same nodal count sequence (\cite{Schapotschnikow06,AlO_viniti92,PokPryObe_mz96}),
which distinguishes them from any non-tree graph \cite{Ban_ptrsa14}.
But, this also means that the nodal count sequence does not store
any information which allows to distinguish between different trees.
In this case, the Neumann surplus sequence is more informative as
it predicts $\left|\dg\right|$, the size of the tree's boundary.
Furthermore, we find numerically that there are trees, which share
the same boundary size, yet their Neumann surplus distributions are
different (though with the same expected value); See Figure \ref{fig: different trees}.
This raises the additional challenge of exposing which information
is stored in the Neumann surplus sequence beyond the value of $\beta-\left|\partial\Gamma\right|$.

\begin{figure}[h]
\begin{centering}
(i) $\quad\Gamma_{1}$\includegraphics[scale=0.4]{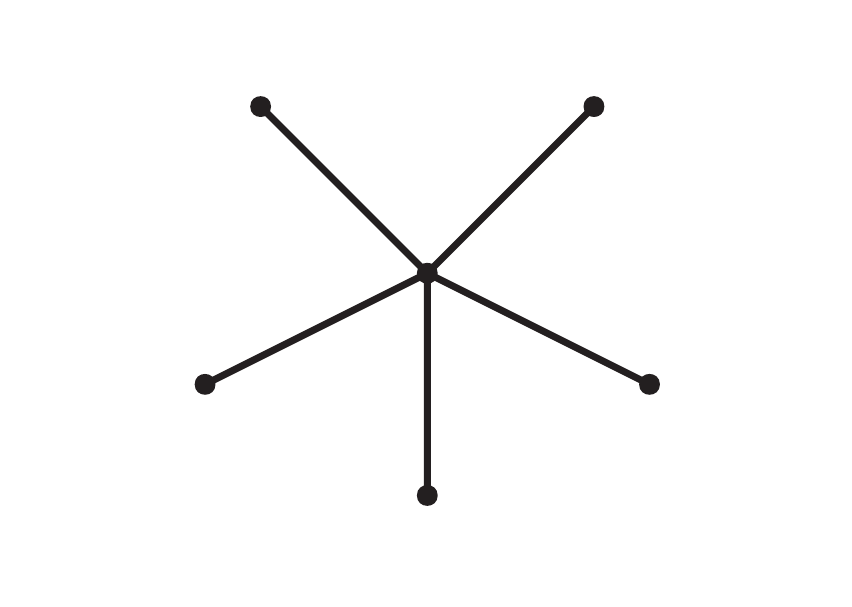} (ii)
$\quad\Gamma_{2}$\includegraphics[scale=0.4]{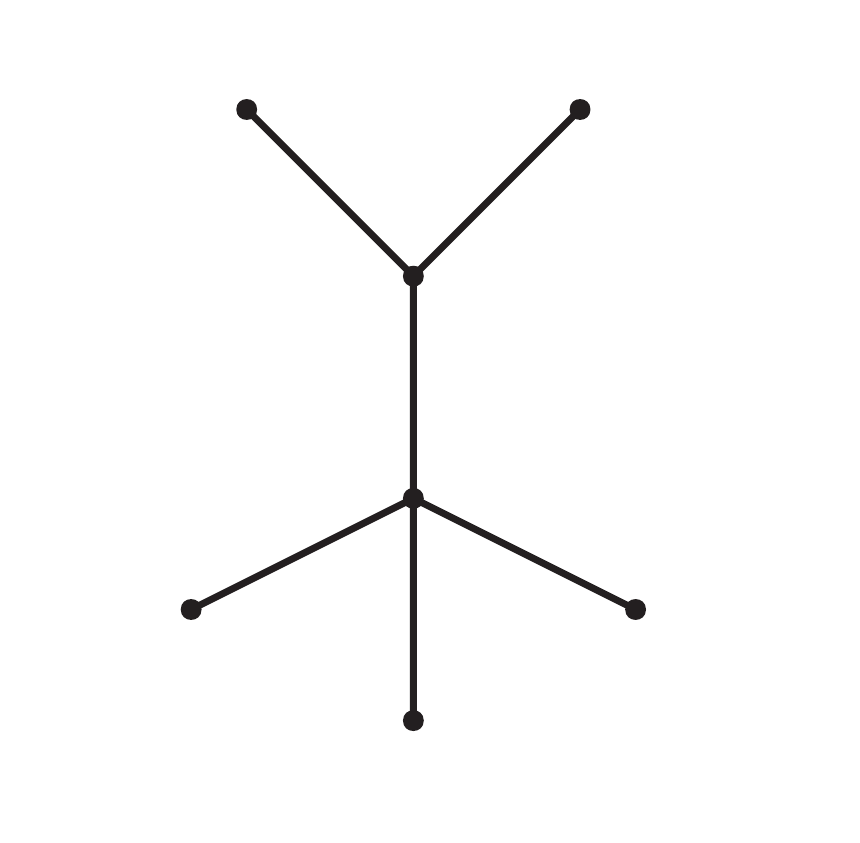}(iii)$\quad\Gamma_{3}$\includegraphics[scale=0.4]{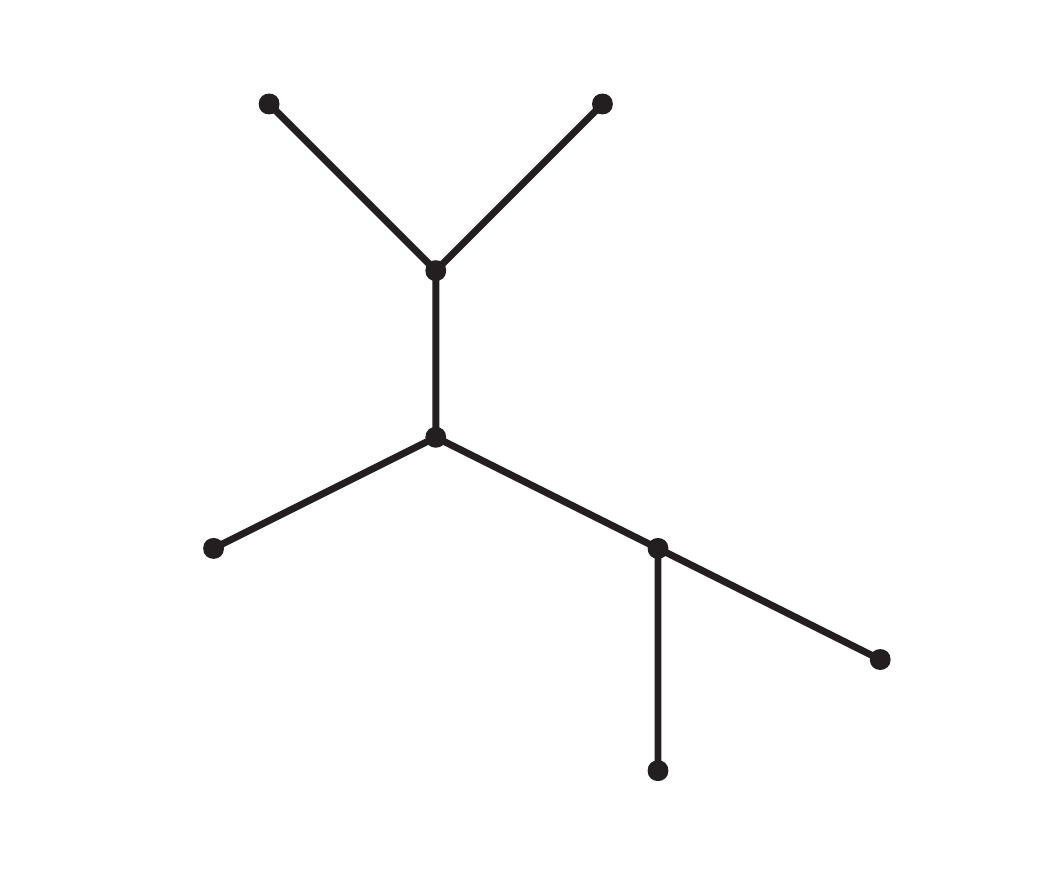}
\par\end{centering}
\includegraphics[width=0.7\paperwidth]{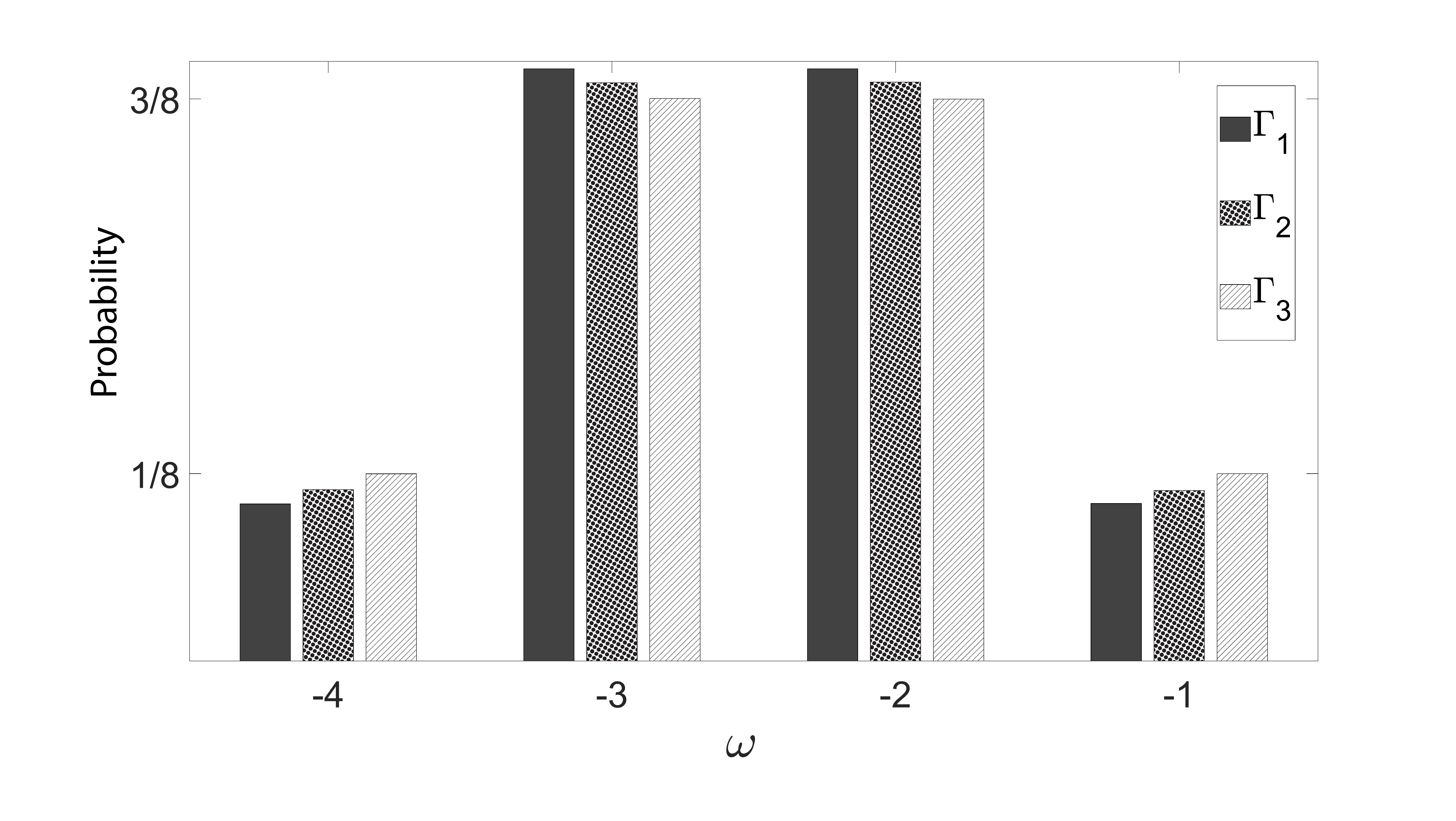} \centering{}\caption{The Neumann surplus distribution of three different trees with the
same boundary size, $\left|\partial\Gamma\right|=5$. The expected
value of all three probability distributions is the same, but the
distributions themselves are different. The numerical data was calculated
for $\sim10^{6}$ eigenfunctions per graph. }
\label{fig: different trees}
\end{figure}

Having demonstrated the differences between the nodal count and the
Neumann count, let us consider also uniting their forces. Namely,
what we can infer on the underlying graph if we know both its nodal
count sequence and its Neumann count sequence. A simple calculation,
taking into account that the minimal degree of an interior vertex
is at least three gives
\begin{align*}
V & \le2\beta+2\left|\partial\Gamma\right|-2\\
E & \le3\beta+2\left|\partial\Gamma\right|-3.
\end{align*}
The above shows that knowing the values of $\beta$ and $\left|\partial\Gamma\right|$
(which are stored in the nodal and Neumann count sequences), bounds
the number of vertices and edges and by this restricts the possible
candidates to a finite set of graphs. This is a huge progress in solving
the inverse problem, as it reduces the infinitely many possible graphs
(e.g., in the example of tree graphs mentioned above) to a finite
number.

A related inverse problem concerns isospectrality. Isospectral graphs
are graphs which share the same eigenvalues. It was conjectured that
such graphs would have different nodal count \cite{GnuSmiSon_jpa05},
or in other words that nodal count resolves isospectrality\footnote{The scope of the conjecture was actually broader than just for quantum
graphs and it was stated also for isospectral manifolds and isospectral
discrete graphs.}. This conjecture in its most general form have been refuted by now
in \cite{BanShaSmi_jpa06,OreBan_jpa12,JuuJoy_jphys18,BruFaj_cmp12}.
Yet, given the discussion above, one may ask whether isospectrality
is resolved by combining both the nodal count and the Neumann count.

The next step would be to investigate the joint probability distribution
of the nodal and Neumann surpluses. See Figure \ref{fig: correlations}.
This figure shows in particular that the random variables $\sigma$
and $\omega$ are dependent. It is interesting to study the correlation
coefficient between both, $\mathrm{corr}(\sigma,\omega):=\frac{\mathbb{E}\left[\left(\sigma-\mathbb{E}[\sigma]\right)\left(\omega-\mathbb{E}[\omega]\right)\right]}{\sqrt{\mathbb{V}\left[\sigma\right]\mathbb{V}\left[\omega\right]}}$,
where $\mathbb{E}$ and $\mathbb{V}$ indicate the corresponding expected
values and variances. In this context, see also in Appendix \ref{sec: Appendix-examples}
a related discussion on the bounds for $\omega-\sigma$ and for $\omega$.

\begin{figure}[h]
\includegraphics[width=0.7\paperwidth]{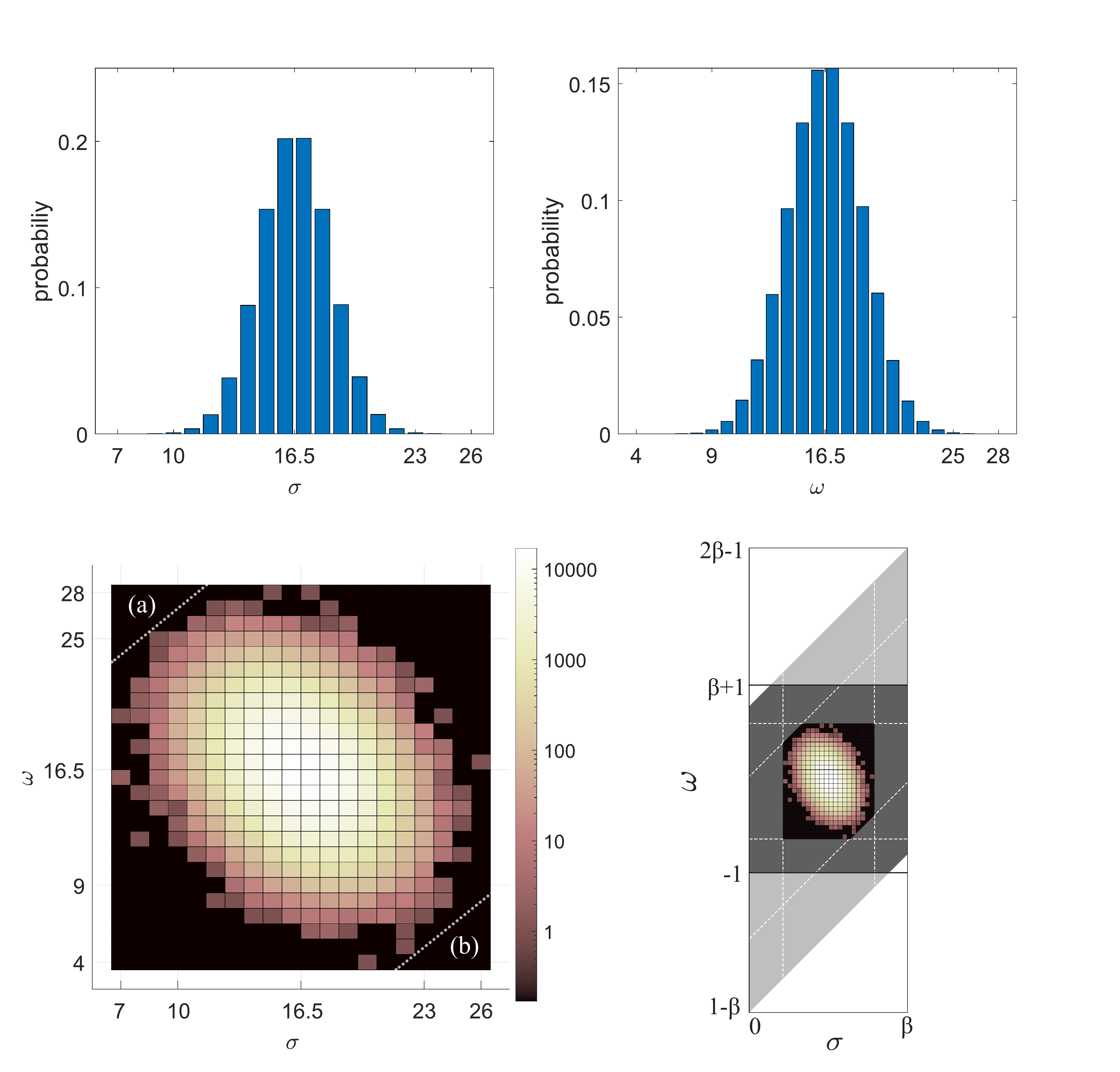} \centering{}\caption{The nodal and Neumann statistics for a random 6 regular graph with
16 vertices. The upper two figures show the probability distributions
of the nodal surplus, $\sigma$, and the Neumann surplus, $\omega$.
The lower left figure shows the joint probability distribution, as
a 2d histogram of the pair $\left(\sigma,\omega\right)$ in a log
scale. The support of $\left(\sigma,\omega\right)$ in this example
is bounded by $7\le\sigma\le26$, $4\le\omega\le28$ and $-16\le\sigma-\omega\le17$.
The white diagonal lines (a) and (b) in the lower left figure represent
the support of $\sigma-\omega$. In the lower right figure, the support
of $\left(\sigma,\omega\right)$ is presented with respect to the
bounds on $\omega$ and $\sigma-\omega$. The region which contains
the allowed values for $\sigma-\omega$ is shaded with gray colors;
light gray is for the $\omega$ bounds of Theorem \ref{thm:Neumann_surplus_main}
(\ref{enu:thm-Neumann_surplus_bounds}) and dark gray for Conjecture
\ref{conj:strict_bounds_Neumann_surplus}. The numerical data was
calculated for $\sim10^{6}$ eigenfunctions. }
\label{fig: correlations}
\end{figure}

We end the discussion of the Neumann surplus distribution with
\begin{conjecture}
\label{conj:universality_Neumann_count} Let $\left\{ \Gamma^{(m)}\right\} _{m=1}^{\infty}$
be a sequence of standard quantum graphs each with rationally independent
edge lengths. Denote by $\beta^{(m)}$ and $\left|\partial\Gamma^{(m)}\right|$
the first Betti number and the boundary size of the graph $\Gamma^{(m)}$,
correspondingly. Denote by $\omega^{(m)}$ the Neumann surplus random
variable of $\Gamma^{(m)}$. If we assume that the graphs in the sequence
increase, in the sense $\lim_{m\rightarrow\infty}\left(\beta^{(m)}+\left|\partial\Gamma^{(m)}\right|\right)=\infty$,
then
\begin{equation}
\frac{\omega^{(m)}-\mathbb{E}\left[\omega^{(m)}\right]}{\sqrt{\mathbb{V}\left[\omega^{(m)}\right]}}\xrightarrow[m\rightarrow\infty]{\mathcal{D}}N(0,1),\label{eq:universality_Neumann_count}
\end{equation}
where the convergence is in distribution and $N(0,1)$ is the standard
normal distribution.
\end{conjecture}

Observe that (\ref{eq:universality_Neumann_count}) describes a convergence
of finitely supported distributions, $\frac{\omega^{(m)}-\mathbb{E}\left[\omega^{(m)}\right]}{\sqrt{\mathbb{V}\left[\omega^{(m)}\right]}}$,
in to a continuous distribution, $N\left(0,1\right)$. This can only
happen if the size of the support of $\omega^{\left(m\right)}$ is
unbounded. The size of the support of $\omega^{\left(m\right)}$ is
bounded by $3\beta+\left|\dg\right|-2$ according to Theorem \ref{thm:Neumann_surplus_main}
(\ref{enu:thm-Neumann_surplus_bounds}), and by $\beta+\left|\dg\right|+2$
if Conjecture \ref{conj:strict_bounds_Neumann_surplus} is true. Both
bounds diverge if and only if $\beta+\left|\dg\right|$ diverges,
which is why we consider the limit, $\lim_{m\rightarrow\infty}\left(\beta^{(m)}+\left|\partial\Gamma^{(m)}\right|\right)=\infty$,
in the conjecture.

Note that by Corollary \ref{cor:Expected_values_of_nodal_and_Neumann_supluses}
it is known that $\mathbb{E}\left[\omega^{(m)}\right]=\frac{1}{2}(\beta^{(m)}-\left|\partial\Gamma^{(m)}\right|)$,
but we do not have a general expression for the variance. There are
a few sources of support towards this conjecture. First, by Theorem
\ref{Thm: (3,1)-regular-trees}, we get that the conjecture holds
for an increasing family of $(3,1)$-regular tree graphs. Indeed,
the Neumann surplus of each such graph is binomially distributed,
so that the central limit theorem guarantees the convergence to the
normal distribution, as in the conjecture. Furthermore, we tested
the validity of this conjecture by various numerical explorations.
We have checked increasing families of graphs such as random d-regular
graphs and complete graphs. For each such family we have calculated
the Kolmogorov-Smirnov distance from a normal distribution and observed
that this distance converges to zero in the limit of increasing graphs.
We also provide an analytic evidence to the above conjecture in an
ongoing work, where we prove that families of stowers and mandarins
(see Appendix \ref{sec: Appendix-examples}) satisfy the above conjecture.
Finally, both authors together with Berkolaiko conjecture a similar
statement for the nodal surplus. See \cite{AloBanBer_conj} for the
exact details, as well as supportive evidence. We believe that the
work on these two conjectures should be done in parallel and that
their confirmation might occur simultaneously.

We dedicate the last part of the discussion to the role of the local
observables; the spectral position and the wavelength capacity. Those
serve as a fundamental tool for studying Neumann domains. Analogous
observables are defined for Neumann domains on manifolds \cite{BanCoxEgg_arxiv20,BanCoxEgg_prep21,BanEggTay_jga20,BanFaj_ahp16}.
For manifolds, estimating the spectral position may serve as a tool
towards calculating the asymptotics of the Neumann count. But, determining
the spectral position of a Neumann domain on a manifold is a hard
task (even numerically). To aid this, it was shown in \cite{AloBanBerEgg_lms20}
that the wavelength capacity\footnote{The analogue of the wavelength capacity for manifolds is called the
area to perimeter ratio in \cite{AloBanBerEgg_lms20}.} is bounded from above in terms of the spectral position. Hence, computation
of the former allows to estimate the latter and tackle the problem
mentioned above. In the current paper we show that even a stronger
connection exists for quantum graphs - bounds from both sides of the
wavelength capacity are expressed in terms of the spectral position.

Taking the perspective of inverse problems, we consider the probability
distributions of the wavelength capacity, $\pi^{(v)}$ and $p^{(v)}$,
as described in Proposition \ref{prop: statistics_of_wavelength_capacity}.
We ask what information on the graph is stored in those distributions.
First, as those distributions are symmetric around $\frac{1}{2}\deg v$,
the degree of that particular vertex is revealed. But, those distributions
obviously do not solely depend on $\deg v$, as is demonstrated in
Figure \ref{fig: graph rho example}.
\begin{figure}[H]
\begin{raggedright}
\includegraphics[width=0.75\paperwidth]{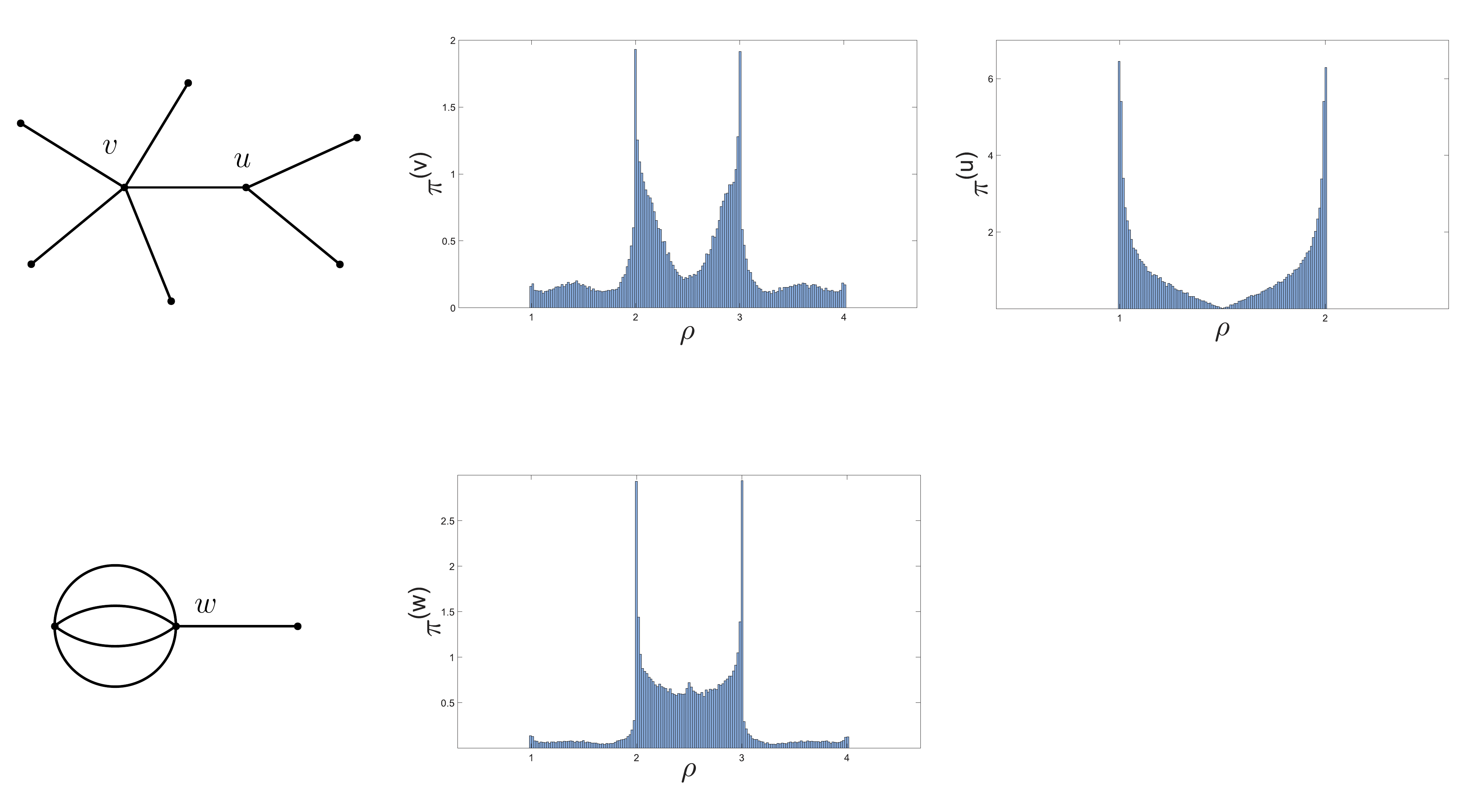}
\par\end{raggedright}
\centering{}\caption{The top line shows a graph with two marked vertices $v$ and $u$
such that $\protect\deg v=5$ and $\protect\deg u=3$. Next to the
graph are the probability distributions $\pi^{(v)}$ and $\pi^{(u)}$.
The bottom line shows a different graph with a marked vertex $w$
of $\protect\deg w=5$ and the probability distribution $\pi^{(w)}$.
The numerical data was calculated for $\sim10^{6}$ eigenfunctions
per graph.}
\label{fig: graph rho example}
\end{figure}
A stimulating problem would be to understand better the probability
distribution $\pi^{(v)}$ and the information it stores on the graph.
In particular, the exact values at which it gets minima and maxima
(see Figure \ref{fig: graph rho example}) are yet unknown. Another
numerical observation is that the probability distribution of the
wavelength capacity contains no atoms. Namely, $p^{(v)}\equiv0$.
We conjecture that this is always the case - see Remark \ref{rem: no_atoms_in_rho_distribution}.

\subsection*{Acknowledgments}

We acknowledge Michael Bersudsky and Sebastian Egger for their permission
to split this work from the review paper \cite{AloBanBerEgg_lms20}
co-authored with them and for the discussions in the course of the
common work. We thank Ron Rosenthal and Uri Shapira for the interesting
discussions throughout the work. We thank Adi Alon for her graphical
assistance with the figures.

The authors were supported by ISF (Grant No. 844/19) and by the Binational
Science Foundation Grant (Grant No. 2016281). LA was also supported
by the Ambrose Monell Foundation and the Institute for Advanced Study.

\appendix

\section{\label{sec: Appendix-examples} Nodal and Neumann surpluses of particular
examples}

An intensive numerical investigation led us to believe that the Neumann
surplus bounds in Theorem \ref{thm:Neumann_surplus_main} are not
optimal and to propose better bounds in Conjecture \ref{conj:strict_bounds_Neumann_surplus}.
In this appendix we provide analytic evidence supporting the conjectured
bounds, by calculating the support of the Neumann surplus for three
families of graphs:
\begin{enumerate}
\item Stowers - Proposition \ref{prop: stower Neumann surplus}.
\item Mandarins - Proposition \ref{prop: mandarins}.
\item Trees - Proposition \ref{prop: tree gets all surplus difference}.
\end{enumerate}
See Figure \ref{fig: stower and mandarin} for examples of a stower
and a mandarin. Using the support of the stowers we deduce that the
bounds on $\sigma-\omega$, as presented in (\ref{eq:bounds_on_nodal_Neumann_diff}),
are optimal in terms of $\beta$ and $\left|\dg\right|$.

\begin{figure}[H]
\begin{raggedright}
\centering{}\includegraphics[width=0.4\paperwidth]{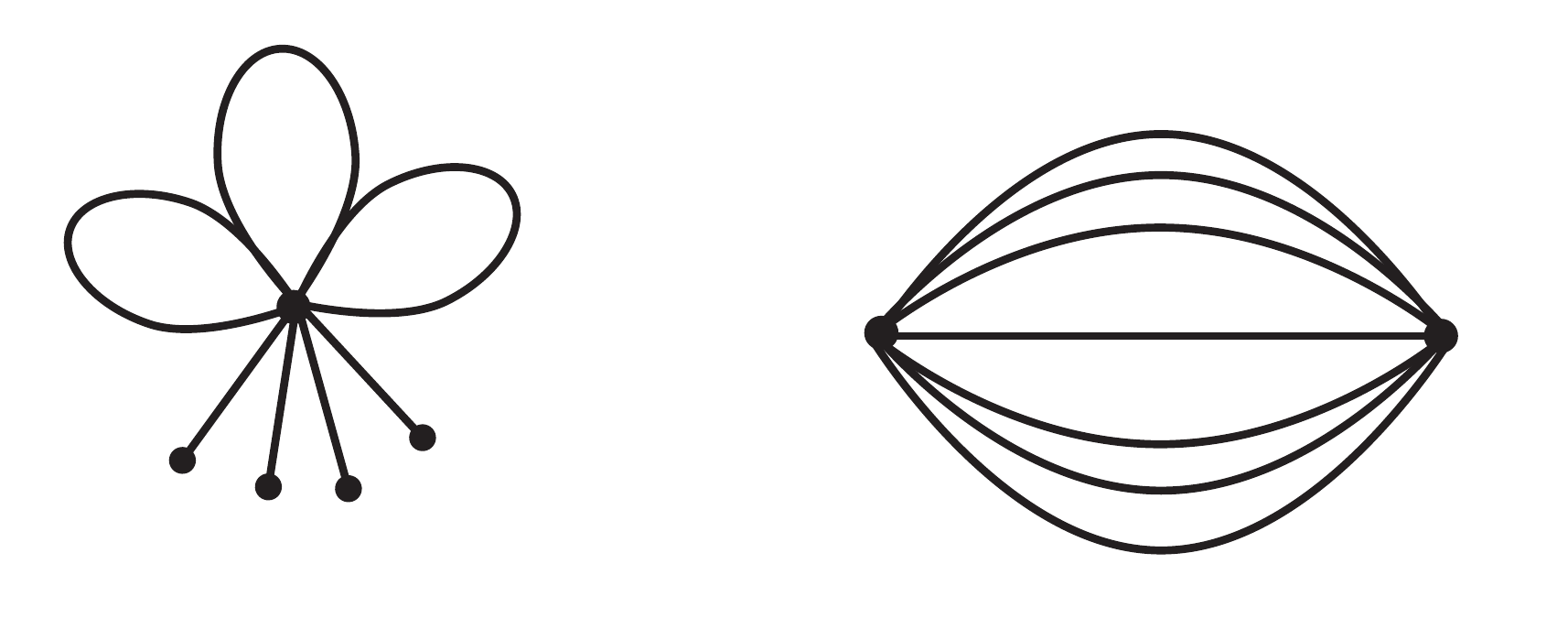}
\par\end{raggedright}
\centering{}\caption{On the left, a stower graph with $n=3$ loops and $m=4$ tails. On
the right, a mandarin graph with $E=7$ edges. }
\label{fig: stower and mandarin}
\end{figure}

\subsection{Stowers}

We say that $\Gamma$ is an $\left(n,m\right)$-stower graph with
$n$ loops (petals) and $m$ tails (leaves), if it has only one interior
vertex $v$ (the central vertex), $m$ boundary vertices, each connected
to the central vertex, and $n$ loops connecting the central vertex
to itself. See figure \ref{fig: stower and mandarin} for example.
In such case the first Betti number is $\beta=n$ and the boundary
size is $\left|\partial\Gamma\right|=m$.
\begin{prop}
\label{prop: stower Neumann surplus}Let $\Gamma$ be a stower, then
its Neumann surplus $\omega$ is bounded by
\[
1-\left|\partial\Gamma\right|\le\omega\le\beta-1.
\]
Moreover, if $\Gamma$ has rationally independent edge lengths then
its Neumann surplus and nodal surplus distributions satisfy
\begin{align*}
\forall j\in\left\{ 1-\left|\dg\right|,...,\beta-1\right\}  & \qquad P\left(\omega=j\right)>0\\
\forall j\in\left\{ 1-\beta,...,\beta+\left|\dg\right|-1\right\}  & \qquad P\left(\sigma-\omega=j\right)>0,
\end{align*}
and if $\left|\dg\right|>0$, then
\[
\forall j\in\left\{ 0,...,\beta\right\} \,\,\,\,\,P\left(\sigma=j\right)>0.
\]
In the case of $\left|\dg\right|=0$, namely $\Gamma$ is a flower
graph, the nodal surplus us bounded by
\[
1\le\sigma\le\beta-1,
\]
and

\[
\forall j\in\left\{ 1,...,\beta-1\right\} \,\,\,\,\,P\left(\sigma=j\right)>0.
\]
\end{prop}

This result shows that for any possible choice of $\beta$ and $\left|\dg\right|$
there is a corresponding stower graph such that its $\sigma-\omega$
sequence achieves both upper and lower bounds in (\ref{eq:bounds_on_nodal_Neumann_diff}).
This result supports Conjecture \ref{conj:strict_bounds_Neumann_surplus}
and shows that the bounds in (\ref{eq:bounds_on_nodal_Neumann_diff})
are optimal.

To prove Proposition \ref{prop: stower Neumann surplus} let us first
state the following.
\begin{defn}
\label{def: stowers bad sets}~Let $\Gamma$ be a stower with $m$
tails and $n$ loops. For distinction, denote points in $\torus$
by $\left(\yv,\zv\right)$ such that $\yv\in\mathbb{T}^{n}$ corresponds
to loops and $\zv\in\mathbb{T}^{m}$ corresponds to tails. Call a
tail coordinate $z_{j}$ ``bad'' if either $\sin(z_{j})=0$ or $\cos(z_{j})=0$.
Similarly, a loop coordinate $y_{i}$ is ``bad'' if either $\sin(\frac{y_{i}}{2})=0$
or $\cos(\frac{y_{i}}{2})=0$.

Denote the set of points in $\torus$ which have at least one ``bad''
coordinate by $\boldsymbol{B}^{\left(1\right)}$, and the similarly
denote the set of points with at least two ``bad'' coordinates by
$\boldsymbol{B}^{\left(2\right)}$.
\end{defn}

\begin{rem}
\label{rem: stowers - codimension}Notice that $\dim\left(\boldsymbol{B}^{\left(j\right)}\right)=E-j$.
\end{rem}

\begin{lem}
\label{lem: stowers explicit}Let $\Gamma$ be a stower with $m$
tails and $n$ loops and consider the notations of Definition (\ref{def: stowers bad sets}).
Then,
\end{lem}

\begin{enumerate}
\item \label{enu: stowers sigma reg }$\mgen$ is given by
\[
\mgen=\set{\left(\yv,\zv\right)\in\torus\setminus\boldsymbol{B}^{\left(1\right)}}{\sum_{j=1}^{m}\tan(z_{j})+2\sum_{i=1}^{n}\tan\left(\frac{y_{i}}{2}\right)=0}.
\]

\item \label{enu: stowers indicies}We denote for $\left(\yv,\zv\right)\in\mgen$
\begin{align*}
\bs i_{\textrm{tails}} & (\zv):=\left|\set{1\leq j\leq m}{\tan\left(z_{j}\right)<0}\right|,\\
\bs i_{\textrm{loops}} & (\yv):=\left|\set{1\leq i\leq n}{\tan\left(\frac{y_{i}}{2}\right)<0}\right|.
\end{align*}
Then the Neumann surplus and nodal surplus functions (introduced in
Lemma \ref{lem: functions_on_secular_manifold_existence_and_symmetry})
satisfy,
\begin{align*}
\boldsymbol{\omega}\left(\yv,\zv\right) & =n-(\bs i_{\textrm{tails}}\left(\zv\right)+\bs i_{\textrm{loops}}\left(\yv\right)),\\
\boldsymbol{\sigma}\left(\yv,\zv\right) & =\bs i_{\textrm{loops}}\left(\yv\right).
\end{align*}
\end{enumerate}
The proof of Lemma \ref{lem: stowers explicit} appears after the
proof of Proposition \ref{prop: stower Neumann surplus}:
\begin{proof}
[Proof of Proposition \ref{prop: stower Neumann surplus} ]According
to Theorem \ref{thm:Neumann_surplus_main} (\ref{enu:thm-density_of_Neumann_surplus-b}),
Theorem \ref{thm: equidistribution_by_BG_measure} and Lemma \ref{lem: functions_on_secular_manifold_existence_and_symmetry},
it is enough to prove that the functions $\boldsymbol{\sigma}$ and
$\boldsymbol{\omega}$, defined on $\mgen$, satisfy
\begin{align}
\imag\left(\boldsymbol{\omega}\right) & =\left\{ 1-\left|\dg\right|,...,\beta-1\right\} ,\label{eq: stower Image omega}\\
\imag\left(\boldsymbol{\sigma}-\boldsymbol{\omega}\right) & =\left\{ 1-\beta,...,\beta+\left|\dg\right|-1\right\} ,\label{eq: stower image sigma}
\end{align}
and
\begin{equation}
\imag\left(\boldsymbol{\sigma}\right)=\begin{cases}
\left\{ 0,...\beta\right\}  & \left|\dg\right|>0\\
\left\{ 1,...\beta-1\right\}  & \left|\dg\right|=0
\end{cases}.\label{eq: stowe image difference}
\end{equation}
Let $\Gamma$ be a stower with $m$ tails and $n$ loops and consider
the notations of Definition (\ref{def: stowers bad sets}). For every
point $\left(\yv,\zv\right)\in\torus$, define $t_{j},s_{i}\in\R\cup\left\{ \infty\right\} $
by $t_{j}:=\tan\left(z_{j}\right)$ and $s_{i}:=\tan\left(\frac{y_{i}}{2}\right)$
for all coordinates $z_{j}$ and $y_{i}$. Notice that $\left(\yv,\zv\right)\in\torus\setminus\boldsymbol{B}^{\left(1\right)}$
if and only if all $t_{j}$'s and $s_{j}$' lie in $\R\setminus\left\{ 0\right\} $.
According to Lemma \ref{lem: stowers explicit} (\ref{enu: stowers sigma reg }),
$\left(\yv,\zv\right)\in\mgen$ if and only if $\left(\yv,\zv\right)\in\torus\setminus\boldsymbol{B}^{\left(1\right)}$
and
\begin{equation}
\sum_{j=1}^{m}t_{j}+2\sum_{j=1+m}^{m+n}s_{j}=0.\label{eq: stower sum of t and s}
\end{equation}
The $\bs i_{\textrm{tails}}\left(\zv\right)$ and $\bs i_{\textrm{loops}}\left(\yv\right)$
indices, defined in Lemma \ref{lem: stowers explicit} (\ref{enu: stowers indicies}),
are equal to the number of negative $t_{j}$'s and $s_{j}$'s correspondingly.
Observe that
\begin{align}
0 & \le\bs i_{\textrm{tails}}\le m,\,\text{and}\label{eq: i tails}\\
0 & \le\bs i_{\textrm{loops}}\le n,\label{eq: i loops}
\end{align}
however a solution to (\ref{eq: stower sum of t and s}) with non-zero
$t_{j}$'s and $s_{j}$' must have at least one positive and one negative
summands, and therefore
\begin{equation}
1\le\bs i_{\textrm{tails}}+\bs i_{\textrm{loops}}\le m+n-1.\label{eq: itails plus iloops}
\end{equation}
Using Lemma \ref{lem: stowers explicit} (\ref{enu: stowers indicies})
and the bounds in (\ref{eq: i tails}), (\ref{eq: i loops}) and (\ref{eq: itails plus iloops}),
proves the corresponding bounds on $\boldsymbol{\omega},\,\boldsymbol{\sigma}$
and $\boldsymbol{\sigma}-\boldsymbol{\omega}$ and hence inclusion
in (\ref{eq: stower Image omega}),(\ref{eq: stower image sigma})
and (\ref{eq: stowe image difference}).

In order to prove the actual equalities in (\ref{eq: stower Image omega}),(\ref{eq: stower image sigma})
and (\ref{eq: stowe image difference}), namely that every value is
attained, we show that for any integers $i_{t}$ and $i_{s}$ satisfying
\begin{align*}
0 & \le i_{t}\le m,\\
0 & \le i_{s}\le n,\,\text{and}\\
1 & \le i_{t}+i_{s}\le m+n-1,
\end{align*}
there exist $\left(\yv,\zv\right)\in\mgen$ for which $\bs i_{\textrm{tails}}\left(\zv\right)=i_{t}$
and $\bs i_{\textrm{loops}}\left(\yv\right)=i_{s}$. Clearly, one
can construct a solution to (\ref{eq: stower sum of t and s}) for
which all $t_{j}$'s and $s_{j}$' lie in $\R\setminus\left\{ 0\right\} $
and there are exactly $i_{t}$ negative $t_{j}$'s and $i_{s}$ negative
$s_{j}$'s. Consider $\tan^{-1}:\R\setminus\left\{ 0\right\} \rightarrow\left(0,\frac{\pi}{2}\right)\cup\left(\frac{\pi}{2},\pi\right)$
and define $z_{j}=\tan^{-1}\left(t_{j}\right)$ for all $j\le m$
and $y_{i}=2\tan^{-1}\left(s_{i}\right)$ for $1\le i\le n$. By construction,
$\left(\yv,\zv\right)\in\mgen$ and satisfies $\bs i_{\textrm{tails}}\left(\zv\right)=i_{t}$
and $\bs i_{\textrm{loops}}\left(\yv\right)=i_{s}$.
\end{proof}
It remains to prove Lemma \ref{lem: stowers explicit}:
\begin{proof}
[Proof of Lemma \ref{lem: stowers explicit} ] The first part of
the lemma is an explicit characterization $\mgen$. To do so consider
the standard graph $\Gamma_{\left(\yv,\zv\right)}$, i.e., with edge
lengths $\lv=\left(\yv,\zv\right)$. In the following, we construct
an eigenfunction $f$ of $\Gamma_{\kv}$ with eigenvalue $k=1$, providing
the conditions on $\left(\yv,\zv\right)\in\torus$ for which such
an eigenfunction exists.

Let $v$ be the central vertex of $\Gamma$, and consider a parametrization
of each tail $e_{j}$ by arc-length parametrization $x_{j}\in\left[0,z_{j}\right]$
with $x_{j}=0$ at $v$. Consider a similar parametrization for each
loop $e_{i}$ but let $x_{i}\in\left[-\frac{y_{i}}{2},\frac{y_{i}}{2}\right]$
such that $x_{i}=\pm\frac{y_{i}}{2}$ at $v$. Let $f$ be an eigenfunction
of eigenvalue $k=1$, then its restriction to every tail $e_{j}$
can be written as:
\begin{align}
f|_{e_{j}}\left(x_{j}\right) & =A_{j}\cos\left(z_{j}-x_{j}\right),\label{eq: stower f restriction}
\end{align}
and its restriction to every loop $e_{i}$ can be written as:

\begin{align}
f|_{e_{i}}\left(x_{i}\right) & =C_{i}\cos\left(x_{i}\right)+B_{i}\sin\left(x_{i}\right).\label{eq: stower f restriction-1}
\end{align}
For every loop $e_{i}$, consider the inversion defined by $x_{i}\mapsto-x_{i}$
on $e_{i}$ while fixing $\Gamma\setminus e_{i}$. It is an isometry
of $\Gamma_{\left(\yv,\zv\right)}$ and so we may choose all eigenfunctions
of $\Gamma_{\left(\yv,\zv\right)}$ to be either symmetric or anti-symmetric
with respect to this inversion. Anti-symmetric eigenfunctions are
loop-eigenfunctions (as defined in Subsection \ref{subsec: loop_eigenfunctions_and_generic_eigenfunctions})
and are not generic eigenfunctions (in the sense of Definition \ref{def: generic_eigenfunction}).
A loop-eigenfunction which is supported on $e_{i}$ exists if and
only if $y_{i}=2\pi$, in which case $\left(\yv,\zv\right)\in\boldsymbol{B}^{\left(1\right)}$.

We call $f$ symmetric if it is symmetric on every loop. If $f$ is
symmetric then all $B_{i}$'s in (\ref{eq: stower f restriction-1})
are zero. For symmetric eigenfunctions, the Neumann vertex condition
at $v$ can be written as,
\begin{align}
\forall j\le m\qquad & A_{j}\cos\left(z_{j}\right)=f\left(v\right),\label{eq: stower continuity tail}\\
\forall i\le n\qquad & C_{i}\cos\left(\frac{y_{i}}{2}\right)=f\left(v\right),\,\,\text{ and}\label{eq: stower continuity loop}\\
-k\big(\sum_{j=1}^{m}A_{j}\sin\left(z_{j}\right) & +2\sum_{i=1}^{m}C_{i}\sin\left(\frac{y_{i}}{2}\right)\big)=0.\label{eq: stower derivatives}
\end{align}
In particular, for (\ref{eq: stower derivatives}) to hold there must
be at least two non zero amplitudes among all $A_{j}$'s and $C_{i}$'s.
We may deduce from (\ref{eq: stower continuity tail}) and (\ref{eq: stower continuity loop})
that if $f$ is symmetric with $f\left(v\right)=0$, then $\left(\yv,\zv\right)\in\boldsymbol{B}^{\left(2\right)}$.

A symmetric eigenfunction $f$, with $f\left(v\right)\ne0$, implies
that $\cos\left(z_{j}\right)\ne0$ for every tail and $A_{j}=\frac{f\left(v\right)}{\cos\left(z_{j}\right)}$
by (\ref{eq: stower continuity tail}). Similarly, for every loop,
$\cos\left(\frac{y_{i}}{2}\right)\ne0$ and $C_{i}=\frac{f\left(v\right)}{\cos\left(\frac{y_{i}}{2}\right)}$
by (\ref{eq: stower continuity loop}). Dividing (\ref{eq: stower derivatives})
by $-kf\left(v\right)\ne0$, gives
\begin{equation}
\sum_{j=1}^{m}\tan\left(x_{j}\right)+2\sum_{i=1}^{n}\tan\left(\frac{y_{i}}{2}\right)=0.\label{eq: derivative condition-1}
\end{equation}
In such case, the outgoing derivatives of $f$ at $v$ do not vanish
if and only if,
\begin{equation}
\forall j\le m\,\,\,\sin\left(z_{j}\right)\ne0\,\,\text{ and }\,\,\forall1\le i\le n\,\,\sin\left(\frac{y_{i}}{2}\right)\ne0.\label{eq: stower sin derivatives-1-1}
\end{equation}
We may now deduce that $f$ is not generic if $\left(\yv,\zv\right)\in\boldsymbol{B}^{\left(1\right)}$.
Moreover, if $\left(\yv,\zv\right)\notin\boldsymbol{B}^{\left(1\right)}$
then $\Gamma_{\left(\yv,\zv\right)}$ has an eigenfunction $f$ of
eigenvalue $k=1$ if and only if (\ref{eq: derivative condition-1})
holds, and in such case the construction implies that $f$ is unique
(up to a scalar multiplication), $f\left(v\right)\ne0$, and non of
the outgoing derivatives of $f$ at $v$ vanish. Namely, $f$ is generic.
This proves Lemma \ref{lem: stowers explicit} (\ref{enu: stowers sigma reg }),
that is:
\[
\mgen=\set{\left(\zv,\yv\right)\in\torus\setminus\boldsymbol{B}^{\left(1\right)}}{\sum_{j=1}^{m}\tan\left(x_{j}\right)+2\sum_{i=1}^{n}\tan\left(\frac{y_{i}}{2}\right)=0}.
\]
To prove Lemma \ref{lem: stowers explicit} (\ref{enu: stowers indicies}),
let $\left(\yv,\zv\right)\in\mgen$ and let $f$ be the corresponding
generic eigenfunction of eigenvalue $k=1$. Since $f$ must be symmetric,
all $B_{i}$'s in are zero, and we may derive the nodal and Neumann
counts from (\ref{eq: stower f restriction}) and (\ref{eq: stower f restriction-1}):

\begin{align}
\phi\left(f\right) & =\left|\set{j\le m}{z_{j}>\frac{\pi}{2}}\right|+\left|\set{j\le m}{z_{j}>\frac{3\pi}{2}}\right|+...\label{eq: nodal count for stower}\\
 & ...+2\left|\set{1\le i\le n}{y_{i}>\pi}\right|,\\
\xi\left(f\right) & =\left|\set{j\le m}{z_{j}>\pi}\right|+n.\label{eq: Neumann count for stower}
\end{align}
The nodal and Neumann surpluses are given by subtracting the spectral
position $N\left(\zv,\yv\right)$ from the above. The spectral position
is defined by,
\begin{align*}
N\left(\yv,\zv\right):= & \left|\set{k\in\opcl{0,1}}{k\,\,\text{is an eigenvalue of }\Gamma_{\left(\yv,\zv\right)}}\right|,\\
= & \left|\set{t\in\opcl{0,1}}{\left(t\yv,t\zv\right)\in\Sigma}\right|,
\end{align*}
both counts including multiplicity. In order to avoid multiplicities
and to ease the computations define the ``good'' set:
\[
M=\set{\left(\yv,\zv\right)\in\mgen}{\forall t\in\opcl{0,1}\,\,\,\,\,\,\left(t\yv,t\zv\right)\notin\boldsymbol{B}^{\left(2\right)}\cup\msing}.
\]
Let $\left(\yv,\zv\right)\in M$ and let $t\in\left(0,1\right)$ such
that $\left(t\yv,t\zv\right)\in\Sigma$, then $\left(t\yv,t\zv\right)\in\mreg$
and corresponds to an eigenfunction that satisfies $f\left(v\right)\ne0$
since $\left(t\yv,t\zv\right)\notin\boldsymbol{B}^{\left(2\right)}$
and $ty_{i}<2\pi$ for all $i\le n$. Therefore, $N\left(\yv,\zv\right)$
for $\left(\yv,\zv\right)\in M$ is given by
\[
N\left(\yv,\zv\right)=\left|\set{t\in\opcl{0,1}}{\sum_{j=1}^{m}\tan\left(tz_{j}\right)+2\sum_{i=1}^{n}\tan\left(t\frac{y_{i}}{2}\right)=0}\right|.
\]
Define $F_{\left(\yv,\zv\right)}:\opcl{0,1}\rightarrow\R\cup\left\{ \infty\right\} $
by
\[
F_{\left(\yv,\zv\right)}\left(t\right)=\sum_{j=1}^{m}\tan\left(tz_{j}\right)+2\sum_{i=1}^{n}\tan\left(t\frac{y_{i}}{2}\right).
\]
If $\left(\yv,\zv\right)\in M$ then only one of the summands of $F_{\left(\yv,\zv\right)}\left(t\right)$
can diverge at each time $t\in\opcl{0,1}$ since $\left(t\yv,t\zv\right)\notin\boldsymbol{B}^{\left(2\right)}$.
Therefore, the poles of $F_{\left(\yv,\zv\right)}\left(t\right)$
are simple. Since $\frac{d}{dt}F_{\left(\yv,\zv\right)}\left(t\right)>0$
whenever $t$ is not a pole then the zeros of $F_{\left(\yv,\zv\right)}$
are simple and they interlace with the poles. Since the interlacing
sequence of zeros and poles ends with a zero at $F_{\left(\yv,\zv\right)}\left(1\right)=0$,
we can find $N\left(\yv,\zv\right)$ by counting poles:
\begin{align*}
N\left(\yv,\zv\right) & =\left|\set{j\le m}{z_{j}>\frac{\pi}{2}}\right|+\left|\set{j\le m}{z_{j}>\frac{3\pi}{2}}\right|+....\\
... & +\left|\set{1\le i\le n}{y_{i}>\pi}\right|.
\end{align*}
Subtracting $N\left(\yv,\zv\right)$ from (\ref{eq: nodal count for stower})
and (\ref{eq: Neumann count for stower}) gives the nodal and Neumann
surplus for any $\left(\yv,\zv\right)\in M$:
\begin{align*}
\boldsymbol{\omega}\left(\yv,\zv\right) & =n+\left|\set{j\le m}{z_{j}>\pi}\right|-N\left(\zv,\yv\right),\\
\boldsymbol{\sigma}\left(\yv,\zv\right) & =\left|\set{1\le i\le n}{y_{i}>\pi}\right|,
\end{align*}
which can be rewritten, using $\bs i_{\textrm{tails}}$ and $\bs i_{\textrm{loops}}$
(as defined in Lemma \ref{lem: stowers explicit} (\ref{enu: stowers indicies})),
by
\begin{align}
\boldsymbol{\omega}\left(\yv,\zv\right) & =n-\left(\bs i_{\textrm{tails}}\left(\zv\right)-\bs i_{\textrm{loops}}\left(\yv\right)\right).\label{eq: Neumann surplus stower}\\
\boldsymbol{\sigma}\left(\yv,\zv\right) & =\bs i_{\textrm{loops}}\left(\yv\right).
\end{align}
This proves the statement of Lemma \ref{lem: stowers explicit} (\ref{enu: stowers indicies})
for points in $M$.

In order to extend this result from $M$ to $\mgen$ we recall that
$\boldsymbol{\sigma}$ and $\boldsymbol{\omega}$ are constant on
each connected component of $\mgen$ (Lemma \ref{lem: functions_on_secular_manifold_existence_and_symmetry}
and Theorem \ref{thm: surplus_function_on_secular_mnfld}). Observe
that $\bs i_{\textrm{tails}}$ and $\bs i_{\textrm{loops}}$ are constant
on connected components of $\torus\setminus\boldsymbol{B}^{\left(1\right)}$
and are therefore constant on connected components of $\mgen\subset\torus\setminus\boldsymbol{B}^{\left(1\right)}$
as well. It is thus enough to show that $M$ intersects each of the
connected components of $\mgen$ to conclude that Lemma \ref{lem: stowers explicit}
(\ref{enu: stowers indicies}) holds. To do so, we prove that $M$
is dense in $\mgen$.

Consider the open set $O=\mgen\setminus\overline{M}$ so that we need
to prove that $O=\emptyset$. Consider the cone $PO:=\set{\left(t\yv,t\zv\right)}{t\in\opcl{0,1},\,\left(\yv,\zv\right)\in O}$.
By the definition of $PO$, it is a union of lines of the form $\set{\left(t\yv,t\zv\right)}{t\in\opcl{0,1}}$
for some $\left(\yv,\zv\right)\in O$. Each such line intersect $\boldsymbol{B}^{\left(2\right)}\cup\msing$,
by the definition of $O$, and each two such lines are either disjoint
or one is contained in the other. It follows that the dimension of
$PO$ is bounded by
\begin{equation}
\dim\left(PO\right)\le\dim\left(\boldsymbol{B}^{\left(2\right)}\cup\msing\right)+1=E-1.\label{eq: dim PO<E}
\end{equation}
In the last equality we use Remark \ref{rem: stowers - codimension}
which states that $\dim\left(\boldsymbol{B}^{\left(2\right)}\right)=E-2$
and the fact that $\dim\left(\msing\right)\le E-2$. On the other
hand, by Lemma \ref{lem: stowers explicit} (\ref{enu: stowers sigma reg }),
\[
\mgen=\set{\left(\yv,\zv\right)\in\torus\setminus\boldsymbol{B}^{\left(1\right)}}{F_{\left(\yv,\zv\right)}\left(1\right)=0},
\]
and since $\frac{d}{dt}F_{\left(\yv,\zv\right)}\left(t\right)|_{t=1}>0$
on every $\left(\yv,\zv\right)\in\mgen$, then each of the lines in
$PO$ intersect $\mgen$ transversely. It follows that
\[
\dim\left(PO\right)=\dim\left(O\right)+1.
\]
It follows from (\ref{eq: dim PO<E}) that $\dim\left(O\right)\le E-2$.
However, $O$ is open in $\mgen$ which is of $\dim\mgen=E-1$ so
$O=\emptyset$. Therefore, $M$ is dense in $\mgen$.
\end{proof}

\subsection{Mandarins. }

A mandarin graph $\Gamma$ is a graph with only two vertices and $E\ge3$
edges such that every edge is connected to both vertices. See figure
\ref{fig: stower and mandarin} for example. In such case $\beta=E-1$
and $\partial\Gamma=\emptyset$.
\begin{prop}
\label{prop: mandarins}If $\Gamma$ is a mandarin then its Neumann
surplus is bounded by
\[
0\le\omega\le\beta,
\]
its nodal surplus is bounded by
\[
1\le\sigma\le\beta-1,
\]
and the difference $\sigma-\omega$ is supported inside the set $\left\{ 1-\beta,3-\beta,5-\beta,...,-1+\beta\right\} $.
\\
Moreover, if $\Gamma$ has rationally independent edge lengths, then
its nodal surplus and Neumann surplus distributions satisfy
\begin{align*}
\forall j\in\left\{ 0,1,...\beta-1\right\} \qquad & P\left(\omega=j\right)+P\left(\omega=j+1\right)>0,\\
\forall j\in\left\{ 1,...\beta-2\right\} \qquad & P\left(\sigma=j\right)+P\left(\sigma=j+1\right)>0,\,\,\text{and}\\
\forall j\in\left\{ 1-\beta,3-\beta,5-\beta,...,-1+\beta\right\} \qquad & P\left(\sigma-\omega=j\right)>0.
\end{align*}
\end{prop}

This example supports Conjecture \ref{conj:strict_bounds_Neumann_surplus}
and shows that the bounds in (\ref{eq:bounds_on_nodal_Neumann_diff})
can be achieved. To prove Proposition \ref{prop: mandarins} let us
first state the following.
\begin{rem}
The bound $1\le\sigma\le\beta-1$ for mandarin graphs was already
shown in \cite{BanBerWey_jmp15}.
\end{rem}

\begin{defn}
\label{def: stowers bad sets-1}~Let $\Gamma$ be a mandarin graph
with $E$ edges. Call a coordinate $\kappa_{j}$ ``bad'' if either
$\sin(\frac{\kappa_{j}}{2})=0$ or $\cos(\frac{\kappa_{j}}{2})=0$.
Denote the set of points in $\torus$ which have at least one ``bad''
coordinate by $\boldsymbol{B}^{\left(1\right)}$, and the similarly
denote the set of points with at least two ``bad'' coordinates by
$\boldsymbol{B}^{\left(2\right)}$.
\end{defn}

\begin{lem}
\label{lem: mandarins explicit}Let $\Gamma$ be a mandarin with $E$
edges. Define two functions on $\torus$ by,
\begin{align}
F_{s}\left(\kv\right) & :=\sum_{j=1}^{E}\tan\left(\frac{\kappa_{j}}{2}\right),\,\text{ and}\\
F_{a}\left(\kv\right) & :=\sum_{j=1}^{E}\cot\left(\frac{\kappa_{j}}{2}\right).
\end{align}
Then $\mgen$ has a disjoint decomposition $\mgen=\mgen_{s}\sqcup\mgen_{a}$
such that,
\end{lem}

\begin{enumerate}
\item \label{enu: mandarin sigma reg}$\mgen_{s}$ is given by
\[
\mgen_{s}=\set{\kv\in\torus\setminus\boldsymbol{B}^{\left(1\right)}}{F_{s}\left(\kv\right)=0,\,\text{ and }\,F_{a}\left(\kv\right)\ne0},
\]
and $\mgen_{a}$ is given by
\[
\mgen_{s}=\set{\kv\in\torus\setminus\boldsymbol{B}^{\left(1\right)}}{F_{s}\left(\kv\right)\ne0,\,\text{ and }\,F_{a}\left(\kv\right)=0}.
\]
In particular, the inversion $T\left(\kv\right):=\modp{\kv+\left(\pi,\pi,\pi....\right)}$,
is an isometry between the two components, $T\left(\mgen_{s}\right)=\mgen_{a}$.
\item \label{enu: mandarin indicies}Denote for $\kv\in\mgen$
\begin{align*}
\bs i(\kv) & :=\left|\set{1\leq j\leq m}{\tan(\frac{\kappa_{j}}{2})<0}\right|,\\
C\left(\kv\right) & :=\begin{cases}
1 & F_{a}\left(\kv\right)\le0\\
0 & F_{a}\left(\kv\right)>0
\end{cases}.
\end{align*}
The Neumann surplus and nodal surplus functions (introduced in Lemma
\ref{lem: functions_on_secular_manifold_existence_and_symmetry} and
Theorem \ref{thm: surplus_function_on_secular_mnfld}) and satisfy,
for $\kv\in\mgen_{s}$,
\begin{align*}
\boldsymbol{\sigma}\left(\kv\right) & =\boldsymbol{i}\left(\kv\right)-C\left(\kv\right),\,\,\,\text{and}\\
\boldsymbol{\omega}\left(\kv\right) & =E-\boldsymbol{i}\left(\kv\right)-C\left(\kv\right).
\end{align*}
For $\kv\in\mgen_{a}$, $T\left(\kv\right)\in\mgen_{s}$ and
\begin{align*}
\boldsymbol{\sigma}\left(\kv\right) & =\boldsymbol{\sigma}\left(T\left(\kv\right)\right),\,\,\,\text{and}\\
\boldsymbol{\omega}\left(\kv\right) & =\boldsymbol{\omega}\left(T\left(\kv\right)\right).
\end{align*}
\end{enumerate}
\begin{rem}
The $s,a$ labeling of $\mgen_{s},\mgen_{a}$ stands for symmetric
and anti-symmetric, as these parts are defined according to a certain
symmetry of the eigenfunctions. It will be introduced in the proof
of the Lemma.
\end{rem}

The proof of Lemma \ref{lem: mandarins explicit} appears after the
proof of the Proposition \ref{prop: mandarins}:
\begin{proof}
[Proof of Proposition \ref{prop: mandarins} ] Define $t_{j}:=\tan\left(\frac{\kappa_{j}}{2}\right)$
so that $\kv\in\mgen_{s}$ if and only if $t_{j}\in\R\setminus\left\{ 0\right\} $
for every $j$ and the following conditions hold:
\begin{align}
\sum_{j=1}^{E}t_{j} & =0,\,\text{and}\\
\sum_{j=1}^{E}\frac{1}{t_{j}} & \ne0.
\end{align}
In such case $\boldsymbol{i}\left(\kv\right)$ is the number of negative
$t_{j}$'s, and it is not hard to show that for any $1\le i\le E-1$,
one can construct a tuple of $t_{j}$'s that satisfy the above conditions
such that the number of negative $t_{j}$'s is equal to $i$. Namely,
there exist $\kv\in\mgen_{s}$ with $\boldsymbol{i}\left(\kv\right)=i$
for any $1\le i\le E-1=\beta$. In such case, by Lemma \ref{lem: mandarins explicit}
(\ref{enu: mandarin indicies}),
\begin{align*}
\boldsymbol{\sigma}\left(\kv\right) & =i-C\left(\kv\right)\in\left\{ i,i-1\right\} ,\\
\boldsymbol{\omega}\left(\kv\right) & =E-i-C\left(\kv\right)\in\left\{ \beta+1-i,\beta-i\right\} ,\,\,\,\text{and so}\\
\boldsymbol{\sigma}\left(\kv\right)-\boldsymbol{\omega}\left(\kv\right) & =2i-E.
\end{align*}
It follows that the image of $\boldsymbol{\sigma}-\boldsymbol{\omega}$
is $\left\{ 2-E,4-E...,E-2\right\} $, and that the Neumann surplus
is bounded by
\[
0\le\boldsymbol{\omega}\le\beta.
\]
Moreover, it follows that for any $j\in\left\{ 0,1,...\beta-1\right\} $,
$\imag\left(\boldsymbol{\omega}\right)$ contains at least one of
$j,j+1$. Same for $\imag\left(\boldsymbol{\sigma}\right)$. Due to
the $T$ invariance the restriction of each of the above functions
to $\mgen_{a}$ has the same image as the restriction to $\mgen_{s}$.
This ends the proof of Proposition \ref{prop: mandarins} by the same
arguments used in the proof of Proposition \ref{prop: stower Neumann surplus}.

As already stated, the bounds $1\le\sigma\le\beta-1$ was shown in
\cite{BanBerWey_jmp15}, but for completeness we will prove that $\boldsymbol{\sigma}\left(\kv\right)\ne0$
from which the above bounds follows from symmetry arguments. Assume
by contradiction that there is a point $\kv\in\mgen_{s}$ for which
$\boldsymbol{\sigma}\left(\kv\right)=0$. It follows that $\boldsymbol{i}\left(\kv\right)=1$
and $C\left(\kv\right)=1$. Without loss of generality assume that
the negative $t_{j}$ is $t_{1}$ so that the $t_{j}$'s satisfy
\[
\left|t_{1}\right|=\sum_{j=2}^{E}t_{j},
\]
and $C\left(\kv\right)=1$ implies that
\[
\sum_{j=2}^{E}\frac{1}{t_{j}}\le\frac{1}{\left|t_{1}\right|}=\frac{1}{\sum_{j=2}^{E}t_{j}}.
\]
As the $t_{j}$'s for $j\ge2$ are strictly positive, we get
\[
\frac{\sum_{j=2}^{E}t_{j}}{E-1}\le\frac{1}{E-1}\frac{1}{\sum_{j=2}^{E}\frac{1}{t_{j}}}<\frac{E-1}{\sum_{j=2}^{E}\frac{1}{t_{j}}},
\]
which contradicts the means inequality.
\end{proof}
The proof of Lemma \ref{lem: mandarins explicit} is very similar
to that of Lemma \ref{lem: stowers explicit} and we will therefore
provide less details.~
\begin{proof}
[Proof of Lemma \ref{lem: mandarins explicit} ] Let $\Gamma$ be
a mandarin graph, and let $v_{-}$ and $v_{+}$ denote its two vertices.
Consider $\Gamma_{\kv}$ for some $\kv\in\torus$ with arc-length
parametrization $x_{j}\in\left[-\frac{\kappa_{j}}{2},\frac{\kappa_{j}}{2}\right]$
for every edge $e_{j}$, such that $x_{j}=\pm\frac{\kappa_{j}}{2}$
at $v_{\pm}$. As discussed in \cite{BanBerWey_jmp15}, the inversion
$x_{j}\mapsto-x_{j}$ for all edges simultaneously, is an isometry
of $\Gamma_{\kv}$ and we can choose all eigenfunctions to be symmetric\textbackslash anti-symmetric.
In particular, define
\begin{align*}
\mgen_{s} & :=\set{\kv\in\mgen}{f_{\kv}\,\text{ is symmetric}},\,\text{ and}\\
\mgen_{a} & :=\set{\kv\in\mgen}{f_{\kv}\,\text{ is anti-symmetric}}.
\end{align*}
A symmetric eigenfunction $f^{\left(s\right)}$ of eigenvalue $k=1$
has the structure:
\begin{equation}
f^{\left(s\right)}|_{e_{j}}\left(x_{j}\right)=A_{j}\cos\left(x_{j}\right).\label{eq: mandarin f sym}
\end{equation}
As in the proof of Lemma \ref{lem: stowers explicit} if $f^{\left(s\right)}\left(v_{\pm}\right)=0$
then $\kv\in\boldsymbol{B}^{\left(2\right)}$, and if $f^{\left(s\right)}\left(v_{\pm}\right)\ne0$
then
\begin{equation}
\forall j\le E\,\qquad\cos\left(\kappa_{j}\right)\ne0,\label{eq: mandarin symmetric cos}
\end{equation}
and the Neumann vertex condition (on each of the vertices) gives
\[
\sum_{j=1}^{E}\tan\left(\kappa_{j}\right)=F_{s}\left(\kv\right)=0.
\]
In such case, the outgoing derivatives of $f^{\left(s\right)}$ does
not vanish if and only if
\begin{equation}
\forall j\le E\,\qquad\sin\left(\kappa_{j}\right)\ne0.\label{eq: mandarin symmetric sin}
\end{equation}
In particular, as in Lemma \ref{lem: stowers explicit}, in the case
of $\kv\notin\boldsymbol{B}^{\left(2\right)}$ there is a unique (up
to scalar multiplication) construction of a $k=1$ symmetric eigenfunction
with $f^{\left(s\right)}\left(v_{\pm}\right)\ne0$ is and only if
$F_{s}\left(\kv\right)=0$. Moreover, if the eigenfunction is generic
then (\ref{eq: mandarin symmetric cos}) and (\ref{eq: mandarin symmetric sin})
implies that $\kv\notin\boldsymbol{B}^{\left(1\right)}$.

An anti-symmetric eigenfunction $f^{\left(a\right)}$ of eigenvalue
$k=1$ has the structure:
\begin{equation}
f^{\left(a\right)}|_{e_{j}}\left(x_{j}\right)=B_{j}\sin\left(x_{j}\right).\label{eq: mandarin f anti sym}
\end{equation}
A similar argument would show that in the case of $\kv\notin\boldsymbol{B}^{\left(2\right)}$
there is a unique (up to scalar multiplication) construction of a
$k=1$ anti-symmetric eigenfunction with $f^{\left(a\right)}\left(v_{\pm}\right)\ne0$
is and only if $F_{a}\left(\kv\right)=0$. Moreover, if the eigenfunction
is generic then (\ref{eq: mandarin symmetric cos}) and (\ref{eq: mandarin symmetric sin})
implies that $\kv\notin\boldsymbol{B}^{\left(1\right)}$.

We may conclude that,
\begin{align*}
\mgen_{s} & =\set{\kv\in\torus\setminus\boldsymbol{B}^{\left(1\right)}}{F_{s}\left(\kv\right)=0,\,\text{ and }\,F_{a}\left(\kv\right)\ne0},\,\text{ and}\\
\mgen_{a} & =\set{\kv\in\torus\setminus\boldsymbol{B}^{\left(1\right)}}{F_{s}\left(\kv\right)\ne0,\,\text{ and }\,F_{a}\left(\kv\right)=0}.
\end{align*}
Observe that $\begin{pmatrix}\cos\left(\frac{\kappa+\pi}{2}\right)\\
\sin\left(\frac{\kappa+\pi}{2}\right)
\end{pmatrix}=\begin{pmatrix}-\sin\left(\frac{\kappa}{2}\right)\\
\cos\left(\frac{\kappa}{2}\right)
\end{pmatrix}$, and therefore $\boldsymbol{B}^{\left(1\right)}$ is invariant to
$T$ and $F_{s}\circ T=-F_{a}$. It follows that
\[
T\left(\mgen_{s}\right)=\mgen_{a}.
\]
This proves Lemma \ref{lem: mandarins explicit} (\ref{enu: mandarin sigma reg}).
To prove Lemma \ref{lem: mandarins explicit} (\ref{enu: mandarin indicies}),
consider $\boldsymbol{i}\left(\kv\right):=\left|\set{j\le E}{\tan\left(\frac{\kappa_{j}}{2}\right)<0}\right|$
and let $\kv\in\mgen_{s}$ so that $f_{\kv}$ is symmetric and generic.
Using (\ref{eq: mandarin f sym}) we derive the nodal and Neumann
counts:

\begin{align}
\phi\left(f_{\kv}\right) & =2\left|\set{j\le E}{\frac{\kappa_{j}}{2}>\frac{\pi}{2}}\right|=2\boldsymbol{i}\left(\kv\right),\,\text{and}\label{eq: mandarin nodal count sym}\\
\xi\left(f_{\kv}\right) & =E.\label{eq: mandarin Neumann count sym}
\end{align}
Similarly, for $\kv\in\mgen_{a}$, using (\ref{eq: mandarin f anti sym}),
\begin{align}
\phi\left(f_{\kv}\right) & =E,\,\text{and}\label{eq: mandarin nodal count anti sym}\\
\xi\left(f_{\kv}\right) & =2\left|\set{j\le E}{\frac{\kappa_{j}}{2}>\frac{\pi}{2}}\right|=2\boldsymbol{i}\left(\kv\right).\label{eq: mandarin Neumann count anti sym}
\end{align}
 As in the proof of Lemma \ref{lem: stowers explicit} (\ref{enu: stowers indicies}),
we define a ``good'' set:
\[
M:=\set{\kv\in\mgen}{\forall t\in\opcl{0,1}\,\,t\kv\notin\boldsymbol{B}\cup\msing},
\]
such that the spectral position for $\kv\in M$ is given by
\begin{align*}
N\left(\kv\right)= & \left|\set{t\in\opcl{0,1}}{F_{s}\left(t\kv\right)=0}\right|+\left|\set{t\in\opcl{0,1}}{F_{a}\left(t\kv\right)=0}\right|.
\end{align*}
For $\kv\in M$ and $t\in\left(0,1\right)$ the function $t\mapsto F_{a}\left(t\kv\right)$
is continuous, and decreasing from infinity and so
\[
\left|\set{t\in\opcl{0,1}}{F_{a}\left(t\kv\right)=0}\right|=\begin{cases}
1 & F_{a}\left(\kv\right)\le0\\
0 & F_{a}\left(\kv\right)>0
\end{cases}=C\left(\kv\right)
\]
For $\kv\in M$, the function $t\mapsto F_{s}\left(t\kv\right)$ is
increasing and has interlacing poles and zeros as in the proof of
Lemma \ref{prop: stower Neumann surplus} and using the same argument
we get:
\begin{align*}
\left|\set{t\in\opcl{0,1}}{F_{s}\left(t\kv\right)=0}\right| & =\left|\set{j\le E}{\tan\left(\frac{\kappa_{j}}{2}\right)<0}\right|-1+\begin{cases}
1 & F_{s}\left(\kv\right)\ge0\\
0 & F_{s}\left(\kv\right)<0
\end{cases}\\
= & \boldsymbol{i}\left(\kv\right)-1+C\left(T\left(\kv\right)\right),
\end{align*}
so that for any $\kv\in M$,
\[
N\left(\kv\right)=\boldsymbol{i}\left(\kv\right)-1+C\left(\kv\right)+C\left(T\left(\kv\right)\right)
\]
Observe that $C\left(\kv\right)\equiv1$ on $\mgen_{a}$ and so $C\left(T\left(\kv\right)\right)\equiv1$
on $\mgen_{s}$. Subtracting $N\left(\kv\right)$ from (\ref{eq: mandarin nodal count sym})
and (\ref{eq: mandarin Neumann count sym}), for any $\kv\in M\cap\mgen_{s}$,
gives
\begin{align*}
\boldsymbol{\sigma}\left(\kv\right) & =\boldsymbol{i}\left(\kv\right)-C\left(\kv\right),\,\,\,\text{and}\\
\boldsymbol{\omega}\left(\kv\right) & =E-\boldsymbol{i}\left(\kv\right)-C\left(\kv\right).
\end{align*}
Similarly for $\kv\in M\cap\mgen_{a}$,
\begin{align*}
\boldsymbol{\sigma}\left(\kv\right) & =E-\boldsymbol{i}\left(\kv\right)-C\left(T\left(\kv\right)\right)=\boldsymbol{i}\left(T\left(\kv\right)\right)-C\left(T\left(\kv\right)\right)=\boldsymbol{\sigma}\left(T\left(\kv\right)\right),\\
\boldsymbol{\omega}\left(\kv\right) & =\boldsymbol{i}\left(\kv\right)-C\left(T\left(\kv\right)\right)=E-\boldsymbol{i}\left(T\left(\kv\right)\right)-C\left(T\left(\kv\right)\right)=\boldsymbol{\omega}\left(T\left(\kv\right)\right).
\end{align*}
The same argument as in the proof of Lemma \ref{prop: stower Neumann surplus}
shows that $M$ is dense in $\mgen$ and so the above holds for all
$\kv\in\mgen$.
\end{proof}

\subsection{Trees}

As a particular case of Proposition \ref{prop: stower Neumann surplus},
the Neumann surplus of a star graph is bounded and gets all integer
values between $-1$ to $1-\left|\partial\Gamma\right|$. This result
can be generalized to tree graphs using the following lemma:
\begin{lem}
\label{lem: tree + star} Let $\Gamma^{\left(1\right)}$ and $\Gamma^{\left(2\right)}$
be tree graphs with edge lengths $\lv^{\left(1\right)}$ and $\lv^{\left(2\right)}$,
and let $\Gamma$ be the tree graph obtained by gluing the two graphs
at two boundary vertices $v_{1}\in\partial\Gamma^{\left(1\right)}$
and $v_{2}\in\partial\Gamma^{\left(2\right)}$. Let $f^{\left(1\right)}$
and $f^{\left(2\right)}$ be generic eigenfunctions of $\Gamma^{\left(1\right)}$
and $\Gamma^{\left(2\right)}$ with the same eigenvalue $k$ and define
the function $f$ on $\Gamma$ by:
\begin{align*}
f|_{\Gamma^{\left(1\right)}} & =\frac{1}{f^{\left(1\right)}\left(v_{1}\right)}f^{\left(1\right)},\\
f|_{\Gamma^{\left(2\right)}} & =\frac{1}{f^{\left(2\right)}\left(v_{2}\right)}f^{\left(2\right)}.
\end{align*}
Then $f$ is a generic eigenfunction of $\Gamma$ with Neumann surplus:
\[
\omega\left(f\right)=\omega\left(f|_{\Gamma^{\left(1\right)}}\right)+\omega\left(f|_{\Gamma^{\left(2\right)}}\right)+1.
\]
\end{lem}

\begin{proof}
First we show that $f$ is an eigenfunction of $\Gamma$. Let $x_{0}$
be the point in $\Gamma$ which is identified with both $v_{1}$ and
$v_{2}$. As $\deg{x_{0}}=2$ we will consider it as an interior point
$x_{0}\in\Gamma\setminus\dg$ and denote the edge containing $x_{0}$
by $e_{0}$. By construction, the Neumann vertex conditions are satisfied
on all vertices of $\Gamma$, and the restrictions of $f$ on every
edge $e\ne e_{0}$ satisfies the ODE, $f|_{e}''=-k^{2}f|_{e}$. It
is also clear that the latter ODE holds for $f$ restricted to $e_{0}\setminus\left\{ x_{0}\right\} $,
and we only need to show that both $f$ and $f'$ are continuous at
$x_{0}$ in order to conclude that $f|_{e_{0}}''=-k^{2}f|_{e_{0}}$.
The normalization of the two parts of $f$ was chosen such that, $f|_{\Gamma^{\left(1\right)}}\left(x_{0}\right)=\frac{1}{f^{\left(1\right)}\left(v_{1}\right)}f^{\left(1\right)}\left(v_{1}\right)=1$,
and in the same way, $f|_{\Gamma^{\left(2\right)}}\left(x_{0}\right)=1$.
The two directional derivatives at $x_{0}$ are
\begin{align*}
f|'_{\Gamma^{\left(1\right)}}\left(x_{0}\right) & \propto\partial_{e_{1}}f^{\left(1\right)}\left(v_{1}\right)=0,\\
f|'_{\Gamma^{\left(2\right)}}\left(x_{0}\right) & \propto\partial_{e_{2}}f^{\left(2\right)}\left(v_{2}\right)=0.
\end{align*}
This proves that $f$ is an eigenfunction of $\Gamma$. By its construction,
it satisfies both conditions (\ref{enu: def-generic-non-zero-value-at-vertices})
and (\ref{enu:None-of-theenu: def-generic-non-zero-derivative-at-vertices})
of the genericity in Definition \ref{def: generic_eigenfunction},
and is therefore generic (see Remark \ref{rem: generic for trees}).
Since $f$ is generic, we may consider its nodal count $\phi\left(f\right)$
and Neumann count $\xi\left(f\right)$. By construction, $\phi\left(f\right)=\phi\left(f^{\left(1\right)}\right)+\phi\left(f^{\left(2\right)}\right)$
and since $x_{0}$ is an interior point on which $f'\left(x_{0}\right)=0$
then $\xi\left(f\right)=\xi\left(f^{\left(1\right)}\right)+\xi\left(f^{\left(2\right)}\right)+1$.
Since $\Gamma,\,\Gamma^{\left(1\right)}$ and $\Gamma^{\left(2\right)}$
are trees then $\sigma\left(f\right),\,\sigma\left(f^{\left(1\right)}\right)$
and $\sigma\left(f^{\left(2\right)}\right)$ are zero (see (\ref{eq: nodal surplus bounds})).
Therefore,
\begin{align*}
\omega\left(f^{\left(1\right)}\right) & =\omega\left(f^{\left(1\right)}\right)-\sigma\left(f^{\left(1\right)}\right)=\xi\left(f^{\left(1\right)}\right)-\phi\left(f^{\left(1\right)}\right),
\end{align*}
and in the same way
\begin{align*}
\omega\left(f^{\left(2\right)}\right) & =\xi\left(f^{\left(2\right)}\right)-\phi\left(f^{\left(2\right)}\right),\,\text{ and}\\
\omega\left(f\right) & =\xi\left(f\right)-\phi\left(f\right).
\end{align*}
Therefore,
\begin{align*}
\omega\left(f\right)= & \xi\left(f\right)-\phi\left(f\right)\\
= & \left(\xi\left(f^{\left(1\right)}\right)-\phi\left(f^{\left(1\right)}\right)\right)+\left(\xi\left(f^{\left(2\right)}\right)-\phi\left(f^{\left(2\right)}\right)\right)+1\\
= & \omega\left(f^{\left(1\right)}\right)+\omega\left(f^{\left(2\right)}\right)+1.
\end{align*}
\end{proof}
\begin{prop}
\label{prop: tree gets all surplus difference} If $\Gamma$ is a
tree graph then its Neumann surplus is bounded by
\[
-\left|\dg\right|+1\le\omega\le-1.
\]
Moreover, if $\Gamma$ has rationally independent edge lengths then
its Neumann surplus distribution satisfies:
\[
\forall j\in\left\{ -\left|\dg\right|+1,-\left|\dg\right|+2...,-1\right\} \qquad P\left(\omega=j\right)>0.
\]
\end{prop}

\begin{proof}
Theorem \ref{thm:Neumann_surplus_main} (\ref{enu:thm-Neumann_surplus_bounds})
ensures that the Neumann surplus of a tree is bounded between $-\left|\partial\Gamma\right|+1$
and $-1$. According to Theorem \ref{thm:Neumann_surplus_main} (\ref{enu:thm-density_of_Neumann_surplus-b}),
Theorem \ref{thm: equidistribution_by_BG_measure} and Lemma \ref{lem: functions_on_secular_manifold_existence_and_symmetry},
given a fixed value $j$, if there exist some $\lv\in\R_{+}^{E}$
and some generic eigenfunction $f$ of $\Gamma_{\lv}$ such that $\omega\left(f\right)=j$,
then for any rationally independent $\lv$ the Neumann surplus distribution
has $P\left(\omega=j\right)>0$.

We will prove that for any tree $\Gamma$ and any $j\in\left\{ -\left|\partial\Gamma\right|+1,...-1\right\} $,
there exist some $\lv\in\R_{+}^{E}$ and some generic eigenfunction
$f$ of $\Gamma_{\lv}$ for which $\omega\left(f\right)=j$. The proof
is done by induction on the number of interior vertices $V_{in}:=\left|\V\setminus\dg\right|$,
where the case $V_{in}=1$ is a star graph, for which the statement
is true by Lemma \ref{prop: stower Neumann surplus}.

Let $m>1$ and assume that the statement holds for any tree with $V_{in}<m$.
Let $\Gamma$ be a tree with $V_{in}=m$, then there is an edge $e_{0}$
which is not a tail. For $x_{0}\in e_{0}$ which is an interior point,
$\Gamma\setminus\left\{ x_{0}\right\} $ has two connected components,
$\Gamma^{\left(1\right)}$ and $\Gamma^{\left(2\right)}$, which are
both trees. Clearly, both $\Gamma^{\left(1\right)}$ and $\Gamma^{\left(2\right)}$
have less then $m$ interior vertices. By the assumption, the statement
holds for both $\Gamma^{\left(1\right)}$ and $\Gamma^{\left(2\right)}$.
That is, for each $i\in\left\{ 1,2\right\} $, for any $j_{i}\in\left\{ -\left|\partial\Gamma_{i}\right|+1,...-1\right\} $,
there exist $\lv^{\left(i\right)}$ and a generic eigenfunction $f^{\left(i\right)}$
of $\Gamma_{\lv^{\left(i\right)}}^{\left(i\right)}$ with Neumann
surplus $\omega\left(f^{\left(i\right)}\right)=j_{i}$. We may scale
$\lv^{\left(i\right)}$ such that $f^{\left(i\right)}$ has eigenvalue
$k=1$. Consider $\Gamma$ with edge lengths according to the gluing
of $\Gamma_{\lv^{\left(1\right)}}^{\left(1\right)}$ and $\Gamma_{\lv^{\left(2\right)}}^{\left(2\right)}$
at $x_{0}$. By Lemma \ref{lem: tree + star}, $\Gamma$ has a generic
eigenfunction $f$ with
\[
\omega\left(f\right)=j_{1}+j_{2}+1.
\]
By construction $\left|\partial\Gamma\right|=\left|\partial\Gamma_{1}\right|+\left|\partial\Gamma_{2}\right|-2$
and so $j:=j_{1}+j_{2}+1$ ranges over all values between
\[
-\left|\partial\Gamma\right|+1=\left(-\left|\partial\Gamma_{1}\right|+1\right)+\left(-\left|\partial\Gamma_{2}\right|+1\right)+1\le j\le-1-1+1=-1.
\]
By induction, the statement is true for any choice of $V_{in}$. This
proves the Lemma.
\end{proof}

\section{Proofs of Lemma \ref{lem: Canonical_eigenfunctions_and_Secular_mfld}
and Lemma \ref{lem: Inversion-properties} \label{sec: appendix=00005C-proof-of-two-lemmas}}
\begin{proof}
[Proof of Lemma \ref{lem: Canonical_eigenfunctions_and_Secular_mfld}
]

We start by providing a basic tool in the spectral analysis of metric
graphs. Let $\Gamma$ be a standard graph with $E:=\text{\ensuremath{\left|\E\right|}}$
edges. For every $\kv\in\torus$ we abuse notation and denote by $\rme^{\rmi\kv}\in U(2E)$
the following unitary diagonal matrix
\[
\rme^{\rmi\kv}:=\mathrm{diag}\left(\rme^{\rmi\kappa_{1}},~\rme^{\rmi\kappa_{1}},~\rme^{\rmi\kappa_{2}},~\rme^{\rmi\kappa_{2}},~\ldots,e^{\rmi\kappa_{E}},~e^{\rmi\kappa_{E}}\right),
\]
such that every exponent $\rme^{\rmi\kappa_{j}}$ appears twice. Given
a choice of edge lengths $\lv$ and an eigenvalue $k$ it is convenient
to denote the matrix $\rme^{\rmi\kv}$ for $\kv=\modp{k\lv}$ by $\rme^{\rmi k\lv}$.
There exists an orthogonal matrix $S\in O(2E)$, uniquely determined
by the connectivity of $\Gamma$, which is independent of the edge
lengths and has the following property. A positive $k>0$ is an eigenvalue
of $\Gamma$ if and only if $\ker(\Id-{\rm e}^{\rmi k\lv}S)$ is nontrivial
\cite[Section 3.2]{AloBanBer_cmp18}. This relation is a consequence
of an isomorphism between $\ker(\Id-{\rm e}^{\rmi k\lv}S)$ and the
$k$-eigenspace of $\Gamma$. Explicitly, this isomorphism sends $\boldsymbol{a}\in\ker(\Id-{\rm e}^{\rmi k\lv}S)$
to an eigenfunction $f$ of the eigenvalue $k$ such that the restriction
to every edge $e\in\E$, is given by
\begin{equation}
\left.f\right|_{e}\left(x\right)=a_{e}\rme^{\rmi kx}e^{-ikl_{e}}+a_{\hat{e}}\rme^{-\rmi kx},\label{eq: parametrization-1}
\end{equation}
where $x\in[0,l_{e}]$ is an arc-length parametrization of the edge
$e\in\E$, and $a_{e},a_{\hat{e}}$ are the two entries of the vector
$\boldsymbol{a}$ which correspond to edge $e$ (according to the
indexing of $\rme^{\rmi\kv}$ which was introduced above). As an implication
we get for each $\lv$ an isomorphism between the $k$-eigenspace
of $\Gamma_{\lv}$ and the $1$-eigenspace of $\Gamma_{\kv}$, where
$\kv=\modp{k\lv}$. This, together with (\ref{eq: def-of-sigma-and-sigma-reg})
proves part (\ref{enu: lem-Canonical_eigenfunctions_and_Secular_mfld-1})
of the Lemma.

The above argument also implies that $\kv\in\mreg$ if and only if
$\dim\ker(1-e^{i\kv}S)=1$. For any such $\kv$ value, the adjugate
matrix $\mathrm{adj}(\Id-{\rm e}^{\rmi\kv}S)$ is a rank one matrix.
Hence, there exists a non-trivial $\boldsymbol{a}\in\ker(\Id-{\rm e}^{\rmi\kv}S)$
such that $\mathrm{adj}(\Id-{\rm e}^{\rmi\kv}S)=\boldsymbol{a}\boldsymbol{a}^{*}$.
We use this $\bs a$ to introduce a function $f_{\kv}$ by (\ref{eq: parametrization-1}).
This provides a family of canonical eigenfunctions $\{f_{\kv}\}_{\kv\in\mreg}$,
as in the statement of the lemma. We should note that these functions
are not uniquely determined. This is since the vectors $\bs a$ above
are determined only up to a multiplication by some unitary factor
$c\in U(1)$; indeed $\mathrm{adj}(\Id-{\rm e}^{\rmi\kv}S)=\boldsymbol{a}\boldsymbol{a}^{*}$
is invariant with respect to such a multiplication of $\bs a$, and
the factor might depend on $\kv$. Next, we show that this family
of canonical eigenfunctions $\{f_{\kv}\}_{\kv\in\mreg}$ indeed satisfy
the requirements in the lemma.

First, note that the arguments above already prove parts (\ref{enu: lem-Canonical_eigenfunctions_and_Secular_mfld-2a})
and (\ref{enu: lem-Canonical_eigenfunctions_and_Secular_mfld-3})
of the lemma (the $c\in\C$ factor appears in part (\ref{enu: lem-Canonical_eigenfunctions_and_Secular_mfld-3})
since the canonical eigenfunction, $f_{\kv}$, is not necessarily
real). We proceed to prove part (\ref{enu: lem-Canonical_eigenfunctions_and_Secular_mfld-2b})
of the lemma.

Let $v\in\V$ and $e\in\E_{v}$ with arc-length parametrization $x\in[0,l_{e}]$
such that $v$ is at $x=0$. By (\ref{eq: parametrization-1}) we
get that for each $\kv\in\mreg$,
\begin{align}
f_{\kv}\left(v\right) & =f_{\kv}|_{e}\left(0\right)=a_{e}\rme^{-\rmi\kappa_{e}}+a_{\hat{e}}\label{eq: f(v)}\\
\partial_{e}f_{\kv}\left(v\right) & =f_{\kv}|_{e}'\left(0\right)=\rmi\left(a_{e}\rme^{-\rmi\kappa_{e}}-a_{\hat{e}}\right).\label{eq: df(v)}
\end{align}
Observe that each term $f_{\kv}(v)$ or $\partial_{e}f_{\kv}(v)$
is linear in the $\boldsymbol{a}$ entries with coefficients which
are taken from $(\rme^{-\rmi\kappa_{1}},...,\rme^{-\rmi\kappa_{E}})$.
Therefore, each product of the form $f_{\kv}\left(u\right)\overline{f_{\kv}\left(v\right)}$
or $f_{\kv}\left(u\right)\overline{\partial_{e}f_{\kv}\left(v\right)}$
is linear in the $\boldsymbol{a}\boldsymbol{a}^{*}$ entries with
coefficients which are trigonometric polynomials (in $(\kappa_{1},...,\kappa_{E})$).
But, the entries of $\boldsymbol{a}\boldsymbol{a}^{*}=\mathrm{adj}(\Id-{\rm e}^{\rmi\kv}S)$
are minors of $\Id-{\rm e}^{\rmi\kv}S$ and so themselves trigonometric
polynomials in $(\kappa_{1},...,\kappa_{E})$. Hence, $f_{\kv}\left(u\right)\overline{f_{\kv}\left(v\right)}$
is a trigonometric polynomial in $(\kappa_{1},...,\kappa_{E})$, which
we denote by $p_{u,v}(\kv)$; and $f_{\kv}\left(u\right)\overline{\partial_{e}f_{\kv}\left(v\right)}$
is a trigonometric polynomial in $(\kappa_{1},...,\kappa_{E})$, which
we denote by $q_{u,v,e}(\kv)$. This proves the relations (\ref{eq: p_u_v_trig_poly})
and (\ref{eq: q_u_v_e_trig_poly}) in the lemma and it is left to
show that $p_{u,v}$ and $q_{u,v,e}$ are actually real trigonometric
polynomials. To do this, we redefine $p_{u,v}$ and $q_{u,v,e}$ to
be equal to their real parts and argue that (\ref{eq: p_u_v_trig_poly})
and (\ref{eq: q_u_v_e_trig_poly}) are still valid after such a modification.
Indeed, all products $f_{\kv}\left(u\right)\overline{f_{\kv}\left(v\right)}$
and $f_{\kv}\left(u\right)\overline{\partial_{e}f_{\kv}\left(v\right)}$
are real since $f_{\kv}$ is an eigenfunction of $\Gamma$ corresponding
a simple eigenvalue and as such it is real up to a global multiplicative
factor.
\end{proof}
\begin{proof}
[Proof of Lemma \ref{lem: Inversion-properties} ]

We begin by referring to the proof of Lemma \ref{lem: Canonical_eigenfunctions_and_Secular_mfld}
above where it was argued that
\[
\kv\in\Sigma\Leftrightarrow\dim\ker\left(\Id-{\rm e}^{\rmi\kv}S\right)>0
\]
 and
\[
\kv\in\mreg\Leftrightarrow\dim\ker\left(\Id-{\rm e}^{\rmi\kv}S\right)=1.
\]
Combining this with the relation
\begin{equation}
\forall\kv\in\torus,\quad\Id-{\rm e}^{\rmi\I\left(\kv\right)}S=\overline{\left(\Id-{\rm e}^{\rmi\kv}S\right)},\label{eq: inversion_symmetry_of_matrices_on_torus}
\end{equation}
gives
\[
\kv\in\Sigma\Leftrightarrow\I(\kv)\in\Sigma\quad\textrm{and }\quad\kv\in\mreg\Leftrightarrow\I(\kv)\in\mreg,
\]
and proves that $\Sigma$ and $\mreg$ are each invariant under $\I$.
We may further deduce that $\mathrm{adj}\left(\Id-{\rm e}^{\rmi\I\left(\kv\right)}S\right)=\overline{\mathrm{adj}\left(\Id-{\rm e}^{\rmi\kv}S\right)}$
and so if $\boldsymbol{a}$ is a vector for which $\boldsymbol{a}\boldsymbol{a}^{*}=\mathrm{adj}\left(\Id-{\rm e}^{\rmi\kv}S\right)$
then $\overline{\boldsymbol{a}}$ is such that $\overline{\boldsymbol{a}}\overline{\boldsymbol{a}}^{*}=\mathrm{adj}\left(\Id-{\rm e}^{\rmi\I\left(\kv\right)}S\right)$.
These vectors define $f_{\kv}$ and $f_{\I\left(\kv\right)}$ by (\ref{eq: f(v)})
and (\ref{eq: df(v)}) and so we get
\begin{align}
\forall v\in\V,\quad\quad f_{\I\left(\kv\right)}\left(v\right) & =\overline{a_{e}}\rme^{-\rmi\left[\I\left(\kv\right)\right]_{e}}+\overline{a_{\hat{e}}}=\overline{a_{e}\rme^{-\rmi\kappa_{e}}+a_{\hat{e}}}\quad\quad\text{and}\label{eq: f(v)-1}\\
\forall v\in\V,~\forall e\in\E_{v}\quad\quad\partial_{e}f_{\I\left(\kv\right)}\left(v\right) & =\rmi\left(\overline{a_{e}}\rme^{-\rmi\left[\I\left(\kv\right)\right]_{e}}-\overline{a_{\hat{e}}}\right)=-\overline{\rmi\left(a_{e}\rme^{-\rmi\kappa_{e}}-a_{\hat{e}}\right)}.\label{eq: df(v)-1}
\end{align}
Comparing the RHS above with (\ref{eq: f(v)}) and (\ref{eq: df(v)})
shows
\begin{equation}
\forall v\in\V,~\forall e\in\E_{v}\quad\quad f_{\I\left(\kv\right)}\left(v\right)=c\overline{f_{\kv}\left(v\right)}\quad\text{and}\quad\partial_{e}f_{\I\left(\kv\right)}\left(v\right)=-\overline{c}\overline{\partial_{e}f_{\kv}\left(v\right)},\label{eq: canonical-func-after-inversion}
\end{equation}
where $c\in U(1)$ is a multiplicative factor which expresses a degree
of freedom in determining $\bs a$ (alternatively $f_{\kv}$) and
$c$ might depend on $\kv$; see also in proof of Lemma \ref{lem: Canonical_eigenfunctions_and_Secular_mfld}.
From (\ref{eq: canonical-func-after-inversion}) together with the
definition of $\mgen$, (\ref{eq: definition generic}) we get that
$\mgen$is invariant under $\I$, which finishes the proof of part
(\ref{enu: lem-Inversion-properties-1}) of the Lemma.

To continue proving the second part of the lemma, we note that (\ref{eq: canonical-func-after-inversion})
implies in particular that all products $f_{\kv}\left(u\right)\overline{f_{\kv}\left(v\right)}$
are $\I$ invariant and all products $f_{\kv}\left(u\right)\overline{\partial_{e}f_{\kv}\left(v\right)}$
are $\I$ anti-symmetric. As a conclusion we get that the real trigonometric
polynomials $p_{u,v}$ and $q_{u,v,e}$ that are defined in the proof
of Lemma \ref{lem: Canonical_eigenfunctions_and_Secular_mfld} satisfy
\begin{align}
\forall\kv\in\mreg\,\,\,\,p_{u,v}\left(\I\left(\kv\right)\right) & =p_{u,v}\left(\kv\right),\,\text{ and}\nonumber \\
\forall\kv\in\mreg\,\,\,\,q_{u,v,e}\left(\I\left(\kv\right)\right) & =-q_{u,v,e}\left(\kv\right).\label{eq: inversion_of_trig_polys}
\end{align}
Note that these are almost the required relations (\ref{enu: lem-Inversion-properties-1}),(\ref{enu: lem-Inversion-properties-2}).
But (\ref{eq: inversion_of_trig_polys}) holds for $\kv\in\mreg$,
whereas (\ref{enu: lem-Inversion-properties-1}),(\ref{enu: lem-Inversion-properties-2})
are stated for all $\kv\in\torus$. Let us abuse notation once again
and redefine $p_{u,v}$ as $\frac{p_{u,v}+p_{u,v}\circ\I}{2}$ and
$q_{u,v,e}$ as $\frac{q_{u,v,e}-q_{u,v,e}\circ\I}{2}$. This ensures
that $p_{u,v}$ and $q_{u,v,e}$ are $\I$ symmetric and anti-symmetric
respectively, while being real trigonometric polynomials whose restrictions
are equal to the products $f_{\kv}\left(u\right)\overline{f_{\kv}\left(v\right)}$
and $f_{\kv}\left(u\right)\overline{\partial_{e}f_{\kv}\left(u\right)}$
for any $\kv\in\Sigma^{reg}$. To show that $\mgen$ is $\I$ invariant,
recall that $\mgen$ is the set of $\kv\in\mreg$ for which $f_{\kv}$
is generic. Definition \ref{def: generic_eigenfunction} for genericity
has three conditions, where condition (\ref{enu: def-generic-simple-evalue})
holds if and only if $\kv\in\mreg$. The other two conditions, (\ref{enu: def-generic-non-zero-value-at-vertices})
and (\ref{enu:None-of-theenu: def-generic-non-zero-derivative-at-vertices}),
hold if and only if $p_{u,u}\left(\kv\right)\ne0$ and $q_{v,v,e}\left(\kv\right)\ne0$
for all $u\in\V,\,v\in\V\setminus\partial\Gamma$ and $e\in\E_{v}$.
As $\mreg$ is $\I$ invariant and $p_{u,u}$ and $q_{v,v,e}$ are
$\I$ symmetric\textbackslash anti-symmetric, then $f_{\I\left(\kv\right)}$
is generic if and only if $f_{\kv}$ is and therefore $\mgen$ is
$\I$ invariant.

The third part of our Lemma is stated and proved in \cite[lem. 3.13]{AloBanBer_cmp18}.
\end{proof}

\section{\label{sec:Appendix nodal domains} An analogue of Proposition \ref{prop:local_observables_bounds}
for nodal domains}

This appendix provides bounds on the wavelength capacity and the Neumann
count of nodal domains.
\begin{prop}
\label{prop: bounds stars Dirichlet} Let $\Gamma$ be a standard
graph with minimal edge length $L_{\mathrm{min}}$. Let $k^{2}$ be
an eigenvalue of $\Gamma$ which satisfies $k>\frac{\pi}{L_{\mathrm{min}}}$.
Let $\Omega_{D}$ be a nodal domain of a generic eigenfunction $f$
which corresponds to $k$. Denote the total length of $\Omega_{D}$
by $\left|\Omega_{D}\right|$, its wavelength capacity by $\rho(\Omega_{D}):=\frac{\left|\Omega_{D}\right|k}{\pi}$,
and its Neumann count by $\xi(\Omega_{D}):=\xi(f|_{\Omega_{D}})$.
The following bounds hold
\begin{align}
1\le & \xi(\Omega_{D})\le\left|\partial\Omega_{D}\right|-1\label{eq: bounds_Neumann_count}\\
1\le\frac{1}{2}\left(\xi(\Omega_{D})+1\right)\leq & \rho(\Omega_{D})\le\frac{1}{2}\left(\xi(\Omega_{D})+\left|\d\Omega_{D}\right|-1\right)\leq\left|\partial\Omega_{D}\right|-1\label{eq: bounds_rho}
\end{align}
\end{prop}

\begin{rem}
The spectral position of a Neumann domain is equal to its nodal count.
Hence, the bounds (\ref{eq: Spectral_Position_bounds}) in Proposition
\ref{prop:local_observables_bounds} may be perceived as the bounds
on the nodal count of a Neumann domain. Therefore, we consider the
bounds in (\ref{eq: bounds_Neumann_count}) of the proposition above
as the analogous bounds. Furthermore, there is no interest in bounds
for the spectral position of a nodal domain, as it trivially equals
to $1$ (see Section \ref{subsec:Spectral-position-and-wavelength-capacity}).
\end{rem}

\begin{proof}
We start by noting that the condition $k>\frac{\pi}{L_{\mathrm{min}}}$
guarantees that $\Omega_{D}$ is either a star graph or an interval
(see also Lemma \ref{lem:spectral-position-equals-nodal-count}).
If $\Omega_{D}$ is an interval the bounds follow trivially as $\xi(\Omega_{D})=1$
and $\rho(\Omega_{D})=1$. We proceed by assuming that $\Omega_{D}$
is a star graph. From here, the main argument in the proof is a map
from star graphs which are nodal domains (such as $\Omega_{D}$) to
star graphs which are Neumann domains, as described next. Denote the
edge lengths of $\Omega_{D}$ by $\{l_{j}\}_{j=1}^{\left|\partial\Omega_{D}\right|}$.
We may write $f|_{\Omega_{D}}$ on every edge $e_{j}$ of $\Omega_{D}$
as
\begin{equation}
f|_{e_{j}}\left(x\right)=A_{j}\sin\left(k(x-l_{j})\right)\,\,\,x\in\left[0,l_{j}\right],\label{eq: function on nodal domain}
\end{equation}
where $x=0$ at the central vertex and $x=l_{j}$ at the boundary
vertex. The absence of nodal points in the interior of $\Omega_{D}$
together with the genericity of $f$ imply that $kl_{j}\in\left(0,\frac{\pi}{2}\right)\cup\left(\frac{\pi}{2},\pi\right)$
for each $j$. We use this to construct an auxiliary star graph, $\Omega_{N}$,
which has the same number of edges, $\left|\partial\Omega_{N}\right|=\left|\partial\Omega_{D}\right|$
and its edge lengths are given by
\[
\forall1\leq j\leq\left|\partial\Omega_{N}\right|,\quad\tilde{l}_{j}=\begin{cases}
l_{j}-\frac{\pi}{2k} & kl_{j}\in\left(\frac{\pi}{2},\pi\right)\\
l_{j}+\frac{\pi}{2k} & kl_{j}\in\left(0,\frac{\pi}{2}\right).
\end{cases}
\]
Define a function $\tilde{f}$ on $\Omega_{N}$ by
\begin{equation}
\tilde{f}|_{\tilde{e}_{j}}\left(x\right)=A_{j}\sin\left(k(x-l_{j})\right)\,\,\,x\in\left[0,\tilde{l}_{j}\right],\label{eq: function on Neumann in nodal}
\end{equation}
for any edge $\tilde{e}_{j}$ of $\Omega_{N}$, where $A_{j}$ and
$l_{j}$ are the same as in (\ref{eq: function on nodal domain}).
This construction guarantees $\tilde{f}'|_{e_{j}}\left(\tilde{l}_{j}\right)=A_{j}\sin\left(\pm\frac{\pi}{2}\right)=0$
and so $\tilde{f}$ satisfies Neumann conditions at the boundary vertices
of $\Omega_{N}$. In addition, $\tilde{f}$ and $f$ share the same
value and derivatives at the corresponding central vertex (of $\Omega_{N}$
and of $\Omega_{D}$). Hence, $\tilde{f}$ satisfies Neumann vertex
conditions at the central vertex of $\Omega_{N}$. We conclude that
$\tilde{f}$ is an eigenfunction of $\Omega_{N}$ with eigenvalue
$k$.

Since $f$ and $\tilde{f}$ share the same value and derivatives at
the central vertex, and $f$ is generic, then $\tilde{f}$ satisfies
both conditions (\ref{enu: def-generic-non-zero-value-at-vertices})
and (\ref{enu:None-of-theenu: def-generic-non-zero-derivative-at-vertices})
of the genericity in Definition \ref{def: generic_eigenfunction}.
In addition, $k$ is a simple eigenvalue of $\Omega_{N}$ by an argument
similar to the one which was given in the proof of Lemma \ref{lem:spectral-position-equals-nodal-count}
(this argument is based on \cite[Corollary 3.1.9]{BerKuc_graphs}
and on $\Omega_{N}$ being a tree). Therefore, $\tilde{f}$ is generic.
By the construction, $\tilde{f}$ has no interior Neumann points,
so that $\Omega_{N}$ is a single Neumann domain of the generic eigenfunction
$\tilde{f}$. In particular, the bounds on $N(\Omega_{N})$ and $\rho(\Omega_{N})$
in Proposition \ref{prop:local_observables_bounds} apply to it. To
finish the proof we just need to relate $N(\Omega_{N})$ and $\rho(\Omega_{N})$
to $\xi(\Omega_{D})$ and $\rho(\Omega_{D})$ and apply those bounds.
The needed relations follow from a simple calculations based on (\ref{eq: function on nodal domain})
and (\ref{eq: function on Neumann in nodal}):
\begin{align*}
\phi\left(\left.\tilde{f}\right|_{\Omega_{N}}\right)+\xi(\Omega_{D}) & =\left|\d\Omega_{D}\right|,
\end{align*}
and
\begin{align*}
\rho\left(\Omega_{N}\right) & =\frac{1}{\pi}\sum_{j=1}^{\left|\partial\Omega_{N}\right|}\tilde{l}_{j}=\frac{1}{\pi}\left(\sum_{j=1}^{\left|\partial\Omega_{D}\right|}\left(l_{j}+\frac{\pi}{2}\right)-\pi\xi\left(f|_{\Omega_{D}}\right)\right)=\rho\left(\Omega_{D}\right)+\frac{\left|\partial\Omega_{D}\right|}{2}-\xi(\Omega_{D}).
\end{align*}

Using the relations above with $\phi(\tilde{f}|_{\Omega_{N}})=N(\Omega_{N})$
(see (\ref{eq:spectral-position-equals-nodal-count})) and applying
the bounds on $N(\Omega_{N})$ and $\rho(\Omega_{N})$ in Proposition
\ref{prop:local_observables_bounds} yields all the desired bounds
in (\ref{eq: bounds_Neumann_count}) and (\ref{eq: bounds_rho}).
\end{proof}
\bibliographystyle{siam}
\bibliography{GlobalBib}

\end{document}